\documentclass[11pt,a4paper]{article}

\usepackage{graphicx,tabularx}
\usepackage{epic,fancybox}
\usepackage[T1]{fontenc}
\usepackage[ansinew]{inputenc}
\usepackage{bbm}				

\usepackage{amssymb,amsmath,amsthm}

\usepackage[in]{fullpage}

\usepackage{tikz}
\usepackage{datetime}

\newtheorem{definition}{Definition}[section]
\newtheorem{lemma}[definition]{Lemma}
\newtheorem{cor}[definition]{Corollary}

\newtheorem{thm}[definition]{Theorem}
\newtheorem{alg}{Algorithm}
\newtheorem{prop}[definition]{Proposition}
\newtheorem{claim}[definition]{Claim}
\newtheorem{assumption}[definition]{Assumption}

\newcommand{\np}{\textrm{\sf NP}}

\newcommand{\dec}{\textrm{\sf Dec}}
\newcommand{\opt}{\textrm{\sf Opt}}

\newcommand{\svp}{\textrm{\sc Svp}}

\newcommand{\cvp}{\textrm{\sc Cvp}}
\newcommand{\optcvp}{\textrm{\sc OptCvp}}

\newcommand{\lmp}{\textrm{\sc Lmp}}

\DeclareMathOperator{\vol}{vol}
\DeclareMathOperator{\Span}{span}

\DeclareMathOperator{\sign}{sign}
\DeclareMathOperator{\size}{size}

\DeclareMathOperator{\out}{out}

\definecolor{blau}{cmyk}{1,0.72,0,0.38}
\definecolor{grey}{rgb}{0.5,0.5,0.5}
\definecolor{orange}{cmyk}{0,0.51,1.0,0}
\definecolor{hellblau}{cmyk}{0.3,0.18,0.06,0}
\definecolor{graublau}{cmyk}{0,0,0,0.2}
\definecolor{violett}{cmyk}{0.6,1,0,0}

\author{Johannes Bl\"omer\thanks{Department of Computer Science, University of Paderborn, {\tt bloemer@upb.de}}
\and Stefanie Naewe\thanks{Department of Computer Science, University of Paderborn,
{\tt stefanie.naewe@upb.de}.
This research was partially supported by German Research Foundation (DFG), Research Training Group GK-693 of the Paderborn Institute for Scientific Computation (PaSCo) and the Heinz Nixdorf Institute.}}
\title{Solving the Closest Vector Problem with respect to $\ell_p$ Norms}
\date{Tuesday 12$^{\mbox{th}}$ July, 2011}

\begin{document}

\maketitle

\begin{abstract}
We present deterministic polynomially space bounded algorithms for the closest vector problem for all $\ell_p$-norms, $1 < p < \infty$, and all polyhedral norms,
in particular for the $\ell_1$-norm and the $\ell_{\infty}$-norm.
For all $\ell_p$-norms with $1 < p < \infty$ the running time of the algorithm is 
$p \cdot \log_2 ( r )^{\mathcal{O} ( 1 )} n^{(2 + o ( 1 ) )n}$, where $r$ is an upper bound on the size of the coefficients of the target vector and the lattice basis and $n$ is the dimension of the vector space.
For polyhedral norms, we obtain an algorithm with running time 
$(s \log_2 ( r ) )^{\mathcal{O} ( 1 )} n^{(2+o(1))n}$, where $r$ and $n$ are defined as above and $s$ is the number of constraints defining the polytope.
In particular, for the $\ell_1$-norm and the $\ell_{\infty}$-norm, we obtain a deterministic algorithm for the closest vector problem with running time $\log_2 ( r )^{\mathcal{O} ( 1 )} n^{(2 + o(1))n}$.\\
We achieve our results by introducing a new lattice problem, the lattice membership problem:
For a given full-dimensional bounded convex set and a given lattice, the goal is to decide whether the convex set contains a lattice vector or not.
The lattice membership problem is a generalization of the integer programming feasibility problem from polyhedra to bounded convex sets.
In this paper, we describe a deterministic algorithm for the lattice membership problem, which is a generalization of Lenstra's algorithm for integer programming.
We also describe a polynomial time reduction from the closest vector problem to the lattice membership problem.
This approach leads to a deterministic algorithm that solves the closest vector problem in polynomial space for all $\ell_p$-norms, $1 < p < \infty$, and all polyhedral norms.
\end{abstract}

\section{Introduction}

\noindent In the closest vector problem ($\cvp$), we are given a lattice $L$ and some target vector $t$ in the $\mathbbm{R}$-vector space $\Span ( L )$ spanned by the vectors in $L$. We are asked to find a vector $u \in L$, whose distance to $t$ is as small as possible.
Since this problem can be defined for any norm on $\mathbbm{R}^n$, we stated this problem without referring to a specific norm.
Often $\cvp$ as well as other lattice problems are considered with respect to the $\ell_2$-norm.
However, it is also common to consider $\cvp$ with respect to other non-Euclidean norms, for example in cryptography, see \cite{pp_nguyen_dark_side_hidden_number}, or in integer programming, see \cite{pp_lenstraoptimization}.
The most commonly used non-Euclidean norms are arbitrary $\ell_p$-norms with $1 \leq p \leq \infty$.
The $\ell_p$-norm of a vector $x \in \mathbbm{R}^n$ is defined by
$\| x \|_p = ( \sum_{i=1}^n | x_i |^p )^{1/p}$ for $1 \leq p < \infty$
and $\| x \|_{\infty} = \max \{ | x_i | | 1 \leq i \leq n \}$.
In general, a norm is defined by a convex body $\mathcal{C}$ symmetric about the origin via the function
$p_{\mathcal{C}} : \mathbbm{R}^n \to \mathbbm{R}$, $p_{\mathcal{C}} ( x ) = \inf \{ \lambda \geq 0 | x \in \lambda \cdot \mathcal{C} \}$.
If the convex body is a bounded polyhedron $P$, i.e., a polytope, then we call the corresponding norm a polyhedral norm, denoted by $\| \cdot \|_P$.
Especially, the $\ell_1$-norm and the $\ell_{\infty}$-norm are polyhedral norms.
Therefore, we distinguish in the following between $\ell_p$-norms with $1 < p < \infty$ and polyhedral norms.

\paragraph*{Algorithms for $\cvp$}
In the last 30 years, the complexity of $\cvp$ has been studied intensively.
It is known that $\cvp$ with respect to all $\ell_p$-norms is \np-hard and even hard to approximate, see \cite{bk_EmdeBoas}, \cite{pp_abbs93}, \cite{pp_dinur98approximatingcvp}, \cite{pp_dinur03}, \cite{pp_rr06}, \cite{pp_peikert07}.
Furthermore, for any $\epsilon > 0$, there is a randomized reduction from $\cvp$ with approximation factor $1 + \epsilon$ with respect to the $\ell_2$-norm  to the exact version of $\cvp$ with respect to the $\ell_p$-norm, see \cite{pp_rr06}.
This suggests that $\cvp$ with respect to the $\ell_2$-norm is easier than $\cvp$ with respect to any other $\ell_p$-norm.\\

\noindent The best polynomial time approximation algorithms are based on the LLL-algorithm and achieve single exponential approximation factors. Basically, they work for the $\ell_2$-norm, but using Hölder's inequality, we obtain results for all $\ell_p$-norms, see \cite{pp_LLL}, \cite{pp_babai}, \cite{pp_schnorr}, \cite{pp_schnorrBlockReduced}.\\

\noindent In this paper, we focus on deterministic and exact algorithms for $\cvp$ with respect to arbitrary $\ell_p$-norms.
Therefore, in the sequel we ignore all probabilistic algorithms for $\cvp$ like the results based on the AKS sampling technique.
Instead, we briefly review the existing deterministic algorithms for $\cvp$ with respect to the $\ell_2$-norm and discuss why it may be difficult or even impossible to generalize them to non-Euclidean norms. For a survey on these algorithms see \cite{pp_survey_hps11}.\\

\noindent In a breakthrough paper, Micciancio and Voulgaris describe a deterministic, single exponential time algorithm that solves $\cvp$ with respect to the $\ell_2$-norm exactly, see \cite{pp_MV10}.
It is based on the computation of Voronoi cells of a lattice.
The algorithm can be generalized to all norms whose unit ball is an ellipsoid as remarked in \cite{pp_dpv}.
Unfortunately, it seems that the $\cvp$-algorithm of \cite{pp_MV10} cannot be generalized to other norms, since then the Voronoi cell of a lattice need not be convex.
Moreover, the algorithm of Micciancio and Voulgaris requires exponential space.\\

\noindent Basically, there exist two other deterministic algorithms for $\cvp$ with respect to the $\ell_2$-norm.
Both algorithms require polynomial space, in particular, they work with numbers whose bit size is polynomially bounded in the input size.
The algorithm of Kannan \cite{pp_kannan} with its improvements by Helfrich \cite{pp_helfrich} and Hanrot and Stehlé \cite{pp_hast07} uses at most $2^{\mathcal{O} ( n )} n^{n/2} \cdot \log_2 ( r )^{\mathcal{O} ( 1 )}$ arithmetic operations, where $n$ is the rank of the lattice and $r$ is an upper bound on the coefficients used to describe the basis.
Another algorithm that solves $\cvp$ optimally is due to Blömer \cite{pp_bloemerhkz}.
It uses $n! \log_2 ( r )^{\mathcal{O} ( 1 )}$ arithmetic operations.
It may be difficult to generalize these two algorithms to non-Euclidean norms (although Kannan claims the opposite in his paper), since they both use orthogonal projections:
At some point during the algorithm, they have to work with a target vector, which is not contained in the vector space spanned by the lattice. 
In this situation, they consider the orthogonal projection of the target vector onto the subspace spanned by the lattice.
However, unlike the $\ell_2$-norm, for arbitrary $\ell_p$-norms the closest lattice vector to the target vector is not a closest lattice vector to the orthogonal projection of the target vector or vice versa.
Also, if we use norm projections as defined in \cite{pp_Mangasarian99} or \cite{pp_ls}, this is not true.
We present a counterexample for both cases in the appendix, see Section \ref{appendix_counterexample}.

\paragraph*{In this paper,}
we consider the lattice membership problem ($\lmp$), where we are given a full-dimensional bounded convex set together with a lattice and we want to decide whether the convex set contains a lattice vector.
First, we show that for all $\ell_p$-norms, $1 < p < \infty$, and all polyhedral norms, e. g. the $\ell_1$-norm and the $\ell_{\infty}$-norm, there exists a polynomial time reduction from $\cvp$ to $\lmp$.
The reduction also preserved the dimension and the rank of the input lattice.
Furthermore, we show that there exists a deterministic algorithm that solves $\lmp$ in polynomial space for all $\ell_p$-balls and polytopes.
If we consider $\ell_p$-norms, $1 < p < \infty$, we obtain an algorithm that uses 
$p \cdot \log_2 ( r )^{\mathcal{O} ( 1 )} n^{(2 + o ( 1 ) )n}$ arithmetic operations
and for all polyhedral norms an algorithm that uses 
$(s \cdot \log_2 ( r ) )^{\mathcal{O} ( 1 )} n^{(2 + o ( 1 )) n}$ arithmetic operations, 
where $s$ is the number of constraints defining the polytope.
For the $\ell_1$-norm, we have $s = 2^n$ and for the $\ell_{\infty}$-norm, we have $s = 2n$.
Hence for these norms, we obtain a deterministic polynomial space algorithm using $\log_2 ( r )^{\mathcal{O} ( 1 )} n^{(2 + o ( 1 ) )n}$ arithmetic operations.
Together with the reduction from $\cvp$ to $\lmp$, we obtain a deterministic algorithm that solves $\cvp$ for all $\ell_p$-norms and all polyhedral norms exactly in polynomial space.
To the best of our knowledge, this is the first result of this type.\\

\noindent The lattice membership problem is a generalization of the integer programming feasibility problem and our algorithm is a variant of Lenstra's algorithm for integer programming used together with a variant of the ellipsoid method, see \cite{pp_lenstraoptimization}. To guarantee that the algorithm runs in polynomial space, we use a preprocessing technique from Frank and Tardos \cite{pp_ft87} developed for Lenstra's algorithm for integer programming.
To put our results in perspective, we shortly review in the following the major results based on Lenstra's technique.

\paragraph*{Lenstra's algorithm for integer programming and related results} In 1979, Lenstra presented a polynomial time algorithm that solves the integer programming feasibility problem in fixed dimension \cite{pp_lenstraoptimization}, which was improved by Kannan in 1987 \cite{pp_kannan}.
Using a further improvement by Frank and Tardos \cite{pp_ft87}, the algorithm requires polynomial space and the number of arithmetic operations of this algorithm is $\mathcal{O} ( n^{5/2n} \log_2 ( r ) )$, where $r$ is an upper bound on the size of the polyhedron.
Hence, our result improves the running time of Lenstra's algorithm by the factor $n^{n/2}$.\\
In 2005, Heinz generalized Lenstra's algorithm to obtain an algorithm for integer optimization over quasiconvex polynomials, which was improved by Hildebrand and Köppe (see \cite{pp_heinz05}, \cite{pp_hk10}). Their results can be used to decide whether the set $\{ x \in \mathbbm{R}^n | \| x \|_p^p - \alpha < 0 \}$ contains a lattice vector, if $p$ is an even number, since for $p$ even, the function $x \mapsto \| x \|_p^p$ is a quasiconvex polynomial. In this case, we obtain an algorithm for the lattice membership problem using at most $\log_2 ( r )^{\mathcal{O} ( 1 )} p^{\mathcal{O} ( n )} n^{(2+o(1))n}$ arithmetic operations.
By comparison, the number of arithmetic operations of our algorithm depends only linearly on the parameter $p$ defining the norm.
If $p$ is not an even number, the function $x \mapsto \| x \|_p^p$ is not even a polynomial and thus, the result of Heinz cannot be applied directly to achieve our results.\\
Recently, Dadush, Peikert and Vempala presented in \cite{pp_dpv} a randomized algorithm for $\lmp$ for well-bounded convex bodies given by a separation oracle. The expected number of arithmetic operations of this algorithm is $\mathcal{O} ( n^{4/3n} ) \log_2 ( r )^{\mathcal{O} ( 1 )}$.
This is also the case, if the convex bodies are generated by an $\ell_p$-norm or a polyhedral norm, see Theorem 4.7 in \cite{pp_dpv_preprint}.
In \cite{pp_dpv_ellipsoid}, Dadush and Vempala derandomize this result without increasing the number of arithmetic operations.
The number of arithmetic operations of their algorithm is better than ours.
But the algorithm require exponential space whereas our algorithm requires only polynomial space.

\paragraph*{Organization}
The paper is organized as follows. In Section \ref{section_basics} we state some basic definitions and important facts used in this paper.
In Section \ref{section_main_result}, we formally define the lattice membership problem and present a polynomial time reduction from $\cvp$ to $\lmp$ for all $\ell_p$-norms with $1 < p < \infty$ and all polyhedral norms.
In Section \ref{sec_algorithm}, we describe Lenstra's algorithm as a general framework for algorithmic solutions of $\lmp$.
Then, we adapt this framework to concrete classes of convex sets:
In Section \ref{sec_LM_alg_polytope}, we consider polytopes and in 
Section \ref{sec_LM_alg_lp}, we consider the class of $\ell_p$-bodies, where $1 < p < \infty$.
In the description of the lattice membership algorithms we assume that we have access to an algorithm that compute a flatness direction of the bounded convex set.
In Section \ref{sec_flatness_algorithm}, we describe how we can compute a flatness direction of a polytope or an $\ell_p$-body,
which completes the description of the algorithm for $\lmp$. We view this section as the main technical contribution of our paper.

\section{Basic definitions and facts} \label{section_basics}

\noindent A polyhedron $P$ is the solution set of a system of inequalities given by a matrix $A \in \mathbbm{R}^{s \times n}$ and a vector $\beta \in \mathbbm{R}^s$,
$P = \{ x \in \mathbbm{R}^n | A x \leq \beta \}$.
In the following, we always assume that a polyhedron is given in this way.
A bounded polyhedron is called a polytope.

\noindent Every vector $d \in \mathbbm{R}^n \backslash \{ 0 \}$ defines a family of hyperplanes in $\mathbbm{R}^n$ by
$H_{k,d} := \{ x \in \mathbbm{R}^n | \langle x , d \rangle = k \}$,
where $k \in \mathbbm{R}$.
For any norm $ \| \cdot \|$ on $\mathbbm{R}^n$, every vector $x \in \mathbbm{R}^n$ and $\alpha >0$, we set $B_n^{(\| \cdot \|)} ( x , \alpha ) : = \{ y \in \mathbbm{R}^n | \| x - y \| < \alpha \}$.
We call this the ball generated by the norm $\| \cdot \|$ with radius $\alpha$ centered at $x$.
By $\bar{B}_n^{(\| \cdot \|)} ( x , \alpha )$ we denote the corresponding closed ball.
The Euclidean norm induces a matrix norm by
$\| A \| := \max \{ \| A x \|_2 | x \in \mathbbm{R}^n \mbox{ with } \| x \|_2 = 1 \}
 = \sqrt{\eta_n ( A^T A )}$,
where $\eta_n ( A^T A )$ is the square root of the largest eigenvalue of the matrix $A^T A$. 
It is called the spectral norm of a matrix.\\
 
\noindent A special case of Hölder's inequality \label{eq_Hölder} relates the $\ell_2$-norm to arbitrary $\ell_p$-norms: For all $x \in \mathbbm{R}^n$, we have $\| x \|_2 \leq \| x \|_p \leq n^{1/p-1/2} \| x \|_2$, if $1 \leq p \leq 2$, and
$n^{1/p-1/2} \| x \|_2 \leq \| x \|_p \leq \| x \|_2$, if $2 < p < \infty$.
For the $\ell_{\infty}$-norm, it holds that $n^{-1/2} \| x \|_2 \leq \| x \|_{\infty} \leq \| x \|_2$.\\

\noindent A lattice $L$ is a discrete subgroup of $\mathbbm{R}^n$.
Each lattice has a basis, i.e., a sequence $b_1, \hdots, b_m$ of $m$ elements of $L$ that generate $L$ as an abelian group.
We denote this by $L = \mathcal{L} ( B )$, where $B = [b_1, \hdots, b_m]$ is the matrix which consists of the columns $b_i$.
We call $m$ the rank of $L$ and $n$ the dimension of $L$.
If $m = n$, the lattice is full-dimensional.
The dual lattice $L^*$ of $L$ is defined as the set $\{ x \in \Span ( L ) | \langle x , v \rangle \in \mathbbm{Z}$ for all $v \in L \}$.
If $B$ is a basis of the full-dimensional lattice $L$, then $(B^T)^{-1}$ is a basis of $L^*$.
By $\lambda_1^{(2)} ( L )$ we denote the Euclidean length of a shortest non-zero vector in $L$.\\

\noindent Since we are interested in computational statements, we always assume that all numbers we are dealing with are rationals.
The size of a rational number $\alpha = p / q$ with $\gcd ( p , q ) = 1$ is defined as the maximum of the numerator and denominator in absolute values,
$\size ( \alpha ) := \max \{| p | , |q| \}$.
The size of a matrix or respectively a vector is the maximum of the size of its coordinates.
If we consider a polyhedron $P$ given by a matrix $A \in \mathbbm{Q}^{s \times n}$ and a vector $\beta \in \mathbbm{Q}^s$, then we denote by the size of $P$ the maximum of $n$, $s$, and the size of the coordinates of $A$ and $\beta$.
The size of a lattice $L \subseteq \mathbbm{Q}^n$ with respect to a basis $B$ is the maximum of $n$, $m$, and the length of the numerators and denominators of the coordinates of the basis vectors.
By the bit size or the representation size of a number $\alpha$, we mean $\log_2 ( \size ( \alpha ) )$.

\section{The lattice membership problem, main result, and reduction to $\cvp$} \label{section_main_result}

\begin{definition}
Given a lattice $L \subset \mathbbm{R}^n$ and a bounded convex set $\mathcal{C} \subseteq \Span ( L )$, we call the problem to decide whether $\mathcal{C}$ contains a vector from $L$, the lattice membership problem ($\lmp$).
\end{definition}

\noindent The lattice membership problem is a generalization of the integer programming feasibility problem from polyhedra to general bounded convex sets.
In Section \ref{sec_algorithm}, we will show that there exists a deterministic polynomial space algorithm, that solves $\lmp$, if the underlying convex set is an $\ell_p$-ball or a polytope.

\begin{thm} \label{state_main_result}
There exists a deterministic polynomial space algorithm that solves the lattice membership problem for all convex sets generated by an $\ell_p$-norm, $1 < p < \infty$, or a polyhedral norm.
\begin{itemize}
 \item If the convex set is generated by an $\ell_p$-norm, $1 < p < \infty$, the number of arithmetic operations is at most
 	$p \log_2 ( r )^{\mathcal{O} ( 1 )} n^{(2 + o ( 1 )) n}$.
 	Each number produced by the algorithm has bit size at most $p \cdot n^{\mathcal{O} ( 1 )} \log_2 ( r )$, where $r$ is an upper bound on the size of the convex set.
 \item If the convex set is a full-dimensional polytope given by $s$ constraints, then the number of arithmetic operations is at most
 $( s \cdot \log_2 ( r ) )^{\mathcal{O} ( 1 )} n^{(2+o(1)) n }$.
 Each number produced by the algorithm has bit size at most $n^{\mathcal{O} ( 1 )} \log_2 ( r )$, where $r$ is an upper bound on the size of the convex set.
\end{itemize}
\end{thm}

\noindent In the remainder of this section, we show that there exists a polynomial time reduction from $\cvp$ to $\lmp$ for all $\ell_p$-norms and all polyhedral norms.
For the reduction, we observe a relation between $\lmp$ and the decisional variant of $\cvp$.
In the decisional closest vector problem, we are given a lattice $L \subseteq \mathbbm{R}^n$, some target vector $t \in \Span ( L )$ and a parameter $\alpha > 0$. The goal is to decide whether the distance between $t$ and the lattice is at most $\alpha$ or not.
Obviously, the decisional closest vector problem is a special case of $\lmp$, where the corresponding convex body is the closed ball $\bar{B}_n^{(\| \cdot \|)} ( t , \alpha )$.
Micciancio and Goldwasser showed that $\cvp$ and its decisional variant are equivalent if one considers $\cvp$ with respect to the Euclidean norm, see \cite{bk_latticeproblems} and \cite{ln_micc_cvp}. 
Their result can be generalized to arbitrary $\ell_p$-norms, $1 < p < \infty$,
and to polyhedral norms.
Since we are interested in algorithmic solutions for this problem, we can always assume that $L \subseteq \mathbbm{Z}^n$ and $t \in \mathbbm{Z}^n$.

\begin{thm} \label{statement_cvp_self_lp} \label{statement_cvp_self_polyhedron}
Let $\| \cdot \|$ be a norm on $\mathbbm{R}^n$.
Assume that there exists an algorithm $\mathcal{A}$ that for all lattices $\mathcal{L} ( B' ) \subset \mathbbm{Z}^n$ of rank $m$ and all target vectors $t' \in \Span ( B ) \cap \mathbbm{Z}^n$ solves the
lattice membership problem for the ball $B_n^{(\| \cdot \|)} ( t' , r )$
in time
$T_{m,n}^{(\| \cdot \|)} ( S' , r )$, where $S'$ is an upper bound on the size of the basis $B'$ and the target vector $t'$.
\begin{itemize}
	\item
		If the norm is an $\ell_p$-norm, $1 \leq p \leq \infty$, then there exists an algorithm $\mathcal{A}'$, that solves the closest lattice vector problem for all lattices $\mathcal{L} ( B ) \subseteq \mathbbm{Z}^n$ and target vectors $t \in \Span ( B ) \cap \mathbbm{Z}^n$ in time
		$$k \cdot n^{\mathcal{O} ( 1 )} \log_2 ( S )^2 \cdot T ( 16 m^3 n^2 S^3 , m  n S ),$$
		where $k = p$ for $1 \leq p < \infty$ and $k = 1$ for $p = \infty$.
	\item 
		If the norm is given by a full-dimensional polytope symmetric about the origin given by $s$ constraints, then there exists an algorithm that solves the closest vector problem for all lattices $\mathcal{L} ( B ) \subseteq \mathbbm{Z}^n$ and target vector $t \in \Span ( B ) \cap \mathbbm{Z}^n$ in time
		$$s \cdot n^{\mathcal{O} ( 1 )} \log_2 ( \size ( P ) \cdot S ) 
		\cdot T_{m,n}^{(P)} ( 16 m^3 n^{n+2}\size ( P )^{n+1} 
		\cdot S^3 , n  m  S \size ( P )).$$
\end{itemize}
Here, $S$ is an upper bound on the size of the basis $B$ and the target vector $t$.
\end{thm}

\noindent For the proof of this theorem, it does not matter whether the algorithm solves $\lmp$ either for the open ball $B_n^{(\| \cdot \|)} ( t , r )$ or the corresponding closed balls. The proof of this theorem appears in the appendix, (see Section \ref{sec_appendix_selfreducibility} in the appendix).\\

\noindent Theorem \ref{statement_cvp_self_lp} together with Theorem \ref{state_main_result} implies a deterministic algorithm that solves the closest vector problem with respect to any $\ell_p$-norm, $1 < p < \infty$, and any polyhedral norm, e.g. the $\ell_1$-norm and the $\ell_{\infty}$-norm.
Furthermore, combining Theorem \ref{statement_cvp_self_polyhedron} with the inapproximability results for $\cvp$ from \cite{pp_dinur03} and \cite{pp_dinur02}, we get the following inapproximability result for $\lmp$.

\begin{thm}
For all bounded convex sets generated by an $\ell_p$-norm, $1 \leq p \leq \infty$, there is some constant $c > 0$ such that $\lmp$ is $\np$-hard to approximate within a factor $n^{c/\log \log n}$.
\end{thm}

\section{A general algorithm for the lattice membership problem} \label{sec_algorithm}

\noindent In this section, we describe a general framework for algorithms that solve the lattice membership problem.
Essentially, the algorithm is a variant of Lenstra's algorithm for integer programming \cite{pp_lenstraoptimization}, its improvements by Kannan \cite{pp_kannan}, and by Frank and Tardos \cite{pp_ft87}.\\


\noindent Our lattice membership algorithm is a recursive algorithm which works for classes of bounded convex sets, which are closed under bijective affine transformation and under intersection with hyperplanes orthogonal to the unit vectors.
In the following, we consider such a class $\mathcal{K}$ and call it \emph{suitable}.\\

\noindent Since $\mathcal{K}$ is closed under bijective affine transformation, it is enough to solve the lattice membership problem for instances, where the corresponding lattice is the integer lattice $\mathbbm{Z}^n$.
Since every vector from a lattice $L = \mathcal{L} ( B )$ is an integer linear combination of the basis vectors of $B$, any bounded convex set $\mathcal{C} \subseteq \Span ( L )$ contains a lattice vector from $L$ if and only if the bounded convex set $B^{-1} \mathcal{C}$ contains an integer vector.

\subsection{The main idea of the lattice membership algorithm}

\noindent The main idea of our lattice membership algorithm is to use the concept of branch and bound.
To decide, whether a given bounded convex set $\mathcal{C}$ from the class $\mathcal{K}$ contains an integer vector, we consider a family $\{ H_{k, \tilde{d}} \}_{k \in \mathbbm{Z}}$ of hyperplanes given by a vector $\tilde{d} \in \mathbbm{Z}^n$.
Obviously, every integer vector $v \in \mathbbm{Z}^n$, which is contained in $\mathcal{C}$, satisfies $\langle \tilde{d} , v \rangle = k$ for some integer value $k \in \mathbbm{Z}$
and $k$ is contained in the interval
\begin{equation} \label{eq_interval_recursive_instances}
\inf \{ \langle \tilde{d} , x \rangle | x \in \mathcal{C} \} 
\leq k \leq \sup \{ \langle \tilde{d} , x \rangle | x \in \mathcal{C} \}.
\end{equation}
Hence, to decide whether the bounded convex set $\mathcal{C}$ contains an integer vector, it is sufficient to consider all integer values $k$, which are contained in the interval (\ref{eq_interval_recursive_instances})
and check recursively whether the convex sets $\mathcal{C} \cap H_{k,\tilde{d}}$ contain an integer vector.\\

\noindent In the following, we will call an algorithm which realizes this idea a lattice membership algorithm.
Furthermore, we will assume the following for the class $\mathcal{K}$.

\begin{assumption} \label{assumption_flatness_algorithm}
Let $\mathcal{K}$ be a class of full-dimensional bounded convex sets and $f : \mathbbm{N} \to \mathbbm{R}^{> 0}$ be a nondecreasing function.
We assume that there exists a deterministic algorithm $\mathcal{A}_{\mathcal{K},f}$ that on input a convex set $\mathcal{C} \in \mathcal{K}$ of dimension $n$ outputs one of the following:
\begin{itemize}
 \item Either it outputs that $\mathcal{C}$ contains an integer vector or 
 \item it outputs a vector $\tilde{d} \in \mathbbm{Z}^n$ and an interval $I_{\mathcal{C}}$ of length at most $f ( n )$ such that $\mathcal{C}$ contains an integer vector if and only if there exists $k \in \mathbbm{Z} \cap I_{\mathcal{C}}$ such that $\mathcal{C} \cap H_{k,\tilde{d}}$ contains an integer vector.
\end{itemize}
\end{assumption}

\noindent We call such an algorithm $\mathcal{A}_{\mathcal{K},f}$ a flatness algorithm.
In Section \ref{sec_flatness_algorithm}, we will show that for certain classes of convex bodies we can realize a flatness algorithm.\\

\noindent Then using the idea of a membership algorithm, we obtain a recursive algorithm for the lattice membership problem, whose recursive instances are given by a full-dimensional bounded convex set $\mathcal{C}$ and an affine subspace $H$.
We start with $H := \mathbbm{R}^n$. Later, $H$ is given by a set of affine hyperplanes $H_{k_i,d_i}$, $m + 1 \leq i \leq n$ for some $m \leq n$.\\

\noindent Since the convex set $\mathcal{C} \cap H$ is not full-dimensional,
we construct a bijective affine mapping which maps the convex set $\mathcal{C} \cap H$ to a convex set in $\mathbbm{R}^n \cap (\bigcap_{i=m+1}^n H_{0,e_i})$ such that every integer vector in $\mathcal{C} \cap H$ is mapped to an integer vector in $\tau ( \mathcal{C} \cap H )$.
Such a convex set can be identified with a full-dimensional convex set in $\mathbbm{R}^m$.
Additionally, this transformation is constructed in such a way such that it guarantees that $\mathcal{C} \cap H$ contains an integer vector if and only if the corresponding convex set in $\mathbbm{R}^n \cap \bigcap_{i=m+1}^n H_{0,e_i}$ contains an integer vector.
Such a transformation is described in the following.

\noindent First of all, we use an integer vector $v \in H$ to map the affine subspace $H$ to the subspace $H - v$ which is given as the intersection of the affine hyperplanes $H_{0,d_i}$, $m+1 \leq i \leq n$.
Since the normal vectors $d_i$ of this subspace are linearly independent, they can be extended to a basis of the whole space $\mathbbm{R}^n$, $B = [ b_1, \hdots, b_m , d_{m+1} , \hdots, d_n]$.
Obviously, every vector $x \in ( H - v )$ satisfies
$B^T x = ( \bar{x}^T , 0^{n-m} )^T$, where $\bar{x} \in \mathbbm{R}^m$.
That means, the function $x \mapsto B^T x$ maps the subspace $( H - v ) = \bigcap_{i=m+1}^n H_{0,d_i}$ to the subspace $\bigcap_{i=m+1}^n H_{0,e_i}$.
To guarantee that we obtain a bijection between the integer vectors in $H - v$ and $\bigcap_{i=m+1}^n H_{0,e_i}$, we construct a basis of the lattice $\mathcal{L} ( B^T ) \cap \bigcap_{i=m+1}^n H_{0,e_i}$ and map every vector in this lattice to its corresponding integer coefficient vector.

\begin{claim} \label{claim_properties_transformation}
Let $\mathcal{C} \subseteq \mathbbm{R}^n$ be a full-dimensional bounded convex set.
For $m \in \mathbbm{N}$, $m < n$, let $H := \bigcap_{i=m+1}^n H_{k_i,d_i}$ be an affine subspace given by $d_i \in \mathbbm{Z}^n$ linearly independent and $k_i \in \mathbbm{Z}$.
Let $v \in \mathbbm{Z} \cap H$ and
$B = [ b_1 , \hdots, b_m , d_{m+1} , \hdots , d_n] \in \mathbbm{Z}^{n \times n}$ be a basis of $\mathbbm{R}^n$ which contains the vectors $d_i$, $m+1 \leq i \leq n$.
Let $\bar{D} \in \mathbbm{Z}^{n \times m}$ be a basis of the lattice $\mathcal{L} ( B^T ) \cap \bigcap_{i=m+1}^n H_{0,e_i}$ and
$\hat{D} := [ \bar{D} , e_{m+1} , \hdots, e_n ] \in \mathbbm{Z}^{n \times n}$.\\
Then, the bijective affine transformation
\begin{align*}
\tau : \mathbbm{R}^n \to \mathbbm{R}^n , ~ x \mapsto \hat{D}^{-1} B^T ( x - v )
\end{align*}
satisfies the following properties:
\begin{itemize}
 \item The transformation $\tau$ is a bijective transformation between the affine subspace $H$ and the subspace $\bigcap_{i=m+1}^n H_{0,e_i}$,
 	$\tau ( H ) = \bigcap_{i=m+1}^n H_{0,e_i}$.
 	\item The transformation $\tau$ is a bijective mapping between $\mathbbm{Z}^n \cap H$ and $\mathbbm{Z}^n \cap \bigcap_{i=m+1}^n H_{0,e_i}$.
\end{itemize}
\end{claim}

\noindent The transformation $\tau$ can be constructed efficiently:
Using the Hermite normal form, we can decide in polynomial time, if there exists an integer vector in the affine subspace $H$ and, if so, compute one, see Theorem 1.4.21 in \cite{bk_gls}.
The basis $\bar{D}$ of the lattice $\mathcal{L} ( B^T ) \cap \bigcap_{i=m+1}^n H_{0,e_i}$ can be constructed efficiently using a polynomial algorithm from Micciancio, see \cite{pp_MiccRed}.

\begin{proof}
Obviously, the transformation $\tau$ is well-defined.\\
We start with the proof of the first statement.
By definition of $\tau$, for all $x \in \mathbbm{R}^n$ and $m+1 \leq i \leq n$ we have that
\begin{align} \label{eq_mem_alg_general_1}
\langle \tau ( x ) , e_i \rangle
= \langle \hat{D}^{-1} B^T ( x - v ) , e_i \rangle
= \langle B^T ( x - v ) , ( \hat{D}^T )^{-1} e_i \rangle.
\end{align}
Since the columns of $\bar{D}$ are vectors in $\mathbbm{R}^n \cap \bigcap_{j=m+1}^n H_{0,e_j}$, we have
$\bar{D}^T e_i = 0$ for all $m+1 \leq i \leq n$.
Furthermore, $\hat{D}^T e_i = e_i$ for all $m+1 \leq i \leq n$.
Combining this with (\ref{eq_mem_alg_general_1}), it follows that
\begin{align*}
\langle \tau ( x ) , e_i \rangle
= \langle B^T ( x - v ) , e_i \rangle
= \langle x - v , B \cdot e_i \rangle
= \langle x - v , d_i \rangle.
\end{align*}
Since $v \in H = \bigcap_{j=m+1}^n H_{k_j,d_j}$, we have
$\langle \tau ( x ) , e_i \rangle
= \langle x , d_i \rangle - \langle v , d_i \rangle
= 0$
for all $m+1 \leq i \leq n$ and $x \in H$.
This shows that $\tau ( x ) \in \bigcap_{j=m+1}^n H_{0,e_j}$.
Since $\tau$ is bijective and the (affine) subspaces $H$ and $\bigcap_{j=m+1}^n H_{0,e_j}$ have the same dimension, it follows that $\tau ( H ) = \bigcap_{j=m+1}^n H_{0,e_j}$.
This proves the first statement.\\

\noindent We show the second statement in two steps.
First, we show that $\tau$ maps every integer vector in $H$ to an integer vector in $\mathbbm{R}^n \cap \bigcap_{i=m+1}^n H_{0,e_i}$.
Furthermore, we show that the inverse transformation $\tau^{-1}$ maps every integer vector in $\mathbbm{R}^n \cap \bigcap_{i=m+1}^n H_{0,e_i}$ to an integer vector in $H$.\\
For every integer vector $x \in \mathbbm{Z}^n$, we have $x - v \in \mathbbm{Z}^n$ and $B^T ( x - v ) \in \mathcal{L} ( B^T )$.
As both $x$ and $v$ are contained in $H$, it follows that
\begin{align*}
\langle B^T ( x - v ) , e_i \rangle
= \langle x - v , B e_i \rangle
= \langle x - v , d_i \rangle = 0
\end{align*}
for all $m+1 \leq i \leq n$.
This shows that $B^T ( x - v )$ is a vector in the lattice $\mathcal{L} ( B^T ) \cap \bigcap_{j=m+1}^n H_{0,e_j}$.
Since $\bar{D} \in \mathbbm{Z}^{n \times m}$ is a basis of this lattice, there exists an integer vector $z \in \mathbbm{Z}^m$ such that
$\bar{D} z = B^T ( x - v )$.
Obviously, the vector $z' = (z^T , 0^{n-m} )^T \in \mathbbm{Z}^n$ satisfies
$\hat{D} z' = B^T ( x - v )$.
From this, it follows that $\hat{D}^{-1} B^T ( x - z ) \in \mathbbm{Z}^n$.\\
The inverse of the bijective affine transformation $\tau$ is given by
\begin{align*}
\tau^{-1} : \mathbbm{R}^n \to \mathbbm{R}^n, ~ y \mapsto (B^T)^{-1} \hat{D} y + v.
\end{align*}
To show that $\tau^{-1} ( y ) \in \mathbbm{Z}^n$ for all integer vectors $y \in \mathbbm{Z}^n \cap \bigcap_{j=m+1}^n H_{0,e_j}$, it is enough to show that $(B^T)^{-1} \hat{D} y \in \mathbbm{Z}^n$.
Every integer vector $y' \in \mathbbm{Z}^n \cap \bigcap_{j=m+1}^n H_{0,e_j}$ is of the form $y' = ( y^T , 0^{n-m} )^T$ with $y \in \mathbbm{Z}^m$.
Obviously, we have $\hat{D} y' = \bar{D} y$.
Since $\bar{D}$ is a basis of the lattice $\mathcal{L} ( B^T ) \cap \bigcap_{j=m+1}^n H_{0,e_j}$, it follows that
\begin{align*}
\bar{D} y \in \mathcal{L} ( B^T ) \cap \bigcap_{j=m+1}^n H_{0,e_j} \subseteq \mathcal{L} ( B^T ).
\end{align*}
Hence, there exists an integer vector $w \in \mathbbm{Z}^n$ such that
$\bar{D} y = B^T w$.
\end{proof}

\noindent With this transformation $\tau$, we are able to identify the bounded convex set $\mathcal{C} \cap H$ with a full-dimensional bounded convex set in $\mathbbm{R}^m$.
Since $\mathcal{K}$ is closed under bijective affine transformation and intersection with hyperplanes orthogonal to the unit vectors, we have $\tau ( \mathcal{C} \cap H ) \in \mathcal{K}$ and we can apply the flatness algorithm $\mathcal{A}_{\mathcal{K},f}$ with input $\tau ( \mathcal{C} \cap H )$.
If the algorithm outputs that $\tau ( \mathcal{C} \cap H )$ contains an integer vector, we output that $\mathcal{C} \cap H$ contains an integer vector.
Otherwise, we obtain a vector $\tilde{d} \in \mathbbm{Z}^n$ and an interval $I_{\tau ( \mathcal{C} \cap H )}$ of length at most $f ( m )$ such that we need to search only in the hyperplane $H_{k,\tilde{d}}$, $k \in \mathbbm{Z} \cap I_{\tau ( \mathcal{C} \cap H )}$.
In this case the recursive instances of our membership algorithm are given by the bounded convex set $\mathcal{C}$ and the affine subspace $H \cap \tau^{-1} ( H_{k,\tilde{d}} )$.
For a complete description of the algorithm, see Algorithm \ref{alg_membership}.

\begin{figure}[!ht]
\framebox{
\begin{minipage}[b]{15.5cm} \small
\begin{alg} \label{alg_membership} {\bf Membership algorithm for bounded convex sets}\\
{\tt 
{\bf Input:}
	\begin{itemize}
	 \item A full-dimensional bounded convex set $\mathcal{C}$ from a suitable class $\mathcal{K}$.
	 \item An affine subspace $H:=\bigcap_{i=m+1}^n H_{k_i,d_i}$, where $d_i \in \mathbbm{Z}^n$ linearly independent and $k_i \in \mathbbm{Z}$ for all $m+1 \leq i \leq n$; alternatively, $H:= \mathbbm{R}^n$.
	\end{itemize}
{\bf Used Subroutine:} Flatness algorithm $\mathcal{A}_{\mathcal{K},f}$ satisfying Assumption \ref{assumption_flatness_algorithm}.\\[0.25cm]
	{\bf If} $m=0$, check if there exists $z \in \mathbbm{Z}^n \cap H$ satisfying $z \in \mathcal{C}$.\\
	{\bf Otherwise,} 
	\begin{enumerate}
	 	\item 
	 			\begin{description}
	 				\item[If] $m=n$, set $v := 0$ and $\bar{V} := I_n$.
	 				\item[Otherwise,] compute $v \in \mathbbm{Z}^n \cap H$, a basis $B := [b_1, \hdots, b_m , d_{m+1} , \hdots, d_n] \in \mathbbm{Z}^{n \times n}$ of $\mathbbm{R}^n$.\\
			Compute a lattice basis $\bar{D} \in \mathbbm{Z}^{n \times m}$ of $\mathcal{L} ( B^T ) \cap \bigcap_{i=m+1}^n H_{0,e_i}$.\\
			Set $\hat{D}:=[\bar{D} , e_{m+1}, \hdots, e_n] \in \mathbbm{Z}^n$ and $\bar{V} := \hat{D}^{-1} B^T$.
			\end{description}
			Define the bijective mapping $\tau: \mathbbm{R}^n \to \mathbbm{R}^n$, $x \mapsto \bar{V} ( x - v )$.
	\item Apply the algorithm $\mathcal{A}_{\mathcal{K},f}$ with input $\tau ( \mathcal{C} \cap H )$.
 				\begin{description}
 				\item[If] the algorithm outputs that $\tau ( \mathcal{C} \cap H )$ contains an integer vector, output this.
 				\item[Otherwise,] the result is a vector $\tilde{d} \in \mathbbm{Z}^m$ together with an interval $I_{\tau ( \mathcal{C} \cap H )}$.
 				\begin{enumerate}
 				 \item Set $d_m := \bar{V}^T ( \tilde{d}^T , 0^{n-m} )^T \in \mathbbm{Z}^n$.
 			 	 \item For all $k \in \mathbbm{Z} \cap I_{\tau ( \mathcal{C} \cap H )}$,
 								\begin{quote}
 								 apply the membership algorithm to the convex set $\mathcal{C}$ and the affine subspace $H \cap H_{k+\langle v , d_m \rangle,d_m}$.\\
												The algorithm outputs whether the convex set $\mathcal{C} \cap H \cap H_{k+ \langle v , d_m \rangle,d_m}$ contains an integer vector or not.
								 \end{quote}
 								 \item \label{alg_step_decision} If there exists an index $k$ such that  $\mathcal{C} \cap H \cap H_{k + \langle v , d_m \rangle,d_m}$ contains an integer vector, output that $\mathcal{C} \cap H$ contains an integer vector.\\
 								 	Otherwise, output that $\mathcal{C} \cap H$ does not contain an integer vector.
 				\end{enumerate}	
 					\end{description}
\end{enumerate}
}
\end{alg}
\end{minipage}}
\end{figure}

\begin{thm} \label{alg_mem_result}
Let $\mathcal{K}$ be a suitable class of bounded convex sets satisfying Assumption \ref{assumption_flatness_algorithm}. 
Given a full-dimensional convex set $\mathcal{C} \in \mathcal{K}$ and an affine subspace $H$ of dimension $m$, the membership algorithm, Algorithm \ref{alg_membership}, decides correctly whether $\mathcal{C} \cap H$ contains an integer vector.
The number of recursive calls of the algorithm is at most $( 2 f ( m ))^m$.
\end{thm}

\noindent Given as input a full-dimensional bounded convex set $\mathcal{C} \subseteq \mathbbm{R}^n$ and as subspace the whole vector space $\mathbbm{R}^n$, the algorithm solves the lattice membership problem correctly.

\begin{proof}
Obviously, if $m = 0$, the affine subspace $H$ consists of a single vector. Hence, the algorithm can decide correctly, whether this vector is an integer vector which is contained in $\mathcal{C}$.\\
For $m \geq 1$, the membership algorithm computes the bijective affine mapping $\tau$ as described in Claim \ref{claim_properties_transformation} and applies the algorithm $\mathcal{A}_{\mathcal{K},f}$ to the full-dimensional bounded convex set $\tau ( \mathcal{C} \cap H ) \subseteq \mathbbm{R}^m$.
Depending on the output, the algorithm distinguishes between two cases:\\

\noindent If the algorithm outputs that $\tau ( \mathcal{C} \cap H )$ contains an integer vector, it follows directly from Claim \ref{claim_properties_transformation} that $\mathcal{C} \cap H$ contains an integer vector.\\

\noindent Otherwise, the algorithm works recursively and checks for each $k \in \mathbbm{Z} \cap I_{\tau ( \mathcal{C} \cap H )}$, whether the convex set $\mathcal{C} \cap H \cap H_{k+\langle v , d_m \rangle , d_m}$ contains an integer vector.
We have seen in Claim \ref{claim_properties_transformation} that $\mathcal{C} \cap H$ contains an integer vector if and only if $\tau ( \mathcal{C} \cap H )$ contains an integer vector, i.e., a vector from $\mathbbm{Z}^n \cap \bigcap_{i=m+1}^n H_{0,e_i}$.\\
If we interpret $\tau ( \mathcal{C} \cap H )$ as a full-dimensional convex set in $\mathbbm{R}^m$, it is guaranteed by Assumption \ref{assumption_flatness_algorithm} that that $\tau ( \mathcal{C} \cap H )$ contains an integer vector if and only if there exists an integer value $k \in \mathbbm{Z} \cap I_{\tau ( \mathcal{C} \cap H)}$ such that
		$\tau ( \mathcal{C} \cap H ) \cap H_{k,\tilde{d}}$ contains an integer vector.
Obviously, this is equivalent to the statement that 
$\tau ( \mathcal{C} \cap H ) \cap H_{k,(\tilde{d}^T,0^{n-m})^T}$ contains an integer vector from $\mathbbm{Z}^n \cap \bigcap_{i=m+1}^n H_{0,e_i}$, if we interpret $\tau ( \mathcal{C} \cap H )$ as a convex set in $\mathbbm{R}^n \cap \bigcap_{i=m+1}^n H_{0,e_i}$.
Since $\tau$ is a bijective affine transformation which maps an integer vector in $\mathcal{C} \cap H$ to an integer vector in $\tau ( \mathcal{C} \cap H )$, this is equivalent to the statement that
		$\mathcal{C} \cap H \cap \tau^{-1} ( H_{k,(\tilde{d}^T,0^{n-m})^T} )$ contains an integer vector.
		Since 
		$\tau^{-1} ( H_{k,(\tilde{d}^T,0^{n-m} )^T} ) = H_{k + \langle v , d_m \rangle, d_m}$,
		it follows that
		$\mathcal{C} \cap H$ contains an integer vector if and only if there exists an index $k \in \mathbbm{Z} \cap I_{\tau ( \mathcal{C} \cap H )}$ such that
		$\mathcal{C} \cap H \cap H_{k+\langle v , d_m \rangle , d_m}$ contains an integer vector.\\

\noindent If we are given as input a convex set in $\mathbbm{R}^n$ together with an affine subspace of dimension $m$, we need at most $f(m)+1$
solutions of recursive instances, where the dimension of the subspace is $m-1$, since the length of the interval $I_{\tau ( \mathcal{C} \cap H )}$ is at most $f ( m )$. 
Hence, the overall number of recursive calls is at most
$$\prod_{i=1}^m ( f(i) + 1 ) \leq 2^m f(m)^m.$$
\end{proof}

\noindent Obviously, our lattice membership algorithm runs in polynomial space if the bit size of each number computed by the algorithm is polynomial in the bit size of the input instance.
Similarly to the algorithms by Lenstra and Kannan, this cannot be guaranteed for the outline of our lattice membership algorithm presented so far.
In fact, the size of the newly constructed affine hyperplane depends not only on the size of the convex set $\mathcal{C}$ but also on the size of the affine subspace.
To avoid this problem, we use a replacement procedure due to Frank and Tardos, see \cite{pp_ft87},
which we describe in the following section.
	
	\subsection{Modification of the Lattice Membership Algorithm}

\noindent The replacement procedure from Frank and Tardos presented in \cite{pp_ft87} is a polynomial algorithm that on input an affine subspace $H \subseteq \mathbbm{R}^n$ and an additional hyperplane $H_{k,d}$ computes a set of new hyperplanes $H_{\bar{k}_i,\bar{d}_i}$, $i \in J$, with small size.
If the parameters are chosen appropriate depending on the shape of the convex set, then it can be guaranteed that each integer vector in the convex set is contained in the affine subspace $H \cap H_{k,d}$ if and only if it is contained in the intersection $H \cap \bigcap_{i \in J} H_{\bar{k}_i,\bar{d}_i}$.
The following result is a slightly generalization of Lemma 5.1 in \cite{pp_ft87}.
The proof of it together with a short description of the procedure appears in the full version of this paper.

\begin{prop} \label{prop_replacement_procedure}
There exists a replacement procedure, which satisfies the following properties:\\
Given as input a parameter $N \in \mathbbm{N}$, an affine subspace $H$ and an additional affine hyperplane $H_{k,d}$ the replacement procedure computes a set of hyperplanes $H_{\bar{k}_i, \bar{d}_i}$, $i \in J \not= \emptyset$, such that the following holds:
Every integer vector $z \in \bar{B}_n^{(1)} ( 0 , N - 1 ) \cap H$ satisfies $\langle d , z \rangle = k$ if and only if it satisfies $\langle \bar{d}_i , z \rangle = \bar{k}_i$ for all $i \in J$.\\
The size of the vectors $\bar{d}_i \in \mathbbm{Z}^n$ and the numbers $\bar{k}_i \in \mathbbm{Z}$ is at most
		$2^{(n+2)^2} N^n$.
The number of arithmetic operations of the replacement procedure is at most $(n \cdot \log_2 ( N ) )^{\mathcal{O} ( 1 )}$.
\end{prop}

\noindent We will use this replacement procedure in the lattice membership algorithm directly before the recursive call of the algorithm with a suitable computed parameter $N$.
This guarantees that we obtain additional hyperplanes whose size depend only on the size of the convex set $\mathcal{C}$, or to be precise on the parameter $N$ defining the radius of a circumscribed $\ell_1$-ball, and not on the size of the affine subspace $H$.
We call this algorithm the modified membership algorithm.
For completeness, a formal description of this algorithm appears in Algorithm \ref{alg_ModMem}.

\begin{figure}[!ht]
\framebox{
\begin{minipage}[b]{15.5cm} \small
\begin{alg} \label{alg_ModMem} {Modified membership algorithm for bounded convex sets}\\
{\tt 
{\bf Input:}
	\begin{itemize}
	 \item A full-dimensional bounded convex set $\mathcal{C}$ from a suitable class $\mathcal{K}$ and
	 \item an affine subspace $H:=\bigcap_{i=m+1}^n H_{k_i,d_i}$, where $d_i \in \mathbbm{Z}^n$ linearly independent and $k_i \in \mathbbm{Z}$ for all $m+1 \leq i \leq n$; alternatively, $H:= \mathbbm{R}^n$.
	\end{itemize}
{\bf Used Subroutine:} Flatness algorithm $\mathcal{A}_{\mathcal{K},f}$ satisfying Assumption \ref{assumption_flatness_algorithm}, replacement procedure.\\[0.25cm]
	{\bf If} $m=0$, check if there exists $z \in \mathbbm{Z}^n \cap H$ satisfying $z \in \mathcal{C}$.\\
	{\bf Otherwise, }
	\begin{enumerate}
	\item \begin{description}
				\item[If] $m=n$, set $v := 0$ and $\bar{V} := I_n$.
				\item[Otherwise,] compute $v \in \mathbbm{Z}^n \cap H$, a basis $B := [b_1, \hdots, b_m , d_{m+1} , \hdots, d_n] \in \mathbbm{Z}^{n \times n}$ of $\mathbbm{R}^n$.\\
			Compute a lattice basis $\bar{D} \in \mathbbm{Z}^{n \times m}$ of $\mathcal{L} ( B^T ) \cap \bigcap_{i=m+1}^n H_{0,e_i}$.\\
			Set $\hat{D}:=[\bar{D} , e_{m+1}, \hdots, e_n] \in \mathbbm{Z}^n$ and $\bar{V} := \hat{D}^{-1} B^T$.
			\end{description}
			Define the bijective mapping $\tau: \mathbbm{R}^n \to \mathbbm{R}^n$, $x \mapsto \bar{V} ( x - v )$.
	\item Apply the algorithm $\mathcal{A}_{\mathcal{K},f}$ with input $\tau ( \mathcal{C} \cap H )$.
 				\begin{description}
 				\item[If] the algorithm outputs that $\tau ( \mathcal{C} \cap H )$ contains an integer vector, output this.
 				\item[Otherwise,] the result is a vector $\tilde{d} \in \mathbbm{Z}^m$ together with an interval $I_{\tau ( \mathcal{C} \cap H )}$.
 				\begin{enumerate}
 				 \item Set $d_m := \bar{V}^T ( \tilde{d}^T , 0^{n-m} )^T \in \mathbbm{Z}^n$.\\
 				 			Compute a parameter $N \in \mathbbm{N}$ such that $\mathcal{C} \subseteq \bar{B}_n^{(1)} ( 0 , N - 1 )$.
 			 	 \item For all $k \in \mathbbm{Z} \cap I_{\tau ( \mathcal{C} \cap H )}$,
 								\begin{itemize}
 									\item apply the replacement procedure to the affine subspace $H$, the hyperplane given by $d_m$ and $k + \langle v , d_m \rangle$ and the parameter $N$.\\
 											The result is an index set $J_k$ and an affine subspace $\bigcap_{i \in I_k} H_{\bar{k}_i,\bar{d}_i}$.
 								  \item Apply the modified membership algorithm the convex set $\mathcal{C}$ and the affine subspace $H \cap 												\bigcap_{i\in J_k} H_{\bar{k}_i,\bar{d}_i}$.\\
												As a result, we get the information if $\mathcal{C} \cap H \cap \bigcap_{i\in J_k} H_{\bar{k}_i,\bar{d}_i}$ contains an integer vector or not.
								 \end{itemize}
 								 \item \label{alg_step_decision} If there exists an index $k$ such that $\mathcal{C} \cap H \cap \bigcap_{i\in J_k} H_{\bar{k}_i,\bar{d}_i}$ contains an integer vector, output that $\mathcal{C} \cap H$ contains an integer vector.\\
 								 	Otherwise, output that $\mathcal{C} \cap H$ does not contain an integer vector.			
 				\end{enumerate}			
 					\end{description}
\end{enumerate}
}\end{alg}
\end{minipage} }
\end{figure}

\begin{thm} \label{alg_mod_mem_result}
Let $\mathcal{K}$ be a suitable class of bounded convex sets satisfying Assumption \ref{assumption_flatness_algorithm}.\\
Given a full-dimensional bounded convex set $\mathcal{C} \subseteq \mathbbm{R}^n$ from the class $\mathcal{K}$ and an affine subspace $H$, the modified lattice membership algorithm, Algorithm \ref{alg_ModMem}, decides correctly whether $\mathcal{C} \cap H$ contains an integer vector or not.
Each recursive instance consists of the original convex set $\mathcal{C}$ and an affine subspace of size at most 
	$\max \{ \size ( H ) , 2^{(n+2)^2} N^n \}$, where $\mathcal{C} \subseteq \bar{B}_n^{(1)} ( 0 , N - 1 )$.
\end{thm}

\begin{proof}
Since $\mathcal{C} \subseteq \bar{B}_n^{(1)} ( 0 , N-1)$, for all $k \in \mathbbm{Z}$, the convex set $\mathcal{C}$ contains an integer vector from $H \cap H_{k+\langle v , d_m \rangle , d_m}$, if and only if it contains an integer vector from
$H \cap \bigcap_{i\in I_k} H_{\bar{k}_i,\bar{d}_i}$, (see Proposition \ref{prop_replacement_procedure}).
Hence, the correctness of the algorithm follows directly Theorem \ref{alg_mem_result}.
Also, the upper bound on the size of the recursive instances follows directly from Proposition \ref{prop_replacement_procedure}.
\end{proof}
	
\noindent  Obviously, we are able to adapt this general framework for all classes of bounded convex sets for which there exists a flatness algorithm.
For polytopes and $\ell_p$-balls we will see that we are able to do this using so called Löwner-John ellipsoids.

\section{A lattice membership algorithm for polytopes} \label{sec_LM_alg_polytope}

\noindent In this section, we consider full-dimensional polytopes given by a matrix $A \in \mathbbm{Z}^{s\times n}$ and a vector $\beta \in \mathbbm{Z}^s$.
Obviously, the class of all full-dimensional polytopes is closed under intersection with hyperplanes and under bijective affine transformation.
Furthermore, in Section \ref{sec_flatness_algorithm_polytopes}, we will show that there exists a flatness algorithm for polytopes.

	\begin{thm} \label{thm_flatness_algorithm_polytopes}
There exists a flatness algorithm that for all full-dimensional polytopes $P \subseteq \mathbbm{R}^n$ outputs one of the following:
Either it outputs that $P$ contains an integer vector or it outputs a vector $\tilde{d} \in \mathbbm{Z}^n$ and an interval $I_P$ of length at most $2n^2$ such that $P$ contains an integer vector if and only if there exists $k \in \mathbbm{Z} \cap I_P$ such that $P \cap H_{k,\tilde{d}}$ contains an integer vector.\\
The number of arithmetic operations of the flatness algorithm is 
$$s^{\mathcal{O} ( 1 )} \log_2 ( r ) n^{n/(2e)+o(n)}$$
and each number computed by the algorithm has size at most $r^{n^{\mathcal{O} ( 1 )}}$, where $r$ is an upper bound on the size of the polytope and $s$ is the number of constraints defining the polytope.
\end{thm}

\noindent Using this result, we can adapt the algorithmic framework from Section \ref{sec_algorithm}, to solve the lattice membership problem for polytopes and we obtain a lattice membership algorithm for polytopes.
To compute the parameter $N$, which defines a circumscribed $\ell_1$-ball of the polytope, we use that the vertices of every full-dimensional polytope given by integral constraints are at most $n^{(n+1)/2} \size ( P )^n$ (in absolute value).
Hence, we set $N$ as $r^{(n+3)/2} r^n$, where $r$ is an upper bound on the size of the polytope.
A detailed description of the algorithm is given in Algorithm \ref{alg_membership_polyhedron}.

\begin{figure}[!ht]
\framebox{
\begin{minipage}[b]{15.5cm} \small
\begin{alg} \label{alg_membership_polyhedron} {\bf Lattice membership algorithm for polytopes}\\
{\tt
{\bf Input:}
	\begin{itemize}
		\item A full-dimensional polytope $P$ given by $A \in \mathbbm{Z}^{s \times n}$ and $\beta \in \mathbbm{Z}^s$ with size $r_P$ and
		\item an affine subspace $H := \bigcap_{i=m+1}^n H_{k_i,d_i}$ given by $d_i \in \mathbbm{Z}^n$ linearly independent and $k_i \in \mathbbm{Z}$, $m+1 \leq i \leq n$; alternatively, $H:= \mathbbm{R}^n$.
	\end{itemize}
{\bf Used Subroutines:} flatness algorithm for polytopes, replacement procedure.\\[0.25cm]
	{\bf If} $m=0$, check if there exists $z \in \mathbbm{Z}^n \cap H$ satisfying $z \in P$.\\
	{\bf Otherwise, }
	\begin{enumerate}
	\item \begin{description}
					\item[If] $m=n$, set $v := 0$ and $\bar{V} := I_n$.
	 				\item[Otherwise,] compute $v \in \mathbbm{Z}^n \cap H$, a basis $B := [b_1, \hdots, b_m , d_{m+1} , \hdots, d_n] \in \mathbbm{Z}^{n \times n}$ of $\mathbbm{R}^n$.\\
			Compute a lattice basis $\bar{D} \in \mathbbm{Z}^{n \times m}$ of $\mathcal{L} ( B^T ) \cap \bigcap_{i=m+1}^n H_{0,e_i}$.\\
			Set $\hat{D}:=[\bar{D} , e_{m+1}, \hdots, e_n] \in \mathbbm{Z}^n$ and $\bar{V} : = \hat{D}^{-1} B^T$.
			\end{description}
	\item Apply the flatness algorithm for polytopes to the polytope $\tilde{P}$ given by $\tilde{A} \in \mathbbm{Z}^{s \times m}$ and $\beta - A v \in \mathbbm{Z}^s$, where $\tilde{A}$ is the matrix which consists of the first $m$ columns of the matrix $A \bar{V}^{-1}$. 
 				\begin{description}
 				\item[If] the algorithm outputs that $\tilde{P}$ contains an integer vector, output that $P \cap H$ contains an integer vector.
 				\item[Otherwise,] the result is a vector $\tilde{d} \in \mathbbm{Z}^m$ together with an interval $I_{\tilde{P}}$.
 				\begin{enumerate}
 				 \item Set $d_m := B ( \tilde{D}^T )^{-1} ( \tilde{d}^T , 0^{n-m} )^T \in \mathbbm{Z}^n$
 				  			and $N:= n^{(n+3)/2} r_P^n+1$.
 			 	 \item For all $k \in \mathbbm{Z} \cap I_{\tilde{P}}$,
 								\begin{itemize}
 									\item apply the replacement procedure to the affine subspace $H$, the hyperplane given by $d_m$ and $k + \langle v , d_m \rangle$ and the parameter $N$.\\
 											The result is an index set $J_k$ and an affine subspace $\bigcap_{i \in J_k} H_{\bar{k}_i,\bar{d}_i}$.
 								  \item Apply the modified membership algorithm to the polytope $P$ and the affine subspace $H \cap 												\bigcap_{i\in J_k} H_{\bar{k}_i,\bar{d}_i}$.\\
												As a result, we get the information if $P \cap H \cap \bigcap_{i\in J_k} H_{\bar{k}_i,\bar{d}_i}$ contains an integer vector or not.
								 \end{itemize}
 								 \item \label{alg_step_decision} If there exists an index $k$ such that $P \cap H \cap \bigcap_{i\in J_k} H_{\bar{k}_i,\bar{d}_i}$ contains an integer vector, output that $P \cap H$ contains an integer vector.\\
 								 	Otherwise, output that $P \cap H$ does not contain an integer vector.
 					
 				\end{enumerate}			
 					\end{description}
\end{enumerate}}
\end{alg}
\end{minipage}}
\end{figure}

\begin{thm} \label{result_membership_polyhedron}
Let $P \subseteq \mathbbm{R}^n$ be a full-dimensional polytope given by a matrix $A \in \mathbbm{Z}^{s \times n}$ and a vector $\beta \in \mathbbm{Z}^s$.
Let $H \subseteq \mathbbm{R}^n$ be an affine subspace of dimension $m \leq n$.\\
Given as input $P$ and $H$, the lattice membership algorithm for polytopes, Algorithm \ref{alg_membership_polyhedron}, decides correctly whether $P \cap H$ contains an integer vector.
The number of arithmetic operations of the algorithm is 
$$( n \cdot s  \log_2 ( r ))^{\mathcal{O} ( 1 )} m^{(2+o(1))m},$$
where $r$ is an upper bound on the size of the polytope $P$ and the affine subspace $H$.
The algorithm runs in polynomial space and each number computed by the algorithm has size at most $r^{n^{\mathcal{O} ( 1 )}}$, that means bit size at most $n^{\mathcal{O} ( 1 )} \log_2 ( r )$.
\end{thm}

\begin{proof}
Since $P$ is a polytope with size at most $r$, it is contained in the ball
$\bar{B}_n^{(\infty)} ( 0 , t )$ with $t = n^{(n+1)/2} r^n$, see Lemma \ref{lemma_polyhedron_out_box} in the Appendix.
Hence, $P \subseteq \bar{B}_n^{(1)} ( 0 , n^{n/2+1} r^n )$ and the parameter $N$ computed by the algorithm satisfies
$P \subseteq \bar{B}_n^{(1)} ( 0 , N-1 )$.\\
The transformation $\tau: x \mapsto \bar{V} ( x - v )$ maps the intersection $P \cap H$ to the polytope
$\{ x \in \mathbbm{R}^n | A \bar{V}^{-1}  x \leq \beta - A v \}
\cap \bigcap_{i=m+1}^n H_{0,e_i}$
which can be identified with the full-dimensional polytope
$\tilde{P} = \{ x \in \mathbbm{R}^m | \tilde{A} x \leq \beta - A v \}$,
where $\tilde{A} \in \mathbbm{Z}^{s \times m}$ consists of the first $m$ columns of the matrix $A \bar{V}^{-1}$.
Hence, the correctness of the algorithm follows from Theorem \ref{alg_mod_mem_result}.\\

\noindent The lattice membership algorithm for polytopes gets as input a polytope $P \subseteq \mathbbm{R}^n$ and an affine subspace $H$ of size of at most $r$.
It is easy to see that each number computed by the algorithm in one reduction step has size at most $r^{n^{\mathcal{O} ( 1 ) }}$.
The recursive instances of the lattice membership algorithm consist of the original input polytope $P$ and a new affine subspace.
The new affine subspace is the intersection of the original subspace $H$ and another subspace $\bigcap_{i \in J_k} H_{\bar{k}_i,\bar{d}_i}$.
Therefore, we denote the size of the polytope $P$ separately by $r_P$.
Obviously, $r \geq r_P$.
According to Theorem \ref{alg_mod_mem_result}, the size of the affine subspace used for the recursive calls of the algorithm is at most 
$$\max \{ \size ( H ) , r_P^{n^{\mathcal{O} ( 1 )}} \}
= \max \{ r , r_P^{n^{\mathcal{O} ( 1 )}} \}.$$
Especially, the replacement procedure guarantees that the size of the additional hyperplanes $H_{\bar{k}_i,\bar{d}_i}$, $i \in J_k$, depends only on the size of the polytope $r_P$.
Hence, it follows that
$$\max \{r_P^{n^{\mathcal{O} ( 1 )}} , r \}^{n^{\mathcal{O} ( 1 )}} = r^{n^{\mathcal{O} ( 1 )}}$$
is an upper bound on the size of each number computed by the lattice membership algorithm.

Finally, we give an upper bound on $T ( m , n , s , r_P , r )$, the number of arithmetic operations of the lattice membership algorithm.
Obviously, $T ( 0 , n , s , r_P , r ) = ( s \cdot n )^{\mathcal{O} ( 1 )}$.
Given a full-dimensional polytope in $\mathbbm{R}^n$ together with an affine subspace of dimension $m > 0$, the computation of the affine bijective transformation $\tau$ can be done using at most $n^{\mathcal{O} ( 1 )}$ arithmetic operations.
The number of arithmetic operations of the flatness algorithm depends on the size of the input polytope $\tilde{P}$.
Hence, the number of arithmetic operations of the flatness algorithm is at most 
$$( s \cdot n )^{\mathcal{O} ( 1 )} \log_2 ( r ) m^{m/(2e) + o(m)},$$
as stated in Theorem \ref{thm_flatness_algorithm_polytopes}.
The number of arithmetic operations of the replacement procedure is polynomial in $n$ and $\log_2 ( N )$.
By our definition of $N$, we have
$\log_2 ( N ) \leq (n^2+1) \log_2 ( 2 m r )$.
This shows that the number of arithmetic operations of the replacement procedure is at most
$(n \cdot \log_2 ( r_P ) )^{\mathcal{O} ( 1 )}$, see Proposition \ref{prop_replacement_procedure}.
The number of recursive calls of the lattice membership algorithm is determined by the length of the interval computed by the flatness algorithm.
The length of this interval is at most $2m^2$, see Theorem \ref{alg_mem_result}.
Hence, we obtain the following recursion of the number of arithmetic operations
\begin{equation*}
T ( m , n , s , r_P , r )
 \leq (n s \log_2 ( r ))^{\mathcal{O} ( 1 )} m^{m/(2e)+o(m)} + ( 2 m^2 + 1 ) \cdot T ( m - 1 , n , s , r_P , \max \{ r , r_P^{n^{\mathcal{O} ( 1 )}} \}  ).
\end{equation*}
Here, the main observation is that the size of the polytope does not change, whereas the size of the new subspace is the maximum of the size of the original input subspace and $r_P^{n^{\mathcal{O} ( 1 )}}$.
Hence, it follows by induction, that for all $m \geq 0$,
$$T ( m , n , s , r_P , r ) \leq
 ( n s \log_2 (r))^{\mathcal{O} ( 1 )}  \cdot m^{(2+o(1))m}.$$
\end{proof}

\noindent If we choose as subspace $H = \mathbbm{R}^n$, we obtain an algorithm for the lattice membership problem.

\begin{cor}
The lattice membership algorithm for polytopes, Algorithm \ref{alg_membership_polyhedron}, solves $\lmp$ for all full-dimensional polytopes given by a matrix $A \in \mathbbm{Z}^{s \times n}$ and a vector $\beta \in \mathbbm{Z}^s$ correctly.
The number of arithmetic operations of the algorithm is at most
$s^{\mathcal{O} (1)} \log_2 ( r )^{\mathcal{O} ( 1 )} n^{(2+o(1))n}$.
The algorithm runs in polynomial space and each number produced by the algorithm has bit size at most $n^{\mathcal{O} ( 1 )} \log_2 ( r )$.
\end{cor}
	
\section{A lattice membership algorithm for $\ell_p$-balls, $1 < p < \infty$} \label{sec_LM_alg_lp}

	\noindent Next, we use the algorithm framework presented in Section \ref{sec_algorithm} to obtain an algorithm that solves $\lmp$ for $\ell_p$-balls with $1 < p < \infty$.
Since the set of all $\ell_p$-balls is not closed under bijective affine transformation,
we consider in the following a generalization of them.
We consider norms, whose unit balls are the linear map of the $\ell_p$-unit ball.

\begin{definition}
Let $V \in \mathbbm{R}^{n \times n}$ be nonsingular.
For a vector $x \in \mathbbm{R}^n$, we define
$$\| x \|_p^V := \| V^{-1} x \|_p.$$
\end{definition}

\noindent Obviously, the mapping $\| \cdot \|_p^V$ defines a norm on $\mathbbm{R}^n$.
We denote the balls generated by such a norm by $B_n^{(p,V)} ( t , \alpha )$, i. e., for $t \in \mathbbm{R}^n$ and $\alpha > 0$ we define
$$B_n^{(p,V)} ( t , \alpha ) := 
 \{ x \in \mathbbm{R}^n | \| x - t \|_p^V < \alpha \}.$$
If the matrix $V$ is an orthogonal matrix, the unit ball of this norm is just the rotation of the $\ell_p$-unit ball by the matrix $V$.
If we consider the standard $\ell_p$-norm, we omit the matrix $I_n$ and write $B_n^{(p)} ( t , \alpha )$ instead.\\

\noindent To use these convex sets in the lattice membership algorithm, we need to consider their intersection with hyperplanes orthogonal to the unit vectors.
To be precise, for $m \in \mathbbm{N}$, $m \leq n$, we define
$$B_{m,n}^{(p,V)} ( t , \alpha ) := B_n^{(p,V)} ( t , \alpha ) \cap \bigcap_{i=m+1}^n H_{0,e_i}.$$
We will call these convex sets \emph{$\ell_p$-bodies}\footnote{Obviously, $\ell_p$-bodies are not convex bodies but bounded convex sets.}.
In the following if we speak of an $\ell_p$-body, we assume that we are given a nonsingular matrix $V \in \mathbbm{R}^{n \times n}$, a vector $t \in \mathbbm{R}^n$, parameter $m \in \mathbbm{N}$, $m \leq n$, and $\alpha > 0$ and we consider the convex set $B_{m,n}^{(p,V)} ( t , \alpha )$.
The size of such an $\ell_p$-body is the maximum of $m$, $n$, $\alpha$ and the size of the coordinates of $V^{-1}$ and $t$.\\

\noindent Formally, the $\ell_p$-body $B_{m,n}^{(p,V)} ( t , \alpha )$ is a $m$-dimensional bounded convex set in the subspace $\Span ( e_1, \hdots, e_m )$ of the vector space $\mathbbm{R}^n$.
But in the following, we will neglect this and we will interpret $B_{m,n}^{(p,V)} ( t , \alpha )$ as a full-dimensional bounded convex set in the vector space $\mathbbm{R}^m$.
Then, we say that a vector $x \in \mathbbm{R}^m$ is contained in $B_{m,n}^{(p,V)} ( t , \alpha )$ if and only if $( x^T , 0^{n-m} )^T \in B_n^{(p,V)} ( t , \alpha )$.
In Section \ref{sec_flatness_algorithm_lp}, we will show that for all $\ell_p$-bodies, there exists a flatness algorithm.
	

\begin{thm} \label{thm_flatness_algorithm_lp}
(Flatness algorithm for $\ell_p$-bodies)
There exists an algorithm that for all $\ell_p$-bodies $B_{m,n}^{(p,V)} ( t , \alpha )$ outputs one of the following:
\begin{itemize}
 \item Either it outputs that $B_{m,n}^{(p,V)} ( t , \alpha )$ does not contain an integer vector, or
 \item it outputs that $B_{m,n}^{(p,V)} ( t , \alpha )$ contains an integer vector, or
 \item it outputs a vector $\tilde{d} \in \mathbbm{Z}^m$ and an interval $I_{\mathcal{B}}$ of length at most $4 m^2$ such that $B_{m,n}^{(p,V)} ( t , \alpha )$ contains an integer vector if and only if there exists $k \in \mathbbm{Z} \cap I_{\mathcal{B}}$ such that $B_{m,n}^{(p,V)} ( t , \alpha ) \cap H_{k,\tilde{d}}$ contains an integer vector.
\end{itemize}
The number of arithmetic operations of the algorithm is
 		$$p \cdot ( n \log_2 ( r ) )^{\mathcal{O} ( 1 )} m^{m/(2e)+o(m)}$$
The algorithm runs in polynomial space and each number computed by the algorithm has size at most
 		$r^{p n^{\mathcal{O} ( 1 )}}$,
 		where $r$ is an upper bound on the size of the $\ell_p$-body.
\end{thm}

\noindent Using this algorithm and combining it with the ideas and methods presented in Section \ref{sec_algorithm}, we are able to show that there exists an algorithm that solves the lattice membership problem for the class of $\ell_p$-bodies with $1 < p < \infty$.
Especially, we obtain an algorithm that solves the lattice membership problem for balls generated by an $\ell_p$-norm.\\

\noindent Substantially, the algorithm works in the same way as the general algorithmic framework presented in Section \ref{sec_algorithm}.
The algorithm gets as input a full-dimensional $\ell_p$-body $B_n^{(p,V)} ( t , \alpha )$ and an affine subspace.
During the algorithm, we have to take into account that it is possible, that the flatness algorithm outputs that the $\ell_p$-body does not contain an integer vector.
For the computation of the parameter $N$, which defines a circumscribed $\ell_1$-ball of the $\ell_p$-body, we use the following result, which computes for a given $\ell_p$-body a circumscribed Euclidean ball.
The idea of this construction is that for a given $\ell_p$-body $B_{m,n}^{(p,V)} ( t , \alpha )$, we use Hölder's inequality to construct an ellipsoid, which contains $B_n^{(p,V)} ( t , \alpha )$.
This ellipsoid is contained in an Euclidean ball whose radius is the largest eigenvalue of $V$. Then, we intersect all with the subspace $\bigcap_{i=m+1}^n H_{0,e_i}$.

\begin{lemma} \label{lemma_circumscribed_convex_body}
Let $B_{m,n}^{(p,V)} ( t , \alpha )$ be an $\ell_p$-body given by $V \in \mathbbm{Q}^{n \times n}$ nonsingular,  $t \in \mathbbm{Q}^n$, $\alpha > 0$ and $1 < p < \infty$.
Then $B_{m,n}^{(p,V)} ( t , \alpha )$ is contained in an $m$-dimensional Euclidean ball with radius $\alpha \sqrt{n} \| V \|$.
The center of this ball is given by the orthogonal projection of $t$ onto $\Span ( e_1, \hdots, e_m )$.
\end{lemma}

\begin{proof}
Using Hölder's inequality, we obtain that the $\ell_p$-body $B_n^{(p,V)} ( t , \alpha )$ is contained in the set 
$\{ x \in \mathbbm{R}^n | \| V^{-1} ( x- t ) \|_2 \leq \alpha \sqrt{n} \}$,
which is the open ellipsoid $\alpha \sqrt{n} \star E ( V V^T , t )$.
The circumscribed radius of an ellipsoid is given by the square root of the largest eigenvalue of the matrix defining it.
The square root of the largest eigenvalue of $V V^T$ is the spectral norm of the matrix $V$.
Hence, we obtain that
$$B_n^{(p,V)} ( t , \alpha ) \subseteq B_n^{(2)} ( t , \alpha \sqrt{n} \| V \| ).$$
Obviously, it follows that the $\ell_p$-body $B_{m,n}^{(p,V)} ( t , \alpha )$ is contained in the intersection of the Euclidean ball $B_n^{(2)} ( t , \alpha \sqrt{n} \| V \| )$ with the  hyperspace $\cap_{i=m+1}^n H_{0,e_i}$, which is an $m$-dimensional ball with radius at most 
$\alpha \sqrt{n} \| V \|$.
The center of this ball is given by the orthogonal projection of $t$ onto $\Span ( e_1, \hdots, e_m )$.
\end{proof}

\noindent Using this result, we can define the parameter $N$ as $2 n r \| V \| + 1$, where $\| V \|$ denotes the spectral norm of the matrix $V$.
A detailed description of the algorithm is given in Algorithm \ref{alg_membership_lp}.

\begin{figure}[!ht]
\framebox{
\begin{minipage}[b]{15.5cm} \small
\begin{alg} \label{alg_membership_lp} {\bf Lattice membership algorithm for $\ell_p$-bodies}\\
{\tt
{\bf Input:}
	\begin{itemize}
		\item An $\ell_p$-body $B_{n}^{(p,V)} ( t , \alpha )$ given by a nonsingular matrix $V \in \mathbbm{Q}^{n \times n}$, a vector $t \in \mathbbm{Q}^n$ and a parameter $\alpha > 0$ with size $r_{\mathcal{B}}$ and 
		\item an affine subspace $H := \bigcap_{i=m+1}^n H_{k_i,d_i}$ given by $d_i \in \mathbbm{Z}^n$ linearly independent and $k_i \in \mathbbm{Z}$, $m+1 \leq i \leq n$; alternatively, $H:= \mathbbm{R}^n$.
	\end{itemize}
{\bf Used Subroutines:} Flatness algorithm for $\ell_p$-bodies, replacement procedure.\\[0.25cm]
	{\bf If} $m=0$, check if there exists $z \in \mathbbm{Z}^n \cap H$ satisfying $z \in B_n^{(V)} ( t , \alpha )$.\\
	{\bf Otherwise, }
	\begin{enumerate}
	\item \begin{description}
					\item[If] $m=n$, set $v := 0$ and $\bar{V} := I_n$.
					\item[Otherwise,] compute $v \in \mathbbm{Z}^n \cap H$, a basis $B := [b_1, \hdots, b_m , d_{m+1} , \hdots, d_n] \in \mathbbm{Z}^{n \times n}$ of $\mathbbm{R}^n$.\\
			Compute a lattice basis $\bar{D} \in \mathbbm{Z}^{n \times m}$ of $\mathcal{L} ( B^T ) \cap \bigcap_{i=m+1}^n H_{0,e_i}$.\\
			Set $\hat{D}:=[\bar{D} , e_{m+1}, \hdots, e_n] \in \mathbbm{Z}^n$ and $\bar{V} := \hat{D}^{-1} B^T$.
			\end{description}
	\item \label{alg_mem_lp_step1} Apply the flatness algorithm with $B_{m,n}^{(p,\bar{V} V)} ( \bar{V} ( t - v ) , \alpha )$.
			\begin{description}
 				\item[If] it outputs that $B_{m,n}^{(p,\bar{V} V)} ( \bar{V} (  t - v ) , \alpha )$ does not contain an integer vector, then output that $B_n^{(p,V)} ( t , \alpha ) \cap H$ does not contain an integer vector.
 				\item[If] it outputs that $B_{m,n}^{(p,\bar{V} V)} ( \bar{V} (  t - v ) , \alpha )$ contains an integer vector, then output that $B_n^{(p,V)} ( t , \alpha ) \cap H$ contains an integer vector.
 				\item[Otherwise,] the result is a vector $\tilde{d} \in \mathbbm{Z}^m$ together with an interval $I_{\mathcal{B}}$.
 				\begin{enumerate}
 				 \item Set $d_m := \bar{V}^T ( \tilde{d}^T , 0^{n-m} )^T \in \mathbbm{Z}^n$
 				  			and $N:= 2 n r_{\mathcal{B}} \| V \| + 1$.
 			 	 \item For all $k \in \mathbbm{Z} \cap I_{\mathcal{B}}$,
 								\begin{itemize}
 									\item apply the replacement procedure to the affine subspace $H$, the hyperplane given by $d_m$ and $k + \langle v , d_m \rangle$ and the parameter $N$.\\
 											The result is an index set $J_k$ and an affine subspace $\bigcap_{i \in J_k} H_{\bar{k}_i,\bar{d}_i}$.
 								  \item Apply the membership algorithm to the $\ell_p$-body $B_n^{(p,V)} ( t , \alpha )$ and the affine subspace $H \cap 												\bigcap_{i\in J_k} H_{\bar{k}_i,\bar{d}_i}$.\\
												As a result, we get the information if $B_n^{(p,V)} ( t , \alpha ) \cap H \cap \bigcap_{i\in J_k} H_{\bar{k}_i,\bar{d}_i}$ contains an integer vector or not.
								 \end{itemize}
 								 \item If there exists an index $k$ such that $B_n^{(p,V)} ( t , \alpha ) \cap H \cap \bigcap_{i\in J_k} H_{\bar{k}_i,\bar{d}_i}$ contains an integer vector, output this.
 								 	Otherwise, output that $B_n^{(p,V)} ( t , \alpha ) \cap H$ does not contain an integer vector.
 				\end{enumerate}			
 					\end{description}
\end{enumerate}}
\end{alg}
\end{minipage}}
\end{figure}

\begin{thm}
Let $B_n^{(p,V)} ( t , \alpha )$ be an $\ell_p$-body given by $V \in \mathbbm{Q}^{n \times n}$ nonsingular, $t \in \mathbbm{Q}^n$, $\alpha > 0$ and $1 < p < \infty$ and let $H$ be an affine subspace of dimension $m \leq n$.
Given as input $B_n^{(p,V)} ( t , \alpha )$ and $H$, the membership algorithm for $\ell_p$-bodies, Algorithm \ref{alg_membership_lp}, decides correctly whether $B_n^{(p,V)}( t , \alpha ) \cap H$ contains an integer vector.
The number of arithmetic operations of the algorithm is at most
$p (n \log_2 ( r ) )^{\mathcal{O} ( 1 )} m^{(2+o(1))m}$,
where $r$ is an upper bound on the size of $B_{n}^{(p,V)} ( t , \alpha )$ and the affine subspace $H$.
The algorithm runs in polynomial space and each number computed by the algorithm has size at most $r^{p \cdot n^{\mathcal{O} ( 1 )}}$,
that means bit size at most $p n^{\mathcal{O} ( 1 )} \log_2 ( r )$.
\end{thm}

\begin{proof}
We have seen in Lemma \ref{lemma_circumscribed_convex_body} that $B_n^{(p,V)} ( t , \alpha )$ is contained in an Euclidean ball with radius $\alpha \sqrt{n} \| V \|$.
Hence,
$$B_n^{(p,V)} ( t , \alpha ) \subseteq \bar{B}_n^{(1)} ( t , \alpha n \| V \| ) \subseteq \bar{B}_n^{(1)} ( 0 , n r ( 1 + \|V \| ).$$
By definition of $N$ this shows that $B_n^{(p,V)} ( t , \alpha ) \subseteq \bar{B}_n^{(1)} ( 0 , N-1)$.
Since the transformation $\tau : x \mapsto \bar{V} ( x - v )$ maps $B_n^{(p)} ( t , \alpha ) \cap H$ to the $\ell_p$-body
$B_{m,n}^{(p,\bar{V} V)} ( \bar{V} ( t - v ) , \alpha )$,
it follows from Theorem \ref{alg_mod_mem_result} that the membership algorithm for $\ell_p$-bodies decides correctly whether $B_n^{(p,V)} ( t , \alpha ) \cap H$ contains an integer vector.\\

\noindent It is obvious, that each number computed by the lattice membership algorithm in one reduction step has size at most
$r^{n^{\mathcal{O} ( 1 )}}$.
The recursive instances of the lattice membership algorithm consist of the original $\ell_p$-body $B_n^{(p,V)} ( t , \alpha )$ and a new subspace.
Therefore, we denote the size of the $\ell_p$-body $B_n^{(p,V)} ( t , \alpha )$ separately by $r_{\mathcal{B}}$.
Obviously $r \geq r_{\mathcal{B}}$.\\
According to Lemma \ref{alg_mod_mem_result}, the size of the affine subspaces used for the recursive calls of the algorithm is at most
$\max \{ r_{\mathcal{B}}^{n^{\mathcal{O} ( 1 )}} , r \}$.
Especially, the replacement procedure guarantees that the size of these subspaces depend only on the size of the $\ell_p$-body $r_{\mathcal{B}}$ and not on the size of the affine subspace $H$.
Hence, it follows that
$\max \{ r_{\mathcal{B}}^{n^{\mathcal{O} ( 1 )}} , r \}^{n^{\mathcal{O} ( 1 )}} = r^{n^{\mathcal{O} ( 1 )}}$
is an upper bound on the size of each number computed by the lattice membership algorithm.\\

\noindent Now we give an upper bound on the number of arithmetic operations of the algorithm, denoted by
$T ( m , n , p , r_{\mathcal{B}} , r )$.
For $m = 0$, we obtain that 
$T ( 0 , n , p , r_{\mathcal{B}} , r ) = n^{\mathcal{O} ( 1 )}$.
For $m \geq 1$, the algorithm constructs the bijective affine transformation according to the construction described in Claim \ref{claim_properties_transformation}.
This needs at most $n^{\mathcal{O} ( 1 )}$ arithmetic operations.
According to Theorem \ref{thm_flatness_algorithm_lp}, the number of arithmetic operations of the flatness algorithm $\mathcal{A}_{\mathcal{B}}$ is at most
$$p \cdot ( n \cdot  \log_2 ( \size ( B_{m,n}^{(p,\bar{V} V)} ( \bar{V} ( t - v ) , \alpha ) ) ) )^{\mathcal{O} ( 1 )}
m^{m/(2e) + o(m)}
\leq p ( n \cdot \log_2 ( r ) )^{\mathcal{O} ( 1 )} m^{m/(2e) + o(m)}.$$
The number of arithmetic operations of the replacement procedure is polynomial in the dimension $n$ and
$\log_2 ( N )$.
Since $N$ is at most $4 n r \| V \|_2 \leq r^{n^{\mathcal{O} ( 1 )}}$ this is at most
$n^{\mathcal{O} ( 1 )} \log_2 ( r )^{\mathcal{O} ( 1 )}$.\\
The number of recursive calls of the algorithm is determined by the length of the interval $I_{\mathcal{B}}$, which is at most $4 m^2$, see Theorem \ref{thm_flatness_algorithm_lp}.
This shows that for $m \geq 1$, the number of arithmetic operations of the lattice membership algorithm can be upper bounded using the following recursion,
\begin{equation*} 
T ( m , n , p, r_{\mathcal{B}} , r )
\leq p \cdot ( n \log_2 ( r ) )^{\mathcal{O} ( 1 )} m^{m/(2e) + o(m)} + ( 4 m^{2} + 1 ) \cdot T (  m - 1 , n , p , r_{\mathcal{B}} , \max \{ r , r_{\mathcal{B}}^{n^{\mathcal{O} ( 1 )}} \}).
\end{equation*}
As in the case of polytopes, we observe that the size of the new subspace is the maximum of the size of the input subspace and $r_{\mathcal{B}}^{n^{\mathcal{O} ( 1 )}}$, whereas the size of the $\ell_p$-body does not change. Hence, 
it follows by induction, that for all $m \geq 0$:
$$T ( m , n , p , r_{\mathcal{B}} , r ) \leq p ( n \log_2 ( r ))^{\mathcal{O} ( 1 )} m^{(2 +o(1))m}.$$
\end{proof}

\noindent If we apply the lattice membership algorithm with input an $\ell_p$-body and as subspace the whole vector space, we obtain an algorithm for the lattice membership problem.

\begin{cor}
The membership algorithm for $\ell_p$-bodies, Algorithm \ref{alg_membership_lp}, solves the lattice membership problem for all $\ell_p$-bodies $B_n^{(p,V)} ( t , \alpha )$ correctly in polynomial space.
The number of arithmetic operations of the algorithm is at most
$p \log_2 ( r )^{\mathcal{O} ( 1 )} n^{(2+o(1))n}$,
where $r$ is an upper bound on the size of the $\ell_p$-body.
Each number computed by the algorithm has size at most
$r^{p \cdot  n^{\mathcal{O} ( 1 )}}$, that means bit size at most
$p \cdot n^{\mathcal{O} ( 1 )} \log_2 ( r )$.
\end{cor}
	
	
	\section{An algorithm for computing a flatness direction} \label{sec_flatness_algorithm}

\noindent In this section, we consider constructive versions of so-called flatness theorems.
The fundamental statement of the flatness theorems is that every bounded convex set $\mathcal{C}$ which does not contain an integer vector has at least one direction where it is flat.
This means that there exists a vector $\tilde{d} \in \mathbbm{Z}^n$ such that the number of hyperplanes $H_{k,\tilde{d}}$, $k \in \mathbbm{Z}$, which intersect $\mathcal{C}$ is bounded.
The first result in this area was due to Khinchin, \cite{pp_khinchin}. For an overview about the existing variants see \cite{bk_barvinok}.\\

\noindent To formalize the idea how many hyperplanes intersect a bounded convex set, we use the notion of the width of a convex set $\mathcal{C} \subseteq \mathbbm{R}^n$ along a vector $\tilde{d} \in \mathbbm{R}^n \backslash \{ 0 \}$, which is defined as the number
$$w_{\tilde{d}} ( \mathcal{C} ) := \sup \{ \langle \tilde{d} , x \rangle | x \in \mathcal{C} \}
- \inf \{ \langle \tilde{d} , x \rangle | x \in \mathcal{C} \}.$$
If $\mathcal{C}$ is closed, we have 
$w_{\tilde{d}} ( \mathcal{C} ) = \max \{ \langle \tilde{d} , x \rangle | x \in \mathcal{C} \}
- \min \{ \langle \tilde{d} , x \rangle | x \in \mathcal{C} \}$.
The width of $\mathcal{C}$ is defined as the minimal value $w_{\tilde{d}} ( \mathcal{C})$, where $\tilde{d} \in \mathbbm{Z}^n \backslash \{ 0 \}$. A vector $\tilde{d}$, which minimizes $w_{\tilde{d}} ( \mathcal{C})$ is called a flatness direction of $\mathcal{C}$.\\
The flatness theorems guarantee that the width of every convex body, which does not contain an integer vector, is bounded by a number which depends only on the dimension.\\

\noindent In the following, we will show that for certain bounded convex sets, we are able to compute such a vector $\tilde{d} \in \mathbbm{Z}^n$.
First, we show this result for special convex bodies, ellipsoids.
Then we will generalize this result to polytopes and $\ell_p$-bodies.

\subsection{A flatness algorithm for ellipsoids}

Ellipsoids are special convex sets.
Formally, a set $E \subset \mathbbm{R}^n$ is called an ellipsoid, if there exists a vector $c \in \mathbbm{R}^n$ and a positive definite matrix $D \in \mathbbm{R}^{n \times n}$ such that
$$E = \{ x \in \mathbbm{R}^n | ( x- c )^T D^{-1} ( x - c ) \leq 1 \}.$$
The vector $c$ is called the center of the ellipsoid and we denote by $E ( D , c )$ the ellipsoid given by the matrix $D$ and the vector $c$.
The ellipsoid is uniquely determined by the symmetric positive definite matrix $D$ and the center $c$.\\

\noindent For every symmetric positive definite matrix $D$, there exists a decomposition $D = Q^T \cdot Q$. Such a matrix $Q$ gives us a bijective affine transformation, that maps the Euclidean unit ball to the ellipsoid $E ( D , c )$.
To be precise, a set $E \subset \mathbbm{R}^n$ is an ellipsoid $E = E ( D ,c )$ for a symmetric positive definite matrix $D \in \mathbbm{R}^{n \times n}$ and a vector $c \in \mathbbm{R}^n$ if and only if $E$ is the affine image of the Euclidean unit ball, i.e., 
$$E = Q^T \cdot B_n^{(2)} ( 0 , 1 ) + c,$$
where $D = Q^T \cdot Q$.
Observe that this affine transformation is not uniquely determined, since the decomposition of a symmetric positive definite matrix is not unique.
Nevertheless, this relation between ellipsoids and the Euclidean unit ball is fundamental in the understanding of ellipsoids. Nearly every property of an ellipsoid can be deduced from the corresponding property of the Euclidean unit ball by applying the bijective transformation $Q^T$.\\

\noindent For example, we can show that for an ellipsoid $E  = E ( D , c ) \subseteq \mathbbm{R}^n$ and a vector $\tilde{d} \in \mathbbm{R}^n \backslash \{ 0 \}$, we have 
$\max \{ \langle \tilde{d} , x \rangle | x \in E ( D , c ) \} = \langle \tilde{d} , c \rangle + \sqrt{\tilde{d}^T D \tilde{d}}$
and
$\min \{ \langle \tilde{d}, x \rangle | x \in E ( D , c ) \} = \langle \tilde{d} , c \rangle - \sqrt{\tilde{d}^T D \tilde{d}}$.
Hence, the width of the ellipsoid along $\tilde{d}$ is
$w_{\tilde{d}} ( E ) = 2 \sqrt{\tilde{d}^T D \tilde{d}}$.\\

\noindent The next proposition characterizes a flatness direction and the width of an ellipsoid.
Additionally, we are able to show which hyperplanes of a family of hyperplanes have a non-empty intersection with an ellipsoid.

\begin{prop} \label{prop_flatness_direction}
Let $E = E ( D , c ) \subseteq \mathbbm{R}^n$ be an ellipsoid and $D = Q^T Q$ be an arbitrary decomposition of the matrix $D$.
Then a vector $\tilde{d} \in \mathbbm{Z}^n$ is a flatness direction of the ellipsoid if and only if $Q \tilde{d}$ is a shortest non-zero vector in the lattice $\mathcal{L} ( Q )$.
That means, we have
$$w ( E ) = w_{\tilde{d}} ( E ) = 2 \lambda_1^{(2)} ( \mathcal{L} ( Q ) )$$
and for $d = Q \tilde{d} \in \mathcal{L} ( Q )$ we obtain, 
$$\max \{\langle \tilde{d} , x \rangle | x \in E \} = \langle \tilde{d} , c \rangle + \| d \|_2
 	\mbox{ and }
 	\min \{ \langle \tilde{d} , x \rangle | x \in E\} = \langle \tilde{d} , c \rangle - \| d \|_2.$$
\end{prop}

\noindent We observe, that it follows from this proposition that the width of an ellipsoid can be computed using an arbitrary decomposition of the matrix defining the ellipsoid.

\begin{proof}
As we have seen, the width of an ellipsoid along a vector $\tilde{d} \in \mathbbm{Z}^n \backslash \{0 \}$ is given by
$w_{\tilde{d}} ( E ( D , c ) ) = 2 \sqrt{\tilde{d}^T D \tilde{d}}$.
Hence, for every decomposition $D = Q^T Q$ of the matrix $D$, we have
\begin{equation} \label{eq_flatness_direction_1}
\sqrt{\tilde{d}^T D \tilde{d}} = \sqrt{(Q\tilde{d})^T (Q\tilde{d})} = \| Q \tilde{d} \|_2.
\end{equation}
which shows that the width $w_{\tilde{d}} ( E ( D , c ) )$ is is minimized for $\tilde{d} \in \mathbbm{Z}^n \backslash \{ 0 \}$, if $Q \tilde{d}$ is a shortest non-zero vector in the lattice $\mathcal{L} ( Q )$ generated by the matrix $Q$.
This proves the first statement.
The proof of the other statements follows directly from (\ref{eq_flatness_direction_1}).
\end{proof}

\noindent With this statement, we are able to prove the flatness theorem for ellipsoids using the well-known transference bound for lattices proven by Banaszczyk, see \cite{pp_banaszczykbounds}.
For completeness, the proof appears in the appendix, see Section \ref{sec_appendix_technical_flatness}.

\begin{thm} \label{thm_flatness_ellipsoid} {\bf (Flatness Theorem for Ellipsoids)}
Let $E \subset \mathbbm{R}^n$ be an ellipsoid.
If the width of the ellipsoid is at least $n$,
$w ( E ) \geq n$,
then the ellipsoid contains an integer vector.
\end{thm}

\noindent Combining Proposition \ref{prop_flatness_direction} together with the flatness theorem for ellipsoids we obtain a flatness algorithm for ellipsoids:
Given an ellipsoid, we compute its width and a corresponding flatness direction $\tilde{d} \in \mathbbm{Z}^n$ by computing a shortest non-zero lattice vector in the lattice $\mathcal{L} ( Q )$.
If the width is larger than $n$, we output that the ellipsoid contains an integer vector.
Otherwise, we output the flatness direction $\tilde{d} \in \mathbbm{Z}^n$ together with an interval
$I_E = [ \min \{ \langle \tilde{d}, x \rangle | x \in E \} , \max \{ \langle \tilde{d} , x \rangle | x \in E \}]$.
In this case, the interval $I_E$ satisfies that $E$ contains an integer vector if and only if there exists $k \in \mathbbm{Z} \cap I_E$ such that $E \cap H_{k,\tilde{d}}$ contains an integer vector.
To compute a shortest non-zero lattice vector in $\mathcal{L} ( Q )$, we cannot use the single exponential time algorithm by Micciancio and Voulgaris \cite{pp_MV10} since it requires exponential space.
Instead we use Kannan's polynomial space algorithm \cite{pp_kannan} with its improvement by Hanrot and Stehlé \cite{pp_hast07}.
For a complete description of the algorithm see Algorithm \ref{alg_flatness_ellipsoid}.

\begin{figure}[ht]
\framebox{
\begin{minipage}[b]{15.5cm} \small
\begin{alg} \label{alg_flatness_ellipsoid} {\bf Flatness Algorithm for Ellipsoids}\\
{\tt 
{\bf Input:}
	Ellipsoid $E:= E ( D , c )$
\begin{itemize}
	\item Compute a decomposition $D = Q^T Q$ of the matrix $D$.
	\item Compute a shortest non-zero lattice vector $d \in \mathcal{L} ( Q )$ using Kannan's algorithm for $\svp$.
		Let $\tilde{d} := Q^{-1} d \in \mathbbm{Z}^n$.
	\item Set $w := 2 \| d \|_2$.\\
	 {\bf If} $w \geq n$, output that $E$ contains an integer vector.\\
	 {\bf Otherwise} output $\tilde{d} \in \mathbbm{Z}^n$ together with
	 	$k_{\min} := \lceil \langle \tilde{d} , c \rangle - \| d \|_2 \rceil$ and
	 	$k_{\max} := \lfloor \langle \tilde{d} , c \rangle + \| d \|_2 \rfloor$.
\end{itemize}
}
\end{alg}
\end{minipage}}
\end{figure}

\begin{prop} \label{prop_flatness_alg_result}
Given an ellipsoid $E \subseteq \mathbbm{R}^n$, the flatness algorithm for ellipsoids, Algorithm \ref{alg_flatness_ellipsoid}, outputs one of the following:
Either it outputs that $E$ contains an integer vector or it outputs a vector $\tilde{d} \in \mathbbm{Z}^n$ and an interval $I_E$ of length at most $n$ such that $E$ contains an integer vector if and only if there exists $k \in \mathbbm{Z} \cap I_E$ such that $E \cap H_{k,\tilde{d}}$ contains an integer vector.
The number of arithmetic operations of the algorithm is $n^{n/(2e)+o(n)}$,
the algorithm has polynomial space complexity, 
and each number computed by the algorithm has size at most
$r^{n^{\mathcal{O} ( 1 )}}$, where $r$ is an upper bound on the size of $E$.
\end{prop}

\begin{proof}
The correctness of the algorithm follows directly from Proposition \ref{prop_flatness_direction} and Theorem \ref{thm_flatness_ellipsoid}.\\
To see that  the size of each number computed by the algorithm is at most $r^{n^{\mathcal{O} ( 1 )}}$, we observe that the length of the flatness direction is at most
$$\| \tilde{d} \|_2 \leq n^{(n+2)/2} \size ( D )^{(n+1)/2} \leq n^{(n+2)/2} r^{(n+1)/2},$$
see Lemma \ref{lemma_appendix_change_repsize_mem_alg} in the appendix.\\
Hence, the only thing we need to take care of is that the numbers $k_{\min}, k_{\max} \in \mathbbm{Z}$ does not become too large.
By definition, they are at most
$\langle \tilde{d}, c \rangle + 2 \lambda_1^{(2)} ( \mathcal{L} ( Q ) )$.
Since the width of the ellipsoid $E$ is at most $n$, we obtain using the Cauchy-Schwarz inequality
$$k \leq n^{(n+2)/2} r^{(n+1)/2} \| c \|_2 + n \leq r^{n^{\mathcal{O} ( 1 )}}.$$
The number of arithmetic operations is dominated by the number of arithmetic operations needed to compute a shortest non-zero lattice vector in $\mathcal{L} ( Q )$ using Kannan's $\svp$-algorithm. By the analysis of Hanrot and Stehlé \cite{pp_hast07} the number of arithmetic operations can be bounded $n^{n/(2e)+o(n)}$ times some factor polynomial in the input size.
\end{proof}

\noindent In general, we are not able to compute a flatness direction of a bounded convex set. But we can approximate the convex set by an ellipsoid and in this way obtain a direction in which the convex set is relatively flat.
By the approximation of a bounded convex set with an ellipsoid, we understand an ellipsoid which is contained in $\mathcal{C}$.
The approximation factor is that factor, whereby we need to scale the ellipsoid such that the scaled ellipsoid contains the convex set.
By scaling an ellipsoid with a positive factor $r > 0$ we understand the ellipsoid obtained from $E$ by scaling it from its center by the factor $r$.
We denote this as $r \star E$.
Formally, if $E = E ( D , c )$, then
$r \star E := r \cdot E ( D , 0 ) + c$.
Alternatively, such a scaled ellipsoid can be characterized as follows:
For $r > 0$, we have
$r \star E = E ( r^2 \cdot D , c )$.
We call an ellipsoid, which approximates a bounded convex set an approximate Löwner-John ellipsoid.

\begin{definition}
Let $\mathcal{C} \subset \mathbbm{R}^n$ be a full-dimensional bounded convex set and $0 < \gamma < 1$.
An ellipsoid $E$ satisfying 
$E \subseteq \mathcal{C} \subseteq ( 1 / \gamma ) \star E$
is called $1 / \gamma$-approximate Löwner-John ellipsoid of $\mathcal{C}$.
We call $1 / \gamma$ the approximation factor of the Löwner-John ellipsoid.
\end{definition}

\noindent If we are able to compute approximate Löwner-John ellipsoids for a class of bounded convex sets, then there exists a flatness algorithm for this class:
Given an approximate Löwner-John ellipsoid $E$ of a full-dimensional bounded convex set $\mathcal{C}$, we can compute the width and a corresponding flatness direction $\tilde{d} \in \mathbbm{Z}^n$ of the ellipsoid.
If this width is larger than $n$, the ellipsoid and therefore the convex set $\mathcal{C}$ contain an integer vector.
Otherwise, we observe that the width of the circumscribed ellipsoid $(1 / \gamma) \star E$ is at most
$(1 / \gamma) \cdot w ( E ) \leq n / \gamma$ and that $\tilde{d} \in \mathbbm{Z}^n$ is also a flatness direction of the circumscribed ellipsoid. 
Hence, the vector $\tilde{d} \in \mathbbm{Z}^n$ satisfies that
$$\left| \max \big\{ \langle \tilde{d} , x \rangle | x \in (1 / \gamma) \star E \big\}
- \min \big\{ \langle \tilde{d} , x \rangle | x \in (1 / \gamma) \star E \big\} \right|
\leq n / \gamma.$$
Since the convex set $\mathcal{C}$ is contained in $( 1 / \gamma ) \star E$, the vector $\tilde{d}$ also satisfies that every hyperplane $H_{k,\tilde{d}}$ which has a non-empty intersection with $\mathcal{C}$ satisfies that
$\min \{ \langle \tilde{d} , x \rangle | x \in ( 1 / \gamma ) \star E \}
\leq k \leq 
\max \{ \langle \tilde{d} , x \rangle | x \in ( 1 / \gamma ) \star E \}$.\\

\noindent Now, we use this idea to obtain flatness algorithms for polytopes and $\ell_p$-bodies.
In what follows, we show that for these convex sets there exist polynomial-time algorithms that compute approximate Löwner-John ellipsoids.
The algorithms are modifications of the famous ellipsoid method from Shor and Khachiyan and are based on an idea due to Yudin and Nemirovskii and Goffin
see \cite{pp_shor_ellipsoid}, \cite{pp_ellipsoid_method}, \cite{pp_yn} and \cite{pp_goffin}.
In general, these algorithms are known as shallow cut ellipsoid methods.
For more information about the ellipsoid method and its modifications see \cite{bk_gls}, \cite{bk_algorithmictheory}, \cite{bk_schrijver} or \cite{bk_kv}.

		\subsection{A flatness algorithm for polytopes} \label{sec_flatness_algorithm_polytopes}

\noindent For polytopes, there exists a polynomial algorithm that computes an approximate Löwner-John ellipsoid. The following result is due to \cite{bk_schrijver}.

\begin{thm} 
There exists an algorithm that given a full-dimensional polytope
$P \subseteq \mathbbm{R}^n$
computes a $2n$-approximate Löwner-John ellipsoid for $P$
in time polynomially bounded by $n$ and the size of $P$.
\end{thm}

\noindent For us, this result is not enough. We need more precise statements about the running time and the size of the ellipsoid, as stated in the following.
A complete description of the algorithm together with a proof of the following theorem appears in the full version of this paper.

\begin{thm} \label{statement_LJ_polyhedra}
(Rounding method for polytopes) There exists an algorithm that given a full-dimensional polytope
$P \subseteq \mathbbm{R}^n$
computes a $2n$-approximate Löwner-John ellipsoid given by a symmetric positive definite matrix $D \in \mathbbm{Q}^{n \times n}$ and a vector $c \in \mathbbm{Q}^n$.
The number of arithmetic operations is
$( n s )^{\mathcal{O} ( 1 )} \log_2 ( r )$,
where $s$ is the number of constraints defining the polytope and $r$ is an upper bound on its size.
Each number computed by the algorithm has size at most
$2^{\mathcal{O} ( n^4 )} r^{\mathcal{O} ( n )}$.
\end{thm}

\noindent Using this algorithm, we can adapt the idea described above and we obtain a flatness algorithm for polytopes.
A complete description of the algorithm is given in Algorithm \ref{alg_flatness_polytopes}.

\begin{figure}[ht]
\framebox{
\begin{minipage}[b]{15.5cm} \small
\begin{alg} \label{alg_flatness_polytopes} {\bf Flatness Algorithm for Polytopes}\\
{\tt 
{\bf Input:}
	A full-dimensional polytope $P \subseteq \mathbbm{R}^n$ given by $A \in \mathbbm{Z}^{s \times n}$ and $\beta \in \mathbbm{Z}^s$.\\
{\bf Used Subroutine:} rounding method for polytopes, Kannan's $\svp$ algorithm
\begin{itemize}
	\item Apply the rounding method for polytopes with input $P$.\\
		The result is $D \in \mathbbm{Q}^{n \times n}$ symmetric positive definite and $c \in \mathbbm{Q}^n$.\\
		Compute a decomposition $D = Q^T Q$ of the matrix $D$.
	\item Compute a shortest non-zero lattice vector $d \in \mathcal{L} ( Q )$ using Kannan's $\svp$ algorithm.
		Let $\tilde{d} := Q^{-1} d \in \mathbbm{Z}^n$.
	\item Set $w := 2 \| d \|_2$.\\
		 {\bf If} $w \geq n$, output that $P$ contains an integer vector.\\
	 	 {\bf Otherwise} output $\tilde{d} \in \mathbbm{Z}^n$ together with
	 	$k_{\min} := \lceil \langle \tilde{d} , c \rangle - 2n \| d \|_2 \rceil$ and
	 	$k_{\max} := \lfloor \langle \tilde{d} , c \rangle + 2n \| d \|_2 \rfloor$.
\end{itemize}
}
\end{alg}
\end{minipage}}
\end{figure}

\begin{thm} (Theorem \ref{thm_flatness_algorithm_polytopes} restated)\\
Given a full-dimensional polytope $P \subseteq \mathbbm{R}^n$, the flatness algorithm for polytopes outputs one of the following:
\begin{itemize}
 \item Either it outputs that $P$ contains an integer vector or
 \item it outputs a vector $\tilde{d} \in \mathbbm{Z}^n$ and an interval $I_P$ of length at most $2n^2$ such that $P$ contains an integer vector if and only if there exists $k \in \mathbbm{Z} \cap I_P$ such that $P \cap H_{k,\tilde{d}}$ contains an integer vector.
\end{itemize}
The number of arithmetic operations of the algorithm is 
$s^{\mathcal{O} ( 1 )} \log_2 ( r ) 2^{\mathcal{O} ( n )}$
and each number computed by the algorithm has size at most $r^{n^{\mathcal{O} ( 1 )}}$, where $r$ is an upper bound on the size of the polytope and $s$ is the number of constraints defining the polytope.
\end{thm}

\begin{proof}
As we have seen in Proposition \ref{prop_flatness_direction}, the value $w$ is the width of an approximate Löwner-John ellipsoid $E$ of the polytope $P$.
The algorithm computes this value and distinguishes between two cases:\\
If $w \geq n$, it is guaranteed by the flatness theorem that $E$ and therefore $P$ contain an integer vector, see Theorem \ref{thm_flatness_ellipsoid} \\
If $w < n$, the algorithm outputs a vector $\tilde{d} \in \mathbbm{Z}^n$ together with an interval  $I_P = [k_{\min}, k_{\max}]$. This interval contains all integers $k \in \mathbbm{Z}$ such that the hyperplane $H_{k,\tilde{d}}$ intersects the ellipsoid $2n \star E ( D , c )$.
It follows from $P \subseteq 2n \star E ( D , c )$ that it also contains all integers $k \in \mathbbm{Z}$ such that the hyperplane $H_{k,\tilde{d}}$ intersects the polytope $P$.
Since the width of the ellipsoid $E ( D , c )$ along the vector $\tilde{d}$ is at most $n$, the width of the ellipsoid $2n \star E ( D , c )$ along $\tilde{d}$, which is an upper bound on the length of the interval, is at most $2n^2$.\\

\noindent According to Theorem \ref{statement_LJ_polyhedra}, the size of an approximate Löwner-John ellipsoid of the polytope $P$ computed by the rounding method is at most
$$2^{\mathcal{O} ( n^4 )} \size ( P )^{\mathcal{O} ( n )} \leq 2^{\mathcal{O} ( n^4 )} r^{\mathcal{O} ( n )}.$$
In fact, the flatness algorithm for polytopes combines the flatness algorithm for ellipsoids for the ellipsoid $E ( D , c )$ and the ellipsoid $2 n \star E ( D , c )$.
Hence, it follows from Proposition \ref{prop_flatness_alg_result} that the size of each number computed by the algorithm is at most
$$\left( 2^{\mathcal{O} ( n^4 )} r^{\mathcal{O} ( n )} \right)^{n^{\mathcal{O} ( 1 )}} = r^{n^{\mathcal{O} ( 1 )}}.$$
The number of arithmetic operations is dominated by the number of arithmetic operations of the rounding method for polytopes and Kannan's $\svp$ algorithm.
Hence, we obtain
$$\left( n \cdot s \right)^{\mathcal{O} ( 1 )} \log_2 ( r ) + n^{n/(2e)+o(n)}
= s^{\mathcal{O} ( 1 )} \log_2 ( r ) n^{n/(2e)+o(n)}.$$
\end{proof}
	
\subsection{A flatness algorithm for $\ell_p$-Bodies} \label{sec_flatness_algorithm_lp}
	
\noindent To obtain a flatness algorithm for $\ell_p$-bodies we need to be able to compute approximate Löwner-John ellipsoids for $\ell_p$-bodies.

		\subsubsection{Computation of Löwner-John Ellipsoids for $\ell_p$-Bodies, $1 < p < \infty$} \label{subsub_LJ_LpBodies}

The algorithm that computes approximate Löwner-John ellipsoids for $\ell_p$-bodies is based on a variant of the shallow cut ellipsoid method due to \cite{bk_gls}.
This algorithm computes in polynomial time a $\sqrt{n} ( n + 1 )$-approximate Löwner-John ellipsoid for any well-formed convex body given by a separation oracle.
That means, we assume that the algorithm has access to an oracle that decides for a given vector whether it is contained in the convex set or not.
If the vector is not contained in the convex set, it provides a hyperplane that strictly separates this vector from the convex body.
To obtain an approximation factor linear in $n$ rather than of the norm $n^{3/2}$ as in \cite{bk_gls}, we combine the result of \cite{bk_gls} with an idea of \cite{pp_hk10} and \cite{pp_kochol_94}.
Unlike the original approach of \cite{bk_gls}, this approach leads to an algorithm whose number of arithmetic operations is single exponential in the dimension, but in our situation this is irrelevant.

\begin{thm} \label{statement_approx_LJ_oracle}
There exists an algorithm which satisfies the following properties:
Given a full-dimensional bounded convex set $\mathcal{C} \subseteq \mathbbm{R}^n$ by a separation oracle together with $r_{in} , R_{\out} > 0$ and $c_{\out} \in \mathbbm{R}^n$ such that
$\mathcal{C} \subseteq \bar{B}_n^{(2)} ( c_{\out} , R_{\out} )$
and 
$\vol_n ( \mathcal{C} ) \geq r_{in}^n \vol_n ( B_n^{(2)} ( 0 , 1 ) )$,
the algorithm computes a $4 n$-approximate Löwner-John ellipsoid.
The number of arithmetic operations of the algorithm is dominated by the number of calls to the oracle, which is at most
$\log_2 ( R_{\out} / r_{in} )^{\mathcal{O} ( 1 )} 2^{\mathcal{O} ( n )}.$
The algorithm requires polynomial space and each number computed by the algorithm has size at most
$2^{\mathcal{O} ( n^4 )} ( R_{\out} / r_{in} )^{\mathcal{O} ( 1 )}$.
\end{thm}

\noindent To apply Theorem \ref{statement_approx_LJ_oracle} to $\ell_p$-bodies, we need to realize a separation oracle for this class of convex sets.
Additionally, we need to determine parameters $R_{out}$, $r_{in} > 0$ and a vector $c_{out} \in \mathbbm{R}^m$ such that an $\ell_p$-body $B_{m,n}^{(p,V)} ( t , \alpha )$ is contained in an $\ell_2$-ball with radius $R_{out}$ centered at $c_{out}$ and such that the volume of an $\ell_p$-body $B_{m,n}^{(p,V)} ( t , \alpha )$ is at least $r_{in}^m$ times the volume of the $m$-dimensional Euclidean unit ball.
We have already seen in Lemma \ref{lemma_circumscribed_convex_body} how we can construct for an $\ell_p$-body a circumscribed $\ell_2$-ball.
Now, we prove a lower bound on the volume of an $\ell_p$-body provided that it contains an integer vector. The lower bound depends on the shape of the convex set, that means on the parameters defining it, and on the radius of a circumscribed Euclidean ball.
For the proof of the lower bound, we consider a special representation of the $\ell_p$-body.
If we consider the following convex function,
\begin{equation} \label{eq_def_function_lower_bound_volume}
F : \mathbbm{R}^m \to \mathbbm{R}, ~ x \mapsto \alpha_d^p \| V^{-1} ( ( x^T , 0^{n-m} )^T - t ) \|_p^p - \alpha_n^p,
\end{equation}
where $V \in \mathbbm{R}^{n \times n}$ nonsingular, $t \in \mathbbm{R}^n$ and $\alpha_n, \alpha_d \in \mathbbm{N}$.
Then we have 
$B_{m,n}^{(p,V)} ( t , \alpha )
= \{ x \in \mathbbm{R}^m | F ( x ) < 0 \}$
with $\alpha := \alpha_n / \alpha_d$.\\

\noindent To illustrate the main idea of the proof, which is due to Heinz \cite{pp_heinz05}, we imagine that the function $F$ is in addition differentiable and  that we know an upper bound $M$ on the length of its gradient $\nabla F ( x )$, $x \in \mathbbm{R}^m$, i. e., $\| \nabla F ( x ) \|_2 \leq M$ for all $x \in \mathbbm{R}^m$.
Further, we assume that we know some parameter $\epsilon > 0$ such that there exists a vector $\hat{x} \in \mathbbm{R}^n$ with $F ( \hat{x} ) \leq - \epsilon < 0$.\\

\noindent Since for every convex function the first-order Taylor approximation is a global underestimator of the function (first-order convexity condition), we obtain for all $x \in \mathbbm{R}^m$ that
$$F( \hat{x} ) \geq F ( x ) + \nabla F ( x )^T ( \hat{x} - x ).$$
Using the Cauchy-Schwarz inequality this yields the upper bound
$$F( x ) \leq F( \hat{x} ) + \nabla F ( x )^T ( x - \hat{x} )
\leq - \epsilon + M \| x - \hat{x} \|_2.$$
Hence, if a vector $x \in \mathbbm{R}^m$ satisfies $\| x - \hat{x} \|_2 \leq \epsilon / M$, then $F( x ) < 0$ and it is contained in the set $B_{m,n}^{(p,V)} ( t , \alpha )$.
This shows that $B_{m,n}^{(p,V)} ( t , \alpha )$
contains an Euclidean ball with radius $\epsilon / M$ centered around $\hat{x}$ and that the volume of
$B_{m,n}^{(p,V)} ( t , \alpha )$ is at least $( \epsilon / M )^m \vol_m ( B_m^{(2)} ( 0 , 1 ) )$.\\

\noindent For the function $F$ defined in (\ref{eq_def_function_lower_bound_volume}) we can compute such a parameter $\epsilon$ since we can show that there exists an integer $K$ such that for all $x \in \mathbbm{Z}^m$ there exists an integer $K' \leq K$ such that $K \cdot F ( x ) \in \mathbbm{Z}$.
That means for every integer vector $x \in \mathbbm{Z}^m$, $F ( x )$ is a rational number with denominator at most $K$.
Hence, if $B_{m,n}^{(p,V)} ( t , \alpha )$ contains an integer vector $\hat{x} \in \mathbbm{Z}^m$, then $F ( \hat{x} ) \leq - 1 / K < 0$.
In the following claim, we give an upper bound on the number $K$. A proof of it appears in the appendix.

\begin{claim} \label{claim_upper_bound_size}
Let $F : \mathbbm{R}^m \to \mathbbm{R}$ be a function defined as in (\ref{eq_def_function_lower_bound_volume}) given by a non-singular matrix $V \in \mathbbm{Q}^{n \times n}$, a vector $t \in \mathbbm{Q}^n$ and $\alpha_n, \alpha_d \in \mathbbm{N}$.
Let $S$ be an upper bound on the size of $V^{-1}$, $t$, $\alpha_n$ and $\alpha_d$.
Then, there exists an integer $K \leq S^{2n^2 p}$ such that
$K \cdot F ( x ) \in \mathbbm{Z}$ for all $x \in \mathbbm{Z}^m$.
\end{claim}

\noindent Now, the main remaining problem is that the function $F$ defined in (\ref{eq_def_function_lower_bound_volume}) is not differentiable.
Hence, we cannot apply the idea of Heinz directly.
We need to modify the idea described above and work with the subgradient instead of the gradient. 
We start with a short overview about subgradients.

\begin{definition}
Let $f : \mathbbm{R}^n \to \mathbbm{R}$ be a convex function and $x \in \mathbbm{R}^n$.
A vector $g \in \mathbbm{R}^n$ is called a subgradient of $f$ at $x$, if the following holds,
\begin{equation} \label{eq_subgradient_inequality}
f ( z ) \geq f ( x ) + \langle g , z - x \rangle \mbox{ for all } z \in \mathbbm{R}^n.
\end{equation}
\end{definition}

\noindent The inequality (\ref{eq_subgradient_inequality}) is called subgradient inequality.
Geometrically, this inequality means that the graph of the affine function
$z \mapsto f ( x ) + \langle g , z - x \rangle$
is a supporting hyperplane of the epigraph of $f$ at $( x ,f ( x ) )$.
If $f$ is differentiable, then the subgradient is unique and it is simply the gradient of $f$ at $x$.
For a more detailed introduction into subgradients see \cite{bk_rockafellar} and \cite{bk_polya}.\\

\noindent Using the subgradient inequality, we can prove a lower bound on the volume of the set $B_{m,n}^{(p,V)} ( t , \alpha )$
under the assumption that for all $R > 0$ and $y \in B_m^{(2)} ( 0 , R )$, the length of a corresponding subgradient is bounded.

\begin{lemma} \label{lemma_lower_bound_volume_general}
Let $B_{m,n}^{(p,V)} ( t , \alpha )$ be an $\ell_p$-body given by $V \in \mathbbm{Q}^{n \times n}$ nonsingular, $t \in \mathbbm{Q}^n$, $\alpha = \alpha_n / \alpha_d > 0$ and $1 < p < \infty$.
Let $F : \mathbbm{R}^m \to \mathbbm{R}$ be a function defined as in (\ref{eq_def_function_lower_bound_volume}).
Let $S$ be an upper bound on the size of $B_{m,n}^{(p,V)} ( t , \alpha )$.
Let $R > 0$ such that $B_{m,n}^{(p,V)} ( t , \alpha )$ is contained in an Euclidean ball with radius $R$ centered at the origin.
Assume that there exists $M \in \mathbbm{R}$ such that the following holds:
For all $y \in B_m^{(2)} ( 0 , R )$ there exists a subgradient $g \in \mathbbm{R}^m$ of $F$ at $y$ which satisfies $\| g \|_2 \leq M$.\\
If $B_{m,n}^{(p,V)} ( t , \alpha )$ contains an integer vector $\hat{x} \in \mathbbm{Z}^m$, then
$$\vol_m ( B_{m,n}^{(p,V)} ( t , \alpha ) )
> ( S^{2n^2p} M )^{-m} \cdot \vol_m ( B_m^{(2)} ( 0 , 1 ) ).$$
\end{lemma}

\begin{proof}
Let $g \in \mathbbm{R}^m \backslash \{ 0 \}$ be a subgradient of $F$ at the vector $y \in B_m^{(2)} ( 0 , R )$ which satisfies $\| g \|_2 \leq M$.
Then it follows from the subgradient inequality (\ref{eq_subgradient_inequality}) for $\hat{x} \in \mathbbm{Z}^m$ that
$$F ( \hat{x} ) \geq F ( y ) + \langle g , \hat{x} - y \rangle.$$
As we have seen in Claim \ref{claim_upper_bound_size}, $F ( \hat{x} )$ is a rational number with denominator at most $S^{2n^2p}$.
Since $F ( \hat{x} ) < 0$, we obtain using the Cauchy-Schwarz inequality
$$F ( y ) \leq F ( \hat{x} ) + \langle g , y - \hat{x} \rangle
\leq - S^{-2n^2p} + \| g \|_2 \cdot \| y - \hat{x} \|_2
\leq - S^{-2n^2p} + M \| y - \hat{x} \|_2$$
which shows that every vector $y \in B_m^{(2)} ( 0 , R )$ with $\| y - \hat{x} \|_2 \leq S^{-2n^2p} / M$ satisfies $F ( y ) < 0$ and is contained in $B_{m,n}^{(p,V)} ( t , \alpha )$.\\
Hence, the $\ell_p$-body $B_{m,n}^{(p,V)} ( t , \alpha )$ contains a ball with radius
$(S^{2n^2p} M )^{-1}$ centered at $\hat{x}$ and the claimed lower bound for the volume follows directly.
\end{proof}

\noindent To obtain a lower bound on the volume of $B_{m,n}^{(p,V)} ( t , \alpha )$ we need to compute for every vector $y \in B_m^{(2)} ( 0 , R )$ an upper bound on the length of a corresponding subgradient of $F$ which depends only on the parameter $R$.
For this, we need to develop an explicit expression of a subgradient of the function $F$ defined in (\ref{eq_def_function_lower_bound_volume}).
We start with the computation of the subgradient of the following simple function.

\begin{lemma} \label{lemma_subgradient_l_p_potenziert}
Let $y \in \mathbbm{R}^n$ and $1 < p < \infty$.
Then a subgradient $g \in \mathbbm{R}^n$ of the function
$F_p : \mathbbm{R}^n \to \mathbbm{R}$,
$x  \mapsto \sum_{i=1}^n | x_i |^p$
at the vector $y$ is given by $g = (g_1, \hdots, g_n)^T$, where
$g_i := \sign (y_i) \cdot | y_i |^{p-1}$.
\end{lemma}

\noindent The proof consists of showing that the vector $g$ satisfies the subgradient inequality. For completeness, it appears in the appendix, see Section \ref{sec_appendix_technical_flatness}.
\noindent To compute the subgradient of the function $F$ defined as in (\ref{eq_def_function_lower_bound_volume}), we combine this result with the following lemma, which shows how the subgradient changes if we consider an affine transformation of the variables or the function.

\begin{lemma} \label{lemma_subgradient_affine_transformation}
Let $f : \mathbbm{R}^n \to \mathbbm{R}$ be a convex function.
\begin{itemize}
	\item  Let $h_1: \mathbbm{R}^n \to \mathbbm{R}$ defined by
$h_1 ( x ) := f ( A x + b )$,
where $A \in \mathbbm{R}^{n \times n}$ is a nonsingular matrix and $b \in \mathbbm{R}^n$.
Let $g_1 \in \mathbbm{R}^n$ be a subgradient of $f$ at the vector $A y + \beta$.
Then, the vector $A^T g_1$ is a subgradient of $h_1$ at the vector $y$.
	\item Let $h_2 : \mathbbm{R}^n \to \mathbbm{R}$ defined by
		$h_2 ( x ) := a \cdot f ( x ) + \beta$, where $a \in \mathbbm{R} \backslash \{ 0 \}$ and $\beta \in \mathbbm{R}$.
		Let $g_2 \in \mathbbm{R}^n$ be a subgradient of $f$ at the vector $y \in \mathbbm{R}^n$.
		Then $a g_2$ is a subgradient of $h_2$ at the vector $y$.
\end{itemize}
\end{lemma}

\noindent The proof of this lemma is straightforward since we need to show that the vectors $A^T g_1$  and $a g_2$ satisfy the subgradient inequality.
If we apply this result with $A = V^{-1}$, $\beta = - V^{-1} t$ and $a = \alpha_d^p$, $b = \alpha_n^p$, and restrict the subgradient to its first $m$ coordinates we are able to give an explicit expression of the subgradient of the function $F$.

\begin{lemma} \label{lemma_subgradient_Fp}
For $ m , n \in \mathbbm{N}$, $m \leq n$, a subgradient at the vector $y \in \mathbbm{R}^m$ of the function
$F : \mathbbm{R}^m \to \mathbbm{R}$,
$x \mapsto \alpha_d^p \| V^{-1} ( ( x^T , 0^{n-m} )^T - t ) \|_p^p - \alpha_n^p$,
where $V \in \mathbbm{R}^{n \times n}$ is nonsingular, $t \in \mathbbm{R}^n$ and $1 < p < \infty$,
is given by the vector $\alpha_d^p g \in \mathbbm{R}^m$ defined by
$g = ( V^{-1})^T \bar{g}_{\{1, \hdots, m \}}$,
where $\bar{g} \in \mathbbm{R}^n$ is defined by 
$\bar{g}_{i}
= \sign ( [V^{-1} ( y - t ) ]_i ) \cdot | [V^{-1} ( y - t ) ]_i |^p$.
\end{lemma}

\noindent Using this explicit expression of the subgradient, we are able to give an upper bound on its length. The proof of the following lemma appears in the appendix, see Section \ref{sec_appendix_technical_flatness}.

\begin{lemma} \label{statement_upper_bound_coordinate_subgradient}
For $ m , n \in \mathbbm{N}$, $m \leq n$, a subgradient at the vector $y \in \mathbbm{R}^m$ of the function
$F : \mathbbm{R}^m \to \mathbbm{R}$,
$x \mapsto \alpha_d^p \| V^{-1} ( ( x^T , 0^{n-m} )^T - t ) \|_p^p - \alpha_n^p$,
where $V \in \mathbbm{R}^{n \times n}$ is nonsingular, $t \in \mathbbm{R}^n$ and $1 < p < \infty$,
is given by the vector $\alpha_d^p g \in \mathbbm{R}^m$ defined by
$g = ( V^{-1})^T \bar{g}_{\{1, \hdots, m \}}$,
where $\bar{g} \in \mathbbm{R}^n$ is defined by 
$\bar{g}_{i}
= \sign ( [V^{-1} ( y - t ) ]_i ) \cdot | [V^{-1} ( y - t ) ]_i |^p$.\\
If $y \in \bar{B}_m^{(2)} ( 0 , R ) \subseteq \mathbbm{R}^m$, then 
$$\| \alpha_d^p g \|_2 \leq m \cdot \left( \alpha_d n S^2 R \right)^{p+1},$$
where $S$ is an upper bound on the size of $V^{-1}$ and $t$.
\end{lemma}

\noindent Using this upper bound together with Lemma \ref{lemma_lower_bound_volume_general} and the upper bound of a radius of a circumscribed Euclidean ball, we get the following lower bound on the volume of $B_{m,n}^{(p,V)} ( t , \alpha )$.

\begin{lemma} \label{cor_lower_bound_volume_final}
Let $B_{m,n}^{(p,V)} ( t , \alpha )$ be an $\ell_p$-body, where $t \in \mathbbm{R}^n$, $V \in \mathbbm{Q}^{n \times n}$ is nonsingular, $\alpha \in \mathbbm{Q}^+$ and $1 < p < \infty$.
If $B_{m,n}^{(p,V)} ( t , \alpha )$ contains an integer vector, then its volume is at least
\begin{align*}
\vol_m ( B_{m,n}^{(p,V)} ( t , \alpha ) )
\geq 
\left( S^{2(n^2+2)} m^2 n^2 \| V \| \right)^{-m(p+1)} \cdot \vol_m ( B_m^{(2)} ( 0 , 1 ) ),
\end{align*}
where $S$ is an upper bound on the size of $V^{-1}$ and $t$.
\end{lemma}

\begin{proof}
It follows from Lemma \ref{lemma_circumscribed_convex_body}, that
the convex body $B_{m,n}^{(p,V)} ( t , \alpha )$ is contained in an Euclidean ball centered at the origin, whose radius is at most
$\alpha \sqrt{n} \| V \| + m S$.
Hence, if we choose
$R := \alpha \sqrt{n} m \| V \| \cdot S$,
then the Euclidean ball $B_m^{(2)} ( 0 , R )$
contains $B_{m,n}^{(p,V)} ( t , \alpha )$.
Combining this with the result from Lemma \ref{statement_upper_bound_coordinate_subgradient}, the statement follows directly from Lemma \ref{lemma_lower_bound_volume_general}.
\end{proof}

\noindent To compute approximate Löwner-John ellipsoids, we need to be able to compute separating hyperplanes. The following result gives a relation between this problem and the computation of subgradients.

\begin{lemma}
Let $f : \mathbbm{R}^n \to \mathbbm{R}$ be a convex function and 
$\mathcal{C}_{\alpha} := \{ x \in \mathbbm{R}^n | f ( x ) < \alpha \}$
for some $\alpha > 0$ be the corresponding convex body.
Let $y \in \mathbbm{R}^n$ with $y \not\in \mathcal{C}_{\alpha}$.
Then, any subgradient $g \in \mathbbm{R}^n$ of $f$ at $y$ defines a hyperplane that separates $y$ from $\mathcal{C}_{\alpha}$,
i.e., $\langle g , x \rangle \leq \langle g , y \rangle$
for all $x \in \mathcal{C}_{\alpha}$.
\end{lemma}

\noindent The proof of this lemma follows directly from the subgradient inequality (\ref{eq_subgradient_inequality}).
Hence, Lemma \ref{lemma_subgradient_Fp} yields to an efficient realization of a separation oracle for an $\ell_p$-body.
Together with the results from Lemma \ref{lemma_circumscribed_convex_body} and Lemma \ref{cor_lower_bound_volume_final}, this shows that we can use the algorithm from Theorem \ref{statement_approx_LJ_oracle}. To obtain an algorithm that computes an approximate Löwner-John ellipsoid for $\ell_p$-bodies.

\begin{thm} \label{statement_LJ_lpbody} (Rounding method for $\ell_p$-bodies)
Let $B_{m,n}^{(p,V)} ( t , \alpha )$ be an $\ell_p$-body given by a nonsingular matrix $V \in \mathbbm{Q}^{n \times n}$, $t \in \mathbbm{Q}^n$, $\alpha > 0$ and $1 < p < \infty$.
Then, there exists an algorithm that given such an $\ell_p$-body outputs one of the following:
\begin{itemize}
			\item Either it outputs that $B_{m,n}^{(p,V)} ( t , \alpha )$ does not contain an integer vector, or
			\item it outputs a $4 m$-approximate Löwner-John ellipsoid in form of a positive definite matrix $D \in \mathbbm{Q}^{m \times m}$ and a vector $c \in \mathbbm{Q}^m$.
				In this case, the size of the ellipsoid is at most $2^{\mathcal{O} ( n^4 )} r^{\mathcal{O} ( n^3 p )}$.
		\end{itemize}
The algorithm uses polynomial space and its number of arithmetic operations is at most $p ( n \log_2 ( r ) )^{\mathcal{O} ( 1 )} 2^{\mathcal{O} ( m )}$.
Here, $r$ is an upper bound on the size of the $\ell_p$-body.
\end{thm}

		\subsubsection{Description and Analysis of the flatness algorithm for $\ell_p$-bodies}

\noindent Using this result, we obtain a flatness algorithm for $\ell_p$-bodies in the same way as we obtain the flatness algorithm for polytopes, see Algorithm \ref{alg_flatness_lp}.

\begin{figure}[!ht]
\framebox{
\begin{minipage}[b]{15.5cm} \small
\begin{alg} \label{alg_flatness_lp} {\bf Flatness Algorithm for $\ell_p$-bodies}\\
{\tt 
{\bf Input:}
	An $\ell_p$-body $B_{m,n}^{(p,V)} ( t , \alpha )$, where $V \in \mathbbm{Q}^{n \times n}$ nonsingular, $t \in \mathbbm{Q}^n$, $\alpha > 0$, $1 < p < \infty$.\\
{\bf Used Subroutine:} Rounding method for $\ell_p$-bodies, Kannan's $\svp$ algorithm.\\[0.25cm]
Apply the rounding method for $\ell_p$-bodies with input $B_{m,n}^{(p,V)} ( t , \alpha )$.\\
{\bf If} it outputs that $B_{m,n}^{(p,V)} ( t , \alpha )$ does not contain an integer vector, then output this.\\
{\bf Otherwise,} the result is $D \in \mathbbm{Q}^{m \times m}$ symmetric positive definite and $c \in \mathbbm{Q}^{m}$.
	\begin{itemize}	
	\item Compute a decomposition $D = Q^T Q$ of the matrix $D$.
	\item Compute a shortest lattice vector $d \in \mathcal{L} ( Q ) \backslash \{ 0 \}$ using Kannan's $\svp$-algorithm.
		Let $\tilde{d} := Q^{-1} d \in \mathbbm{Z}^m$.
	\item Set $w := 2 \| d \|_2$.\\
	{\bf If} $w \geq m$, output that $B_{m,n}^{(p,V)} ( t , \alpha )$ contains an integer vector.\\
	{\bf Otherwise} output $\tilde{d} \in \mathbbm{Z}^n$ together with
	 	$k_{\min} := \lceil \langle \tilde{d} , c \rangle - 4 m  \| d \|_2 \rceil$ and
	 	$k_{\max} := \lfloor \langle \tilde{d} , c \rangle + 4 m \| d \|_2 \rfloor$.
\end{itemize}

}
\end{alg}
\end{minipage}}
\end{figure}

\begin{thm} (Theorem \ref{thm_flatness_algorithm_lp} restated)
Given as input an $\ell_p$-body $B_{m,n}^{(p,V)} ( t , \alpha)$, the flatness algorithm for $\ell_p$-bodies outputs one of the following:
\begin{itemize}
 \item Either it outputs that $B_{m,n}^{(p,V)} ( t , \alpha )$ does not contain an integer vector, or
 \item it outputs that $B_{m,n}^{(p,V)} ( t , \alpha )$ contains an integer vector, or
 \item it outputs a vector $\tilde{d} \in \mathbbm{Z}^m$ and an interval $I_{\mathcal{B}}$ of length at most $4m^2$ such that $B_{m,n}^{(p,V)} ( t , \alpha )$ contains an integer vector if and only if there exists $k \in \mathbbm{Z} \cap I_{\mathcal{B}}$ such that $B_{m,n}^{(p,V)} ( t , \alpha ) \cap H_{k,\tilde{d}}$ contains an integer vector.
\end{itemize}
The number of arithmetic operations of the algorithm is
 		$$p \cdot ( n \log_2 ( r ) )^{\mathcal{O} ( 1 )} m^{m/(2e)+o(m)}$$
 		and each number computed by the algorithm has size at most
 		$r^{p n^{\mathcal{O} ( 1 )}}$,
 		where $r$ is an upper bound on the size of the $\ell_p$-body.
\end{thm}

\begin{proof}
Obviously, we can assume that the rounding method computes an approximate Löwner-John ellipsoid.
For this ellipsoid $E$, the algorithm computes a flatness direction as well as its width $w$, see Proposition \ref{prop_flatness_direction}.
Then the algorithm distinguishes between two cases:\\
If $w \geq m$, it is guaranteed by the flatness theorem that $E$ and therefore $B_{m,n}^{(p,V)} ( t , \alpha )$ contain an integer vector, see Theorem \ref{thm_flatness_ellipsoid}.\\
Otherwise, we have $w < m$ and the algorithm outputs a vector $\tilde{d} \in \mathbbm{Z}^m$ together with an interval $I_{\mathcal{B}} = [ k_{\min} , k_{\max} ]$.
This interval contains all integers $k \in \mathbbm{Z}$ such that the hyperplane $H_{k,\tilde{d}}$ intersects the ellipsoid $4m \star E ( D , c )$.
Since $B_{m,n}^{(p,V)} ( t , \alpha ) \subseteq 4 m \star E ( D , c )$, 
this interval contains also all integers $k \in \mathbbm{Z}$ such that $H_{k,\tilde{d}}$ intersects the $\ell_p$-body.
Since $w ( E ) < m$, the length of the interval is at most $4 m^2$.\\

\noindent According to Theorem \ref{statement_LJ_lpbody}, the size of an approximate Löwner-John ellipsoid computed by the rounding method is at most
$2^{\mathcal{O} ( n^4 )} r^{\mathcal{O} ( n^2 p )}$.
Since the flatness algorithm is a combination of the flatness algorithm for ellipsoids applied with the inscribed ellipsoid $E ( D , c )$ and the circumscribed ellipsoid $4 m^{2} \star E$, it follows from Proposition \ref{prop_flatness_alg_result} that the size of each number computed by the algorithm is at most
$$\left ( 2^{\mathcal{O} ( n^4 )} r^{\mathcal{O} (n^2 p )} \right)^{n^{\mathcal{O} ( 1 )}}
= r^{p \cdot n^{\mathcal{O} ( 1 )}}.$$
The number of arithmetic operations is dominated by the number of arithmetic operations of the rounding method for $\ell_p$-bodies and by the number of arithmetic operations required by Kannan's $\svp$ algorithm.
Hence, it is upper bounded by
$$p ( n \cdot \log_2 ( r ) )^{\mathcal{O} ( 1 )} 2^{\mathcal{O} ( m ) } + m^{m/(2e)+o(m)}
= p ( n \log_2 ( r ) )^{\mathcal{O} ( 1 )} m^{m/(2e)+o(m)}.$$
\end{proof}

\noindent Using this rounding method, we obtain a flatness algorithm for $\ell_p$-bodies.
Hence, our assumptions made in Section \ref{sec_LM_alg_lp} are satisfied and there exists a deterministic algorithm that solves $\lmp$ for balls generated by an $\ell_p$-norm, $1 < p < \infty$.
As stated in Theorem \ref{statement_cvp_self_lp} this leads to a deterministic algorithm that solves $\cvp$ with respect to an $\ell_p$-norm with $1 < p < \infty$.
In the same way, we obtain a deterministic algorithm that solves $\lmp$ for polytopes and a deterministic algorithm for $\cvp$ for all polyhedral norms, e.g. the $\ell_1$-norm and the $\ell_{\infty}$-norm.

\paragraph{Acknowledgment.} We thank Friedrich Eisenbrand for several stimulating discussions that greatly benefited the paper.


\bibliography{../../../../bib}
\bibliographystyle{alpha}

\appendix

\section{Appendix}

\subsection{Projection for non-Euclidean norms} \label{appendix_counterexample}

\noindent If we consider other norms than the Euclidean norm, we had to differ between two types of norms on the vector space $\mathbbm{R}^n$: The norms which are induced by a scalar product or inner product and the norms which are not.
A norm on $\mathbbm{R}^n$ is induced by a scalar product, if for all $x \in \mathbbm{R}^n$, $\| x \| = \sqrt{\langle x , x \rangle}$, where $\langle \cdot , \cdot \rangle : \mathbbm{R}^n \times \mathbbm{R}^n \to \mathbbm{R}$ denotes a scalar product.
Particularly, all $\ell_p$-norms with $p \not= 2$ are not induced by a scalar product. 
The norms on $\mathbbm{R}^n$ which are induced by an scalar product are exactly that norms whose unit ball is an ellipsoid. For such norms the solution of the closest vector problem can be easily reduced to the solution of the closest vector problem with respect to the Euclidean norm using the fact that each ellipsoid is the image of the Euclidean unit ball under a bijective affine transformation.\\
If the norm is not induced by a scalar product it does not seem to be possible to use projections for algorithmic solution of the closest vector problem.\\

\noindent We start with a description of the situation and show how we can use projections if we consider the closest lattice vector problem with respect to a norm induced by a scalar product.
Then, we give a counterexample why dimension reduction it does not seem to work for norms which are not induced by a scalar product.\\

\noindent In the following, we assume that we are given a vector space $\Span ( b_1, \hdots, b_n )$, where $b_1, \hdots, b_n$ are linearly independent,
a target vector $t \in \Span ( b_1, \hdots, b_n )$ and
a lattice $L = \mathcal{L} ( b_1, \hdots, b_{n-1} )$.
We are searching for the lattice vector in $L$, which is closest to $t$, see Figure \ref{cvp_situation} for an illustration.\\

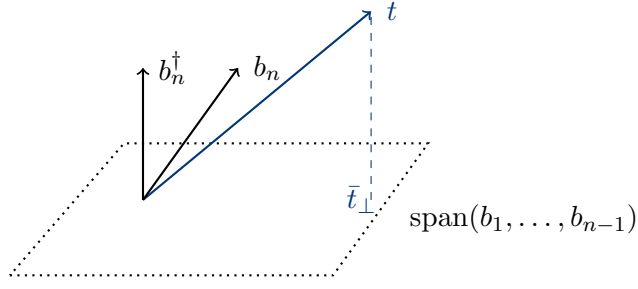
\begin{figure} \label{cvp_situation}
\begin{center}
\begin{tikzpicture}[scale=0.25]
	\draw[thick,dotted] (0,0) -- (17,0) -- (22,7) -- (6,7) -- (0,0);
	\draw[thick,->,color=blau]	(7,4) -- (19,14);
	\draw[dashed,color=blau]		(19,4) -- (19,14);
	\draw[thick,->]	(7,4) -- (12,11);
	\draw[thick,->] (7,4) -- (7,11);
	\draw (22.5,3) node[right=-0.5cm] {$\Span ( b_1, \hdots, b_{n-1} )$};
	\draw (7,11) node[right=2pt] {$b_n^\dagger$};
	\draw (12,11) node[right=2pt] {$b_n$};
	\draw (18.5,4) node[] {\color{blau} $\bar{t}_{\perp}$};
	\draw (19,14) node[right=2pt] {\color{blau} $t$};
\end{tikzpicture}
\end{center}
\caption{{\bf Projection in a subspace.}
The vector $t$ lies in $\Span ( b_1 , \hdots, b_n )$. The vector $\bar{t}_{\perp}$ denotes the orthogonal projection of $t$ in $\Span ( b_1, \hdots, b_{n-1} )$.}
\end{figure}

\noindent In this situation, the distance between the target vector and the lattice can be arbitrarily large.
In order to handle this problem, we consider the orthogonal projection of $t$ in $\Span ( b_1, \hdots, b_{n-1} )$, which is given by
\begin{equation} \label{eq_orthogonal_projection}
\bar{t}_{\perp} = t - \pi_n ( t ) = t - \frac{\langle t , b_n^\dag \rangle}{\langle b_n^\dag, b_n^\dag \rangle} b_n^\dag,
\end{equation}
where $b_n^\dag$ is a vector orthogonal to $\Span ( b_1, \hdots, b_{n-1} )$ with respect to the corresponding scalar product, for example the $n$-th Gram-Schmidt-vector of the basis $[b_1, \hdots, b_n]$.\\

\noindent If we are searching for a solution of the closest vector problem with respect to a norm that is induced by a scalar product, it is easy to prove the following:

\begin{prop} \label{projection_euclid}
Let $\| \cdot \|$ be a norm on $\mathbbm{R}^n$ induced by a scalar product. 
The vector $v \in L$ is a closest lattice vector to $t$
if and only if
$v$ is a lattice vector in $L$ closest to the projection $t'$ of $t$ in $\Span ( b_1, \hdots, b_{n-1} )$.
\end{prop}

\begin{proof}
 Let $y \in L \subset \Span ( b_1, \hdots, b_{n-1} )$ be the closest lattice vector to $t'$.
 Since the norm is induced by a scalar product, we have 
 $\| t - y \|^2 = \langle t - y , t - y \rangle$,
 where
 $t - y = \bar{t}_{\perp} + \langle t , b_n^\dagger \rangle / \langle b_n^\dagger, b_n^\dagger \rangle b_n^\dagger$.
 Hence, 
	\begin{equation*}
	\| t - y \|^2 	
 	= \langle \bar{t}_{\perp} - y , \bar{t}_{\perp} - y \rangle + 2 \frac{\langle t , b_n^\dagger \rangle}{\langle b_n^\dagger, b_n^\dagger \rangle} \langle b_n^\dagger, \bar{t}_{\perp} - y \rangle + \langle \frac{\langle t , b_n^\dagger \rangle}{\langle b_n^\dagger, b_n^\dagger \rangle} b_n^\dagger, \frac{\langle t , b_n^\dagger \rangle}{\langle b_n^\dagger, b_n^\dagger \rangle} b_n^\dagger \rangle.
 	\end{equation*}
 	Since $b_n^\dagger$ is orthogonal to $\bar{t}_{\perp} - y \in \Span ( L )$, we get
 	$$\| t - y \|^2 = \| \bar{t}_{\perp} - y \|^2 + \| \frac{\langle t , b_n^\dagger \rangle}{\langle b_n^\dagger, b_n^\dagger \rangle} b_n^\dagger \|^2,$$
 	where the term
 	$\| \langle t , b_n^\dag \rangle / \langle b_n^\dag, b_n^\dag \rangle b_n^\dag \|^2$ is independent of the choice of $y$.
 	Hence, we see that
 	$\| t - v \|$ is minimized over $v \in L$
 	if and only if $\| \bar{t}_{\perp} - v \|$ is minimized over $L$.
\end{proof}

\noindent To show that this statement is not true if the norm is not induced by a scalar product, we give a counterexample.
Additionally, we show that this statement is not true, if we consider the corresponding norm projection instead of the orthogonal projection:
As the norm projection of a vector in a subspace we understand that vector in the subspace with minimal distance with respect to the corresponding norm, i. e., we consider the vector in $\Span ( L )$ which is closest to $t$ with respect to the corresponding norm:
\begin{equation} \label{eq_arbitrary_projection}
 \bar{t}_{\min} \in \Span ( L ) \mbox{ with } \min_{ \bar{x} \in \Span ( L )} \| t - \bar{t} \|.
\end{equation}
Mangasarian gave an explicit closed form for this projection, (see \cite{pp_Mangasarian99}).
\noindent If we consider a norm induced by a scalar product then the norm projection and the orthogonal projection coincide. 
Additionally, we need to observe that if the norm is not strictly convex, then the norm projection might not be uniquely determined!\\

\noindent The following counterexample considers the closest vector problem with respect to the $\ell_1$-norm, which is very descriptive.
But there exists also counterexamples for norms, which are not strictly convex, for example for the $\ell_3$-norm.
They will appear in the full version of this paper.\\

\noindent We consider the $\mathbbm{R}^2$ and the lattice spanned by the vector
$b_1 = ( 4 , 7 )^T \in \mathbbm{R}^2$.
Additionally, we consider the target vector
$t = ( 0 , 5 ) ^T \in \mathbbm{R}^2$,
which is not contained in the subspace $\Span ( b_1 )$.
We are searching for a lattice vector in $\mathcal{L} (b_1)$ which is closest to $t$ with respect to the $\ell_1$-norm,
see Figure \ref{zeichnung_gegenbsp_l1} for an illustration.
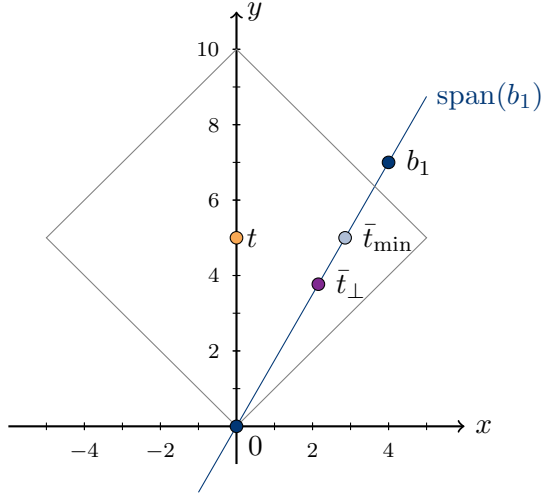
\begin{figure} \label{zeichnung_gegenbsp_l1}
\begin{center}
\begin{tikzpicture}[scale=0.5]
	\draw[->,thick] (-6,0) -- (6,0) node[right] {$x$};
	\draw[->,thick] (0,-1) -- (0,11) node[right] {$y$};
	\foreach \x in {-5,-4,-3,-2,-1,1,2,3,4,5}
		\draw (\x,-.1) -- (\x,.1); 
	\foreach \x in {-4,-2,2,4}
		\draw (\x,-.1) -- (\x,.1) node[below=4pt] {$\scriptstyle\x$};
	\foreach \y in {1,2,3,4,5,6,7,8,9,10}
		\draw (-.1,\y) -- (.1,\y); 
	\foreach \y in {2,4,6,8,10}
		\draw (-.1,\y) -- (.1,\y) node[left=4pt] {$\scriptstyle\y$};
	\draw[blau,domain=-1:5] plot (\x,{1.75*\x}) node[right=0.0cm] {$\Span ( b_1 )$};
	\draw[grey] (0,0) -- (5,5) -- (0,10) -- (-5,5) -- (0,0);
	\draw[fill=blau] (0,0) circle (1ex) node[right=0.25cm,below] {$0$};
	\draw[fill=blau] (4,7) circle (1ex) node[right=0.1cm] {$b_1$};
	\draw[fill=orange!75] (0,5) circle (1ex) node[right] {$t$};
	\draw[fill=hellblau] ({20/7},5) circle (1ex) node[right=0.1cm] {$\bar{t}_{\min}$};
	\draw[fill=violett] ({28/13},{49/13}) circle (1ex) node[right=0.1cm] {$\bar{t}_{\perp}$};
\end{tikzpicture}
\end{center}
\caption{{\bf Counterexample for projections with respect to the $\ell_1$-norm.}
We consider the lattice spanned by the vector $b_1$, together with the target vector $t$.
The vector $\bar{t}_{\perp}$ is the orthogonal projection of $t$ in $\Span ( b_1 )$, $\bar{t}_{\min}$ is the $\ell_1$-projection.}
\end{figure}

\begin{claim}
The vector $v = 0$ is the closest lattice vector to $t$ in $\mathcal{L} (b_1)$ with respect to the $\ell_1$-norm.
\end{claim}

\begin{proof}
Every lattice vector $v \in \mathcal{L} (b_1)$ is of the form
$v = v_1 b_1 = (4 v_1 , 7 v_1 )^T$
with $v_1 \in \mathbbm{Z}$.
With this representation, the distance between $t$ and a lattice vector is given by 
$\| t - v_1 b_1 \|_1 =  4 | v_1 | + | 5 - 7 v_1 |$
and it becomes minimal over $\mathbbm{Z}$, if $v_1 = 0$.
\end{proof}

\noindent Now we consider the orthogonal projection $\bar{t}_{\perp}$ of $t$ in $\Span (b_1)$ with respect to the Euclidean norm, (see (\ref{eq_orthogonal_projection})). The vector $(-7,4)^T$ is orthogonal to $b_1$. Hence, $\bar{t}_{\perp}$ is given by
\begin{equation*}
\bar{t}_{\perp} 
= t - \frac{\langle t , \left( \begin{array}{c} -7 \\ 4 \end{array} \right) \rangle}{\langle \left( \begin{array}{c} -7 \\ 4 \end{array} \right), \left( \begin{array}{c} -7 \\ 4 \end{array} \right) \rangle} \left( \begin{array}{c} -7 \\ 4 \end{array} \right)
= \frac{7}{13} \left( \begin{array}{c} 4 \\ 7 \end{array} \right).
\end{equation*}
Now, we are searching for the closest lattice vector to $\bar{t}_{\perp}$ with respect to the $\ell_1$-norm. Obviously, in a lattice of rank 1, we get the closest lattice vector by rounding. Hence,

\begin{claim}
The vector $b_1$ is a closest lattice vector to $\bar{t}_{\perp} = 7 / 13 \cdot ( 4 , 7 )^T$ in $L = \mathcal{L} ( b_1 )$ with respect to the $\ell_1$-norm.
\end{claim}

\noindent Hence, this is an example where the lattice vector which is closest to $t$ is not the lattice vector which is closest to the orthogonal projection of $t$ in the lattice.
Now we consider the vector $\bar{t}_{\min} \in \Span (b_1 )$ which is closest to $t$ with respect to the $\ell_1$-norm, as defined in (\ref{eq_arbitrary_projection}).\\

\noindent The $\ell_1$-projection of a point $t$ onto a subspace $S$ depends of the orientation of the subspace. In $\mathbbm{R}^2$, when the angle $\theta$ is different from $\pi/4$, the projection is unique but directly along the $y$-axis or the $x$-axis. When $\theta = \pi/4$, the projection is a segment and it includes the points along both unit directions.\\

\noindent In our example, we obtain 
$$\min_{\bar{x} \in \Span (b_1)} \| t - \bar{x} \|_1 = \min_{x_1 \in \mathbbm{R}} \| \left( \begin{array}{c} 0 \\5 \end{array} \right) - x_1 \left( \begin{array}{c} 4 \\ 7 \end{array} \right) \|_1 = \min_{x_1 \in \mathbbm{R}} 4 |x_1| + | 5 - 7 x_1 |.$$
This value becomes minimal, if $x_1 = 5 / 7$. Hence,
$\bar{t}_{\min} = \frac 5 7 \cdot ( 4 , 7 )^T$.
Obviously, we get

\begin{claim} \label{claim_1}
The vector $b_1 = ( 4 ,7 )^T$ is the closest lattice vector to $\bar{t}_{\min}$ in $\mathcal{L} (b_1)$ with respect to the $\ell_1$-norm.
\end{claim}

\noindent Hence, this is additionally an example, where a lattice vector that is closest to $t$ is not closest to the target vector $\bar{t}_{\min}$ which is the $\ell_1$-projection of $t$ in $\Span ( L )$.\\

\subsection{Technical Stuff}

\noindent To prove the statements in the appendix, we use some facts about the change of the representation size under basic matrix operations. We state them in the following. For a proof of results of this type see for example \cite{bk_gls} or \cite{bk_schrijver}.

\begin{claim} \label{facts_representation_size} $ $
\begin{itemize}
 \item Let $A \in \mathbbm{Q}^{n \times n}$. Then,
 		$\size ( A^{-1} ) \leq n^{n/2} \size ( A )^{n(n-1)}$.
 \item Let $A \in \mathbbm{Z}^{m \times n}$, $x,y \in \mathbbm{Z}^n$. Then,
 		$\size ( x + y ) \leq \size ( x ) + \size ( y )$ and
 		$\size ( A x ) \leq n \cdot \size ( A ) \cdot \size ( x )$.
 \item Let $A \in \mathbbm{Q}^{m \times n}$, $x,y \in \mathbbm{Q}^n$. Then,
 		$\size ( x + y ) \leq 2 \size ( x ) \cdot \size ( y )$ and
 			$\size ( A x ) \leq ( 2 \cdot \size ( A ) \cdot \size ( x ) )^n$.
 \item Let $P \subseteq \mathbbm{R}^n$ be a full-dimensional polytope centered about the origin and $x \in \mathbbm{R}^n$.
 			Then, $\| x \|_P \leq n \cdot \size ( P ) \cdot \size ( x )$.
\end{itemize}
\end{claim}

	 \subsection{Selfreducibility of the Closest Vector Problem} \label{sec_appendix_selfreducibility}

\noindent Since lattices are discrete objects, for the proof of Theorem \ref{statement_cvp_self_lp} it is enough to show that there exists a polynomial reduction from the closest vector problem to the decisional vector problem.
In the decisional closest vector problem, we are given a lattice $L$ and some target vector $t \in \Span ( L )$ together with some parameter $r > 0$ and we need to decided whether the distance from the target vector to the lattice is at most $r$.\\

\noindent The reduction from the closest vector problem to the decisional closest vector problem uses as a intermediate problem the optimization variant of the closest vector problem.
In the optimization closest vector problem ($\optcvp^{(\| \cdot \|)}$), we are given a lattice $L$ and some target vector $t \in \Span ( L )$ and we are asked to compute the minimal distance from this target vector to the lattice.\\
The reduction from the closest vector problem to the decisional closest vector problem consists of a reduction from the optimization closest vector problem to the decisional closest vector problem, which we will present in Section \ref{subsec_reduction_opt_to_dec}, and a reduction from the optimization closest vector problem to the closest vector problem, (see Section \ref{subsec_reduction_search_to_opt}).

\paragraph{Reduction of the Optimization Closest Vector Problem to the Decisional Closest Vector Problem} \label{subsec_reduction_opt_to_dec}

\noindent The reduction from the optimization variant to the decision variant of the closest vector problem is based on binary search.
\noindent This binary search is performed on the set of all possible values which can be achieved by the norm of an integer vector, if the norm lies in some certain interval.
Hence, we need to ensure that we are able to enumerate all these values
and we need an upper bound on the cardinality of such a set - depending on the size of the interval.
To guarantee all that, we consider special norms which we call enumerable.
In general, we call a function enumerable, if it maps every integer vector to a discrete enumerable set.

\begin{definition} \label{def_enumerable_function}
A function $f: \mathbbm{R}^n \to \mathbbm{R}_0$ is called $(k,K)$-enumerable for parameters $k,K \in \mathbbm{N}$, or simply enumerable, if there exists $\tilde{K} \in \mathbbm{N}$, $\tilde{K} \leq K$, such that
$$\tilde{K} \cdot f(x)^k  \in \mathbbm{N}_0 \mbox{ for all } x \in \mathbbm{Z}^n.$$
\end{definition}

\noindent Obviously, every $\ell_p$-norm, $1 \leq p \leq \infty$, is $(k,1)$-enumerable with $k = p$ for $1 \leq p < \infty$ and $k = 1$ for $p = \infty$.
Later, we will show that also all polyhedral norms are enumerable.

\begin{prop} \label{prop_cvp_dec_to_opt}
Let $\| \cdot\|$ be a $(k,K)$-enumerable norm on $\mathbbm{R}^n$.
Assume that there exists an algorithm $\mathcal{A}_{\dec}$ that for all lattices $\mathcal{L} ( B' ) \subset \mathbbm{Z}^n$ of rank $m$, all target vectors $t' \in \Span ( B' ) \cap \mathbbm{Z}^n$ and all $r > 0$ solves the decisional closest vector problem in time $T_{m,n}^{(\| \cdot \|)} ( S', r )$, 
where $S'$ is an upper bound on the size of the basis $B'$ and the target vector $t'$.\\
Then there exists an algorithm that solves the optimization closest vector problem for all lattices $L = \mathcal{L} ( B ) \subseteq \mathbbm{Z}^n$, $B = [b_1, \hdots, b_m]$, and all target vectors $t \in \Span ( L ) \cap \mathbbm{Z}^n$ in time
$$\mathcal{O} \left( k \cdot \log_2 ( \frac m 2 \cdot \max_j \| b_j \| ) + \log_2 ( K ) \right) \cdot n^{\mathcal{O} ( 1 )} \cdot T_{m,n}^{(\| \cdot \|)} ( S , \frac m 2  \cdot \max_j \| b_j \| ),$$
where $S$ is an upper bound on the representation size of the basis $B$ and the target vector $t$.
\end{prop}

\begin{proof}
Let $B = [b_1, \hdots, b_m] \subseteq \mathbbm{Z}^{n \times m}$ be a lattice basis of the lattice $L$ and $t \in \Span ( L ) \cap \mathbbm{Z}^n$ a target vector.
Without loss of generality, we assume that $t \not\in L$, i.e., $\mu^{( \| \cdot \|)} ( t , L ) > 0$.\\
As an upper bound for the distance between $t$ and the lattice, we can choose 
$$R := \frac m 2  \max \{ \| b_j \| | 1 \leq j \leq m \},$$
(see \cite{bk_cassel}), since $t \in \Span ( L )$.
We have $L \subseteq \mathbbm{Z}^n$ and $t \in \mathbbm{Z}^n$.
Hence, the distance vector of $t$ and its closest lattice vector is an integer vector.
Using that $\| \cdot \|$ is a $(k,K)$-enumerable norm, we obtain that the distance is of the form
$$\mu^{( \| \cdot \|)} ( t , L ) = \sqrt[k]{\frac p q}, \mbox{ where } p , q \in \mathbbm{N} \mbox{ with } \gcd ( p , q ) = 1 \mbox{ and } 1 \leq q \leq K.$$
Hence, we are able to perform a binary search using the algorithm $\mathcal{A}_{\dec}$ to find $\mu^{( \| \cdot \|)} ( t , L )$. 
The number of calls to the algorithm $\mathcal{A}_{\dec}$ is at most $\mathcal{O} ( \log_2 ( R^k \cdot K^2 ) )$, since we are finished if the length of the current interval is less than $1 / K^2$.
As a consequence, the running time to solve $\optcvp$ is
$$\mathcal{O} \left( k \cdot \log_2 ( R ) + 2 \log_2 ( K ) \right) \cdot n^{\mathcal{O} (1)} \cdot 
T_{m,n}^{(\| \cdot \|)} ( S , R ).$$
\end{proof}

\paragraph{Reduction of the Closest Vector Problem to the Optimization Closest Vector Problem} \label{subsec_reduction_search_to_opt}
 
\noindent Now, we will present a reduction from the search variant to the optimization variant of the closest vector problem. The running time of this reduction depends on the knowledge of non-decreasing functions $c, C : \mathbbm{N} \to \mathbbm{R}^{>0}$ such that $c(n) \cdot \| x \|_2 \leq \| x \| \leq C(n) \cdot \| x \|_2$ for all $x \in \mathbbm{R}^n$.
In what follows, if the parameter $n$ is obvious by the context, we will omit it and we will write $c$ or $C$ instead of $c ( n )$ or $C ( n )$.
Geometrically, these functions can be interpreted as the radius of an inscribed or circumscribed Euclidean ball of the unit ball of the norm $\| \cdot \|$.

\begin{prop} \label{prop_cvp_search_to_opt}
Let $\| \cdot\|$ be a norm on $\mathbbm{R}^n$ and $c,C : \mathbbm{N} \to \mathbbm{R}^{>0}$ be non-decreasing functions such that $c(n) \cdot \| x \|_2 \leq \|x \| \leq C(n) \cdot \| x \|_2$ for all $x \in \mathbbm{R}^n$.\\
Assume that there exists an algorithm $\mathcal{A}_{\opt}$, that for all lattices $\mathcal{L} ( B' ) \subset \mathbbm{Z}^n$ of rank $m$ and all target vectors $t' \in \Span ( B' ) \cap \mathbbm{Z}^n$ solves $\optcvp^{(\| \cdot\|)}$ in time
$T_{m,n}^{(\| \cdot \|)} ( S')$, 
where $S'$ is an upper bound on the size of the basis $B'$ and the target vector $t'$.\\
Then there exists an algorithm $\mathcal{A}'$ that solves the closest vector problem for all lattices $\mathcal{L} ( B ) \subset \mathbbm{Z}^n$ of rank $m$ and target vectors $t \in \Span ( B ) \cap \mathbbm{Z}^n$ in time
$$2m \cdot \log_2 \left( m \sqrt{n} \cdot  (C \cdot c^{-1}) \cdot S \right) \cdot
T_{m,n}^{(\| \cdot \|)} ( 16 m^3 n \cdot ( C \cdot c^{-1} )^2 S^3 ),$$
where $S$ is an upper bound on the size of the basis $B$ and the target vector $t$.
\end{prop}

\noindent The idea of the reduction is to modify the lattice basis such that the lattice becomes thinner and thinner. Simultaneously, the distance between the target vector and the lattice remains the same. We repeat this until the lattice is so thin such that we are able to compute the closest lattice vector in polynomial time.\\
Before proving Proposition \ref{prop_cvp_search_to_opt}, we will show that a closest lattice vector can be computed efficiently if the lattice is thin enough.
That means, we consider special $\cvp^{(\| \cdot\|)}$-instances, where the distance between the target vector and the lattice is small compared with the minimum distance of the lattice.

\begin{lemma} \label{lemma_solving_cvp_thin_lattice}
Let $i \in \mathbbm{N}$. Let $B \subseteq 2^i \mathbbm{Z}^{n \times m}$ be a lattice basis of rank $m$ and $t \in \Span ( B ) \cap \mathbbm{Z}^n$ a target vector.
Let $\| \cdot\|$ be a norm on $\mathbbm{R}^n$ and $c : \mathbbm{N} \to \mathbbm{R}^{>0}$ be a non-decreasing function  such that $\| x \| \geq c(n) \cdot \| x \|_2$ for all $x \in \mathbbm{R}^n$.
Let
$$i > 1 + \log_2 ( \mu^{(\| . \|)} ( t , \mathcal{L} ( B ) ) ) - \log_2 ( c ).$$
If we consider the following representation of $t = \sum_{j=1}^n \beta_j 2^i e_j$, then the vector $v := \sum_{j=1}^n \lfloor \beta_j \rceil 2^i e_j$ is the closest lattice vector to $t$ in $\mathcal{L} ( B )$ with respect to norm $\| \cdot\|$.
Especially, the closest lattice vector to $t$ in $\mathcal{L} ( B )$ with respect to the norm $\| \cdot\|$ can be computed in polynomial time.
\end{lemma}

\begin{proof}
To prove the lemma, we consider the lattice $2^i \mathbbm{Z}^n$ and show that there exists exactly one vector, whose distance to $t$ is at most $\mu^{( \| \cdot \|)} ( t , \mathcal{L} ( B ) )$, namely the vector $v$.
Since $\mathcal{L} ( B )$ is a sublattice of $2^i \mathbbm{Z}^n$, the statement follows.
We show that $v \in 2^i \mathbbm{Z}^n$ is the only vector in $2^i \mathbbm{Z}^n$ whose distance to $t$ is at most $\mu^{( \| \cdot \| )} ( t , \mathcal{L} ( B ) )$ by showing that the distance of every lattice vector in $2^i \mathbbm{Z}^n \backslash \{  v \}$ is greater than $\mu^{( \| \cdot \| )} ( t , \mathcal{L} ( B ) ) \geq \mu^{(\| \cdot \|)} ( t , 2^i \mathbbm{Z}^n )$.\\
We consider a lattice vector $u \in 2^i \mathbbm{Z}^n \backslash \{ v \}$, together with its representation as a linear integer combination of the standard basis of the lattice $2^i \mathbbm{Z}^n$, $u = \sum_{j=1}^n u_j 2^i e_j$ with $u_j \in \mathbbm{Z}$, $1 \leq j \leq n$. Since $u \not= v$, there exists an index $k$, $1 \leq k \leq n$, where the coefficient $u_k$ is not the nearest integer of the coefficient $\beta_j$, i.e., $u_k \not= \lfloor \beta_k \rceil$.
Using the function $c$, we can show that this coefficient is responsible, that the distance between the target vector $t$ and this lattice vector is larger than $\mu^{(\| \cdot \|)} ( t , \mathcal{L} ( B ) )$:
\begin{align*}
\| u - t \|^2 \geq c^2 \cdot \| \sum_{j=1}^n ( u_j - \lfloor \beta_j \rceil ) 2^i e_j \|_2^2
 \geq c^2 \cdot | u_k - \lfloor \beta_k \rceil |^2 2^{2i} = \frac{c^2}{4} \cdot 2^{2i}.
\end{align*}
Since the value $i$ satisfies $i > 1 + \log_2 ( \mu^{( \| \cdot \| )} ( t , \mathcal{L} ( B ) ) ) - \log_2 ( c )$, we obtain
$2^{i-1} c > \mu^{(\| \cdot\|)} ( t , \mathcal{L} ( B ) )$,
which shows that $\| u -t \| > \mu^{( \| \cdot \| )} ( t , \mathcal{L} ( B ) )$.
\end{proof}

\noindent Now we are able to give a reduction from the closest vector problem to the optimization closest vector problem.
In the reduction, we will transform the given $\cvp$-instance into a new \cvp-instance, which 
satisfies the assumptions from Lemma \ref{lemma_solving_cvp_thin_lattice}. Additionally, both $\cvp$-instances will have the same distance between the target vector and the lattice. Hence, we are able to conclude from the solution of the new \cvp-instance to the solution of the original instance.

\begin{proof} (of Proposition \ref{prop_cvp_search_to_opt})\\
We are given a lattice basis $B \in \mathbbm{Z}^{n \times m}$ and a target vector $t \in \Span ( B )$. Using the algorithm $\mathcal{A}_{\opt}$ with input $B$ and $t$, we can compute
$$\mu := \mu^{(\| \cdot\|)} ( t , \mathcal{L} ( B ) ).$$
Without loss of generality, we can assume that $\mu \not= 0$, i.e., $t \not\in \mathcal{L} ( B )$.

\noindent Assume that we are able to construct a sequence of \cvp-instances
$$(B_i, t_i), ~ 0 \leq i \leq i_{\max} := \lceil \log_2 ( m \cdot \max_j \| b_j \| ) + 2 - \log_2 ( c ) \rceil,$$
where each tuple satisfies the following properties:
\begin{equation} \label{eq_condition_1}
B_i = 2^i B \subseteq 2^i \mathbbm{Z}^n, ~
t_i - t \in \mathcal{L} ( B ) \mbox{ and } \mu^{(\| \cdot\|)} ( t_i , \mathcal{L} ( B_i ) ) = \mu.
\end{equation}
Since the distance between the target vector $t$ and the lattice $\mathcal{L} ( B )$ is at most $m \cdot \max \{ \| b_j \| | 1 \leq j \leq m \}$,
each index $i \in \mathbbm{N}$ with $i \geq i_{\max}$ satisfies that
$i \geq \log_2 ( \mu ) + 1 - \log_2 ( c )$.
Hence, the $\cvp$-instance $(B_{i_{\max}}, t_{i_{\max}} )$ satisfies the assumptions of lemma \ref{lemma_solving_cvp_thin_lattice} and a closest lattice vector to $t_{i_{\max}}$ in the lattice $\mathcal{L} ( B_{i_{\max}} )$ can be found efficiently.
Using a solution of this $\cvp$-instance, we are able to compute a closest lattice vector to $t$:
If $x_{i_{\max}} \in \mathcal{L} ( B_{i_{\max}} )$ is a solution of this \cvp-instance, then the vector
$x := x_{i_{\max}} - ( t_{i_{\max}} - t ) \in \mathcal{L} ( B )$
is a solution of the \cvp-instance $(B,t)$.\\

\noindent Hence, it remains to show how to construct a sequence of \cvp-instances $(B_i,t_i)$ satisfying the properties stated in (\ref{eq_condition_1}):
As initialization, we set
$B_0 := B$ and $t_0 := t$.
Then, we continue inductively. For simplicity, we describe the construction only for $(B_1,t_1)$.\\
The basis $B_1 = 2 B_0$ is constructed in $m$ steps and in each step we construct a $\cvp$-instance $( \tilde{B}_j , \tilde{t}_j )$, $0 \leq j \leq m$ such that $(\tilde{B}_0 , \tilde{t}_0 ) = ( B_0 , t_0 )$ and $( \tilde{B}_m , \tilde{t}_m ) = ( B_1 , t_1 )$.\\
The construction is done in that way that each constructed instance satisfies the stated properties: We have $\tilde{t}_j - t \in \mathcal{L} ( B )$ and $\mu^{( \| \cdot \| )} ( \tilde{t}_j , \mathcal{L} ( \tilde{B}_j ) ) = \mu = \mu^{( \| \cdot \| )} ( t  ,\mathcal{L} ( B ) )$ for all $0 \leq j \leq m$.\\
Each lattice vector is a linear integer combination of the basis vectors. The idea of the construction is to fix in each step on e basis vector $b_j$, $1 \leq j \leq m$, and to check whether there exists a closest lattice vector to $t$, whose representation uses the vector $b_j$ an even number of times.
The closest lattice vector to $t$ is a linear integer combination of the basis vectors.
In each step, we fix one basis vector $b_j$, $1 \leq j \leq m$, and check whether the above representation uses this basis vector $b_j$ an odd or an even number of times.\\
This is done as follows: We consider the lattice which consists of all lattice vectors of the original lattice, which have a basis representation which uses the vector $b_j$ an even number of times.
If the distance of the target vector to this lattice is the same as its distance to the original lattice, this is the case and we do not change the target vector.
This can be checked using the algorithm $\mathcal{A}_{\opt}$.\\
In the other case, we construct a new target vector by $t - b_j$.
Obviously, this new target vector has the same distance to the new lattice as the original target vector to the original lattice.\\
But we need to be aware of the following: It is not possible to make the decisions described above independently:
If there exists several lattice vectors which are closest to the target vector, then in general they have a different representation as a linear combination of the basis vectors. Here different is meant with respect to the parity of the coefficients.
Hence, the construction need to be done sequentially.
For a detailed description and an illustration of the construction see Figure \ref{figure_cvp_selfreduction}.\\

\begin{figure}[t]
\begin{center}
\framebox{
\begin{minipage}[b]{13.25cm}
\tt \small
\emph{Construction:}\\
 Input: $\cvp$-instance $(B_0 , t_0 )$\\
 Set $\tilde{B}_0 := B_0$ and $\tilde{t}_0 := t_0$.\\
 For $1 \leq j \leq m$:
		\begin{itemize}
 			\item Start $\mathcal{A}_{\opt}$ with input $(\tilde{B}_j, \tilde{t}_{j-1}$),
 			where $\tilde{B}_j = [ 2 b_1, \hdots, 2 b_j, b_{j+1} , \hdots, b_m ]$.\\
 			The algorithm computes $\mu ( \tilde{B}_j , \tilde{t}_{j-1} )$.
 \item If $\mu ( \tilde{B}_j , \tilde{t}_{j-1} ) = \mu$, then $\tilde{t}_{j+1} := \tilde{t}_j$.\\
 			Otherwise $\tilde{t}_{j+1} := \tilde{t}_j - b_j$. 
\end{itemize}
Output: $\cvp$-instance $(\tilde{B}_m , \tilde{t}_m )$.
\end{minipage}
}
\end{center}
\caption{Construction of a new $\cvp$-instance}
\label{figure_cvp_selfreduction}
\end{figure}

\noindent It is easy to see that the $\cvp$-instance $(B_1,t_1)$ satisfies the properties stated in (\ref{eq_condition_1}).
This proves the correctness of the construction and at the same time the correctness of the algorithm for the closest vector problem.
It remains to show that the algorithm has the claimed running time.\\

\noindent As described, for the construction of the instance $(B_i,t_i)$ from the instance $(B_{i-1},t_{i-1})$ we need $m$ calls to the algorithm $\mathcal{A}_{\opt}$. Hence, the total number of calls to the algorithm $\mathcal{A}_{\opt}$ is
$m \cdot i_{\max},$
where $i_{\max}$ depends on the length of the basis vectors of $B$.
Using the knowledge of the function $C$, we obtain that
\begin{equation} \label{eq_repsize_upper_bound}
\max_{1 \leq j \leq m} \| b_j \| \leq C \cdot \max_{1 \leq j \leq m} \| b_j \|_2 \leq C \cdot \sqrt{n} \cdot S
\end{equation}
using Lemma \ref{facts_representation_size}.
Hence, we get the following upper bound for the number of calls to the algorithm $\mathcal{A}_{\opt}$,
\begin{align*}
 m \cdot
\left(
	\log_2 ( m \sqrt{n} \cdot C \cdot S ) + 2 - \log_2 ( c )
\right)
\leq ~ 
2m \cdot \log_2 ( m \sqrt{n}  ( C \cdot c^{-1} ) \cdot S ).
\end{align*}
Finally, we need to care about the magnitude of the representation size of the instances:
We apply the algorithm $\mathcal{A}_{\opt}$ to lattice bases $\tilde{B} \in \mathbbm{Z}^{n \times m}$, where each basis vector is the original basis vector multiplied with a factor $2^i$, where $i \leq i_{\max}$. Hence,
$$\size ( \tilde{B} ) \leq \size ( 2^{i_{\max}} \size ( B ) ) \leq 2^{i_{\max}} \cdot S.$$
The corresponding target vector $\tilde{t}$ is of the form $t - v$, where $v$ is a summand of at most $m \cdot i_{\max}$ basis vectors. Hence, if $b$ is the basis vector of $B$ with $\size ( b ) = \size ( B )$, then
\begin{equation*}
\size( \tilde{t} )
\leq \size ( t + \sum_{j=1}^{m \cdot i_{\max}} 2^{i_{\max}} b )
= \size ( t + m \cdot i_{\max} 2^{i_{\max}} b ).
\end{equation*}
Since all vectors are integer vectors, we obtain
\begin{equation*}
\size ( \tilde{t} ) 
 \leq \size ( t ) + \size ( m \cdot i_{\max} \cdot 2^{i_{\max}} b )
 \leq \size ( t ) + m \cdot i_{\max} \cdot 2^{i_{\max}} \size ( B ).
 \end{equation*}
 The parameter $S$ is an upper bound on the representation size of the basis $B$ and the vector $t$. Hence, we have
 \begin{equation*}
\size ( \tilde{t} ) \leq 2 m \cdot i_{\max} \cdot 2^{i_{\max}} \cdot S
\end{equation*}
and the size of each instance is at most $2m \cdot i_{\max} 2^{i_{\max}} S \leq m \cdot 2^{2 i_{\max}} S$.
Using the definition of $i_{\max}$, this is upper bounded by
\begin{equation*}
m \cdot 2^{2 \cdot \left( \log_2 ( m \cdot \max_j \| b_j \| ) + 2 - \log_2 ( c ) \right)} \cdot S
\leq m \cdot ( m \cdot \max_j \| b_j \| )^2 \cdot 2^4 \cdot c^{-2} \cdot S
= 16 m^3 \max_j \| b_j \|^2 \cdot c^{-2} \cdot S.
\end{equation*}
Using the upper bound  (\ref{eq_repsize_upper_bound}) for the length of the basis vectors, this is at most
\begin{equation*}
16 m^3 \cdot C^2 \cdot n \cdot S^2 \cdot c^{-2} \cdot S
\leq 16 \cdot m^3 \cdot n \cdot ( C \cdot c^{-1} )^2 \cdot S^3.
\end{equation*}
Hence, the running time of the algorithm to solve $\cvp$-Search is at most
\begin{equation*}
2m \cdot \left(
	\log_2 ( m \sqrt{n} \cdot ( C \cdot c^{-1} ) \cdot S ) 
\right)
 \cdot T_{m,n}^{(\|.\|)} ( 16 m^3 n ( C \cdot c^{-1} )^2 S^3 ).
 \end{equation*}
 \end{proof}
 
\begin{thm} \label{state_selfreduction_cvp_general}
Let $\| \cdot \|$ be a $(k,K)$-enumerable norm on $\mathbbm{R}^n$ and $c, C : \mathbbm{N} \to \mathbbm{R}^{>0}$ be non-decreasing functions such that
$c(n) \cdot \| x \|_2 \leq \| x \| \leq C ( n ) \cdot \| x \|_2$
for all $x \in \mathbbm{R}^n$.\\
Assume that there exists an algorithm $\mathcal{A}_{\dec}$ that for all lattices $\mathcal{L} ( B' ) \subseteq \mathbbm{Z}^n$ of rank $m$, all target vectors $t ' \in \Span ( B' ) \cap \mathbbm{Z}^n$ and all $r > 0$ solves the decisional closest vector problem in time $T_{m,n}^{(\| \cdot \|)} ( S' , r )$, where $S'$ is an upper bound on the size of the basis $B'$ and the target vector $t'$.\\
Then, there exists an algorithm $\mathcal{A}'$, that solves the closest vector problem for all lattices $\mathcal{L} ( B ) \subseteq \mathbbm{Z}^n$ of rank $m$ and target vectors $t \in \Span ( B ) \cap \mathbbm{Z}^n$ in time
$$n^{\mathcal{O} ( 1 )} \log_2 ( ( C \cdot c^{-1} ) S )
\cdot \left( k \cdot \log_2 ( \max_j \| b_j \| ) + \log_2 ( K ) \right)
\cdot T ( 16 m^3 n ( C \cdot c^{-1} ) S^3 , m \cdot \max_j \|b_j \| ),$$
where $S$ is an upper bound on the size of the basis $B$ and the target vector $t$.
\end{thm}

\noindent The corresponding result for all $\ell_p$-norms follows directly from a special case of Hölder's inequality, which we stated on page \pageref{eq_Hölder}. It provides also a proof for Theorem \ref{statement_cvp_self_lp} in the case of $\ell_p$-norms.

\begin{cor} \label{appendix_statement_cvp_self_lp}
For all $\ell_p$-norms, $1 \leq p \leq \infty$, assume that there exists an algorithm $\mathcal{A}_{\dec}$ that for all lattices $\mathcal{L} ( B' ) \subseteq \mathbbm{Z}^n$ of rank $m$, all target vectors $t' \in \Span ( B' ) \cap \mathbbm{Z}^n$ and all $r > 0$ solves the decisional closest vector problem in time $T_{m,n}^{(\| \cdot \|)} ( S' , r )$, where $S'$ is an upper bound on the size of the basis $B'$ and the target vector $t'$.\\
Then, there exists an algorithm $\mathcal{A}'$, that solves the closest vector problem for all lattices $\mathcal{L} ( B ) \subseteq \mathbbm{Z}^n$ in time
$$k \cdot n^{\mathcal{O} ( 1 )} \log_2 ( S )^2 T ( 16 m^3 n^2 S^3 , m \cdot n S ),$$
where $k = p$ for $1 \leq p < \infty$ and $k = 1$ for $p = \infty$.
Here, $S$ is an upper bound on the size of the basis $B$ and the target vector $t$.
\end{cor}

\begin{proof}
Hence, we can apply Theorem \ref{state_selfreduction_cvp_general} with parameter $c, C$ such that $C ( n ) \cdot c (n )^{-1} \leq n$.
Additionally, every $\ell_p$-norm is $(k,1)$-enumerable with $k = p$ for $1 \leq p < \infty$ and $k = 1$ for $p = \infty$.
Hence, we obtain that there exists an algorithm for the closest vector problem whose running time is at most
$$n^{\mathcal{O} ( 1 )} \log_2 ( n S ) ( k \cdot \log_2 ( \max_j \| b_j \| ) \cdot T_{m,n}^{(p)} ( 16m^3 n^2 S^3 , m \cdot \max_j \| b_j \| ).$$
Since the length of all basis vectors $b_i$, $1 \leq i \leq n$, is upper bounded by
$\| b_i \|_p \leq \sqrt[p]{n} \size ( b_i ) \leq \sqrt[p]{n} \size ( B ) \leq n \cdot S$
for $1 \leq p < \infty$ and $\| b_i \|_{\infty} \leq \size ( B ) \leq S$,
see for example \cite{bk_gls}, we obtain the claimed result.
\end{proof}

\noindent To get the corresponding result for polyhedral norms, we need to show that all polyhedral norms are enumerable. This is done in the following lemma.

\begin{lemma} \label{lemma_polyhedral_norms_enumerable}
Let $P \subset \mathbbm{R}^n$ be a full-dimensional polytope symmetric about the origin with $F$ facets.
Let $P$ be given by a set $H_P = \{ h_1, \hdots, h_{F/2} ) \subset \mathbbm{Z}^n$ and a set of parameters $\{ \beta_1, \hdots, \beta_{F/2} \} \subset \mathbbm{N}$, i.e.,
$$P = \{ x \in \mathbbm{R}^n | \langle x , h_i \rangle \leq \beta_i \mbox{ and } \langle x , -h_i \rangle \leq \beta_i \mbox{ for all } 1 \leq i \leq F / 2 \}.$$ 
Then $\| \cdot \|_P$ is a $(1 , \prod_{j=1}^{F/2} \beta_j )$-enumerable norm.
\end{lemma}

\begin{proof}
Given an integer vector $x \in \mathbbm{Z}^n \backslash \{ 0 \}$, its polyhedral norm has value $r$ if the following two properties are satisfied:
\begin{itemize}
	\item The vector $x$ is contained in the scaled polytope $r \cdot P$, that means $\langle x , h_i \rangle \leq r \cdot \beta_i$ and $\langle x , -h_i \rangle \leq \beta_i$ for all $1 \leq i \leq F/2$.
	\item There exists at least one inequality defining the polytope which is satisfied with equality.
		Let $j \in \mathbbm{N}$, $1 \leq j \leq F/2$, be such an index.
		Without loss of generality, we assume that $\langle x , h_j \rangle = r \cdot \beta_j$.
		Since $\langle x , h_j \rangle \in \mathbbm{Z}$, we have $r = \langle x , h_j \rangle / \beta_j \in \mathbbm{Q}$.
		That means, there exists $p,q\in \mathbbm{N}$ with $\gcd( p , q ) = 1$ such that $r = p / q$.
		Additionally, we know that $\beta_j$ is divisible by $q$.
\end{itemize}
That means, that each value, which can be achieved by the norm $\| \cdot \|_P$ of an integer vector, is a rational of the form $p / q$ with $p,q \in \mathbbm{N}$, $\gcd ( p , q ) = 1$ and there exists an index $j$, $1 \leq j \leq F/2$, such that $q$ divides $\beta_j$.
Hence, for each vector $x \in \mathbbm{Z}^n$, we obtain that $(\prod_{j=1}^{F/2} \beta_j ) \cdot \| x \|_P \in \mathbbm{N}_0$.
\end{proof}

\noindent Additionally, we need to compute the radius of an in- and circumscribed Euclidean ball.

\begin{lemma} \label{statement_polyhedron_innervolume}
Let $P \subset \mathbbm{R}^n$ be a full-dimensional polytope symmetric about the origin,
$$P = \{ x \in \mathbbm{R}^n | \langle x , h_i \rangle \leq 1 \mbox{ and } \langle x , - h_i \rangle \leq 1 \mbox{ for all } 1 \leq i \leq F/2\}.$$
Define $P$ contains an Euclidean ball with radius 
$h: = \min \{ \| h_i \|_2^{-1} | 1 \leq i \leq F / 2 \}$
centered at the origin.
\end{lemma}

\begin{proof}
Let $x \in B_n^{(2)} ( 0 , h ) = h \cdot B_m^{(2)} (0, 1 )$. Then this vector is of the form
$x = h \cdot x'$,
where $x' \in \mathbbm{R}^n$ with $\| x' \|_2 \leq 1$.
Using the Cauchy-Schwarz inequality, it is obvious that this vector $x$ satisfies all inequalities defining the polytope.
\end{proof}

\begin{lemma} \label{lemma_polyhedron_approx_ellipsoid_initialization}
Let $P \subseteq \mathbbm{R}^n$ be a full-dimensional polytope given by a vector $A \in \mathbbm{Z}^{m \times n}$ and a vector $\beta \in \mathbbm{Z}^m$.
Let $r$ be an upper bound on the representation size of $P$.\\
Then $P$ ist contained in an Euclidean ball with radius  
$R_{out} = \sqrt{n} \left( n^{n/2} r^n \right)$
centered at the origin.
\end{lemma}

\noindent For a proof of this statement see Lemma 3.1.33 in \cite{bk_gls}.

\begin{cor} \label{cor_cvp_search_to_opt_polyhedron}
Let $P \subseteq \mathbbm{R}^n$ be a full-dimensional polytope symmetric about the origin with $F$ facets.
Assume that there exists an algorithm $\mathcal{A}_{\dec}$ that for all lattices $\mathcal{L} ( B' ) \subset \mathbbm{Z}^n$ of rank $m$ and all target vectors $t \in \Span ( B') \cap \mathbbm{Z}^n$ solves the decisional closest vector problem with respect to the polyhedral norm $\| \cdot \|_P$ in time $T_{m,n}^{(P)} ( S' )$,
where $S'$ is an upper bound on the size of the basis $B'$ and the target vector $t'$.\\
Then there exists an algorithm $\mathcal{A}'$ that solves the closest vector problem with respect to the polyhedral norm $\| \cdot \|_P$ for all lattices $\mathcal{L} ( B ) \subset \mathbbm{Z}^n$ of rank $m$ and target vectors $t \in \Span ( B ) \cap \mathbbm{Z}^n$ in time
\begin{equation*}
F \cdot n^{\mathcal{O} ( 1 )} \log_2 ( P \cdot S ) 
\cdot T_{m,n}^{(P)} ( 16 \cdot m^3 n^{n+2}\size ( P )^{n+1} 
\cdot S^3 , n \cdot m \cdot S \cdot \size ( P )),
\end{equation*}
where $S$ is an upper bound on the size of the basis $B$ and the target vector $t$.
\end{cor}

\noindent This corollary provides also a proof of Theorem \ref{statement_cvp_self_polyhedron} in the case of polytopes.

\begin{proof}
Assume that $P$ is given by a set $H_P = \{ h_1, \hdots, h_{F/2} \} \subset \mathbbm{Z}^n$ and a set of parameters $\{ \beta_1, \hdots, \beta_{F/2} \} \subset \mathbbm{N}$, i.e.,
$$P = \{ x \in \mathbbm{R}^n | \langle x , h_i \rangle \leq \beta_i \mbox{ and } \langle x , -h_i \rangle \leq \beta_i \mbox{ for all } 1 \leq i \leq F/2 \}.$$
As we have seen in Lemma \ref{statement_polyhedron_innervolume}, $P$ contains an Euclidean ball with radius $\min_i \{ 1 / \| h_i \|_2 \}$. 
The radius $\min \{ 1 / \| h_i \|_2 | 1 \leq i \leq F / 2 \}$ is at least $\sqrt{n} / \size ( P )$, since we have
$$\| h_i \|_2 \leq \sqrt{n} \size (h_i ) \leq \sqrt{n} \size ( P )$$
for all $1 \leq i \leq F / 2$, using the result from Lemma \ref{facts_representation_size}.
That means, we have $P \subset B_n^{(2)} (0 , (\sqrt{n} \cdot \size ( P ))^{-1})$.\\
Additionally, we have seen in Lemma \ref{lemma_polyhedron_approx_ellipsoid_initialization} that $P$ is contained in a ball with radius $\sqrt{n} \cdot  n^{n/2} \cdot \size ( P )^n$.
Using these results, the relation between the in- and circumscribed unit ball is at most
$$\frac{\sqrt{n} \cdot n^{n/2} \size( P )^n}{(\sqrt{n} \cdot \size ( P ))^{-1}}
= n^{n/2 + 1} \size ( P )^{n+1}
\leq ( n \cdot \size ( P ) )^{n+1}.$$
Now, it follows from Theorem \ref{state_selfreduction_cvp_general}, that there exists an algorithm $\mathcal{A}'$, that solves the closest vector problem with respect to the norm defined by the polytope $P$.\\
Additionally, we have seen In Lemma \ref{lemma_polyhedral_norms_enumerable} that the norm $\| \cdot \|_P$ defined by the polytope $P$ is $( 1, \prod_{j=1}^{F/2} \beta_j )$-enumerable.
Since the parameters $\beta_j$, $1 \leq j \leq F/2$, are integers, we have
$$\prod_{j=1}^{F/2} \beta_j \leq \size ( P )^{F/2}.$$
with Lemma \ref{facts_representation_size}, we see that the length of each basis vector $b_i$, $1 \leq i \leq m$, with respect to the norm $\| \cdot \|_P$ is at most
$$\| b_i \|_P \leq n \cdot \size ( B ) \cdot \size ( P ).$$
Hence, we obtain that
$$\max \{ \| b_j \|_P | 1 \leq j \leq m \} \leq n \cdot S \cdot \size ( P ).$$
Now, it follows from Theorem \ref{state_selfreduction_cvp_general}, that the running time of the algorithm $\mathcal{A}'$ is at most
\begin{equation*}
F \cdot n^{\mathcal{O} ( 1 )} \log_2 ( P \cdot S ) 
\cdot T_{m,n}^{(P)} ( 16 \cdot m^3 n^{n+2}\size ( P )^{n+1} 
\cdot S^3 , n \cdot m \cdot S \cdot \size ( P )),
\end{equation*}
\end{proof}

	\subsection{Technical lemmata for the lattice membership algorithm}

\begin{lemma} \label{lemma_polyhedron_out_box}
Let $P \subseteq \mathbbm{R}^n$ be a full-dimensional bounded polyhedron given by $m$ integral inequalities $\langle a_i , x \rangle \leq \beta_i$ where $a_i \in \mathbbm{Z}^n$, $\beta_i \in \mathbbm{Z}$ for $1 \leq i \leq m$, i.e.,
$$P = \{ x \in \mathbbm{R}^n | \langle a_i , x \rangle \leq \beta_i \mbox{ for } 1 \leq i \leq m \} = \{ x \in \mathbbm{R}^n | A^T x \leq \beta \},$$
where $A$ is the matrix which contains of the columns $a_i$. Then
$$P \subset \{ x \in \mathbbm{R}^n | - t \leq x ( i ) \leq t \} \mbox{ with } t = n^{n/2} r^n,$$
where $r$ is the representation size of the polyhedron.
\end{lemma}

\begin{proof}
Let $v \in P$ be an arbitrary vertex of the polyhedron. Then there exists a $n \times n$ submatrix $C$ of $A^T$ such that $C \cdot v = d$, where $d$ is the column vector which consists of the corresponding coefficients of $b$. Using Cramer's Rule, the coefficients $v_i$ of the vertex $v$ are given by
$$v_i = \frac{\det ( C_i )}{\det ( C )}.$$
Here $C_i$ is the $n \times n$ matrix $C$ where the $i$-th column is replaced by $d$. Since $A^T$ is a matrix with integral coefficients, $|\det ( C ) | \geq 1$ and we get for all coefficients
$$| v_i | \leq | \det ( C_i )| \leq n^{n/2} \size ( C )^n,$$
where the last inequality can be shown using Hadamard's inequality.
This proves the lemma.
\end{proof}

\subsection{Technical lemmata for the flatness algorithm} \label{sec_appendix_technical_flatness}

\begin{thm} (Theorem \ref{thm_flatness_ellipsoid} restated)\\
Let $E \subset \mathbbm{R}^n$ be an ellipsoid.
If the width of the ellipsoid is at least $n$,
$w ( E ) \geq n$,
then the ellipsoid contains an integer vector.\end{thm}

\begin{proof}
We prove the contraposition:
If the ellipsoid $E ( D , c )$ does not contain an integer vector, then every integer vector $x \in \mathbbm{Z}^m$ satisfies
$(x-c)^T D^{-1} (x-c) > 1$.
Since
$$( x - c )^T D^{-1} ( x - c )
= ( x - c )^T ( Q^T Q )^{-1} ( x - c )
= \| ( Q^T )^{-1} x - ( Q^T )^{-1} c \|_2^2,$$
the distance from the vector $(Q^T)^{-1} c$ to the lattice $\mathcal{L} ( (Q^T)^{-1})$ is greater than 1.
This implies that the covering radius of the lattice $\mathcal{L} ( (Q^T )^{-1} )$ is greater than 1,
$\mathcal{L} ( ( Q^T )^{-1})$.
Since $\mathcal{L} (  ( Q^T )^{-1} ) = \mathcal{L} ( Q )^*$, we obtain from the transference bound due to Banaszczyk \cite{pp_banaszczykbounds} that
$$\frac n 2 \geq \mu^{(2)} ( \mathcal{L} ( Q )^* ) \cdot \lambda_1^{(2)} ( \mathcal{L} ( Q ) ) > \lambda_1^{(2)} ( \mathcal{L} ( Q ) ).$$
Since we have seen in Proposition \ref{prop_flatness_direction} that the width of the ellipsoid is exactly $2 \lambda_1^{(2)} ( \mathcal{L} ( Q ) )$, it follows that
$w ( E ( D , c ) ) < n$.
\end{proof}

\begin{claim} (Claim \ref{claim_upper_bound_size} restated)
Let $F : \mathbbm{R}^m \to \mathbbm{R}$ be a function defined as in (\ref{eq_def_function_lower_bound_volume}) given by a non-singular matrix $V \in \mathbbm{Q}^{n \times n}$, a vector $t \in \mathbbm{Q}^n$ and $\alpha_n, \alpha_d \in \mathbbm{N}$.
Let $S$ be an upper bound on the representation size of $V^{-1}$, $t$, $\alpha_n$ and $\alpha_d$.
Then, there exists an integer $K \leq S^{2n^2 p}$ such that
$K \cdot F ( x ) \in \mathbbm{Z}$ for all $x \in \mathbbm{Z}^m$.
\end{claim}

\begin{proof}
Since $\alpha_n, \alpha_d \in \mathbbm{N}$, we observe that $F ( x ) \in \mathbbm{Z}$ if all coefficients of the matrix $V^{-1}$ and the vector $t$ are integers.
If $V^{-1} = ( v_{ij} ) \in \mathbbm{Q}^{n \times n}$ and $t = ( t_i ) \in \mathbbm{Q}^n$, then the coefficients of the vector $V^{-1} t $ are rationals of the form $\sum_{j=1}^n v_{ij} t_j$.
That means, each coefficient is the sum of $n$ rational numbers whose denominators are at most $S^2$.\\
Hence, the multiplication of this vector with the product of these denominators yields an integer vector.
The multiplication of $V^{-1}$ with the same number yields an integer matrix.\\
Hence, there exists a number, which is at most
$(S^2)^{n^2} = S^{2n^2}$
such that $V^{-1} ( ( x^T , 0^{n-m} )^T - t )$ becomes an integer if multiplied with this number.
Since $F$ consists of the $p$-th power of an $\ell_p$-norm, there exists a number which is at most
$(S^{2n^2})^p = S^{2n^2p}$
such that $F ( x )$ becomes an integer if multiplied with this number.
\end{proof}

\begin{lemma} (Lemma \ref{lemma_subgradient_l_p_potenziert} restated)
Let $y \in \mathbbm{R}^n$ and $1 < p < \infty$.
Then a subgradient $g \in \mathbbm{R}^n$ of the function
\begin{equation*}
F_p : \mathbbm{R}^n \to \mathbbm{R}, ~
x  \mapsto \sum_{i=1}^n | x_i |^p
\end{equation*}
at the vector $y$ is given by $g = (g_1, \hdots, g_n)^T$, where
$$g_i := \sign (y_i) \cdot | y_i |^{p-1}.$$
\end{lemma}

\begin{proof}
Since $F_p$ is a nonnegative combination of the functions $x \mapsto | x_i |^p$, it is enogugh to consider the case, where $n = 1$.\\
We will show that the vector $g$ defined as above satisfies the subgradient inequality (\ref{eq_subgradient_inequality}).
For all $z \in \mathbbm{R}$ and $0 < \lambda \leq 1$ it follows from the convexity of the function $F_p$ that
$$F_p ( y + \lambda ( z - y ) ) \leq ( 1 - \lambda ) F_p ( y ) + \lambda F_p ( z )$$
or
$$F_p ( z ) \geq \frac{1}{\lambda} \left( F_p ( y + \lambda ( z -y ) ) - ( 1 - \lambda ) F_p ( y ) \right)
= F_p ( y ) + \frac{1}{\lambda} \left( F_p ( y + \lambda ( z - y ) ) - F_p ( y ) \right).$$
Hence, it remains to show that
$$F_p ( y + \lambda ( z - y ) ) - F_p ( y )
\geq \lambda \sign ( y ) \cdot | y |^{p-1} ( z - y ).$$
By definition of $F_p$, we have 
$F_p ( y + \lambda ( z - y ) ) - F_p ( y ) = | y + \lambda ( z - y )|^p - |y|^p$.
Since for all $a,b \in \mathbbm{R}$, $m \in \mathbbm{N}$, it holds that
$b^m - a^m = ( b - a ) \cdot \sum_{i=0}^{m-1} b^{m-1-i} a^i$, 
we see that
\begin{align*}
| y + \lambda ( z - y ) |^p - | y |^p
& = \left( | y + \lambda ( z - y ) | - | y | \right) \cdot \sum_{i=0}^{p-1} | y + \lambda ( z - y ) |^{p-1-i} \cdot | y |^i\\
& \geq \left( | y + \lambda ( z - y ) | - | y | \right) | y |^{p-1}.
\end{align*}
Since for all $a,b \in \mathbbm{R}$,
$| a |  - | b | \geq \sign ( b ) \cdot ( a - b )$,
this is at least
$\lambda \cdot \sign ( y ) ( z -y ) \cdot | y |^{p-1}$.
\end{proof}

\begin{lemma} (Lemma \ref{statement_upper_bound_coordinate_subgradient} restated)
For $ m , n \in \mathbbm{N}$, $m \leq n$, a subgradient at the vector $y \in \mathbbm{R}^m$ of the function
$F : \mathbbm{R}^m \to \mathbbm{R}$,
$x \mapsto \alpha_d^p \| V^{-1} ( ( x^T , 0^{n-m} )^T - t ) \|_p^p - \alpha_n^p$,
where $V \in \mathbbm{R}^{n \times n}$ is nonsingular, $t \in \mathbbm{R}^n$ and $1 < p < \infty$,
is given by the vector $\alpha_d^p g \in \mathbbm{R}^m$ defined by
$g = ( V^{-1})^T \bar{g}_{\{1, \hdots, m \}}$,
where $\bar{g} \in \mathbbm{R}^n$ is defined by 
$\bar{g}_{i}
= \sign ( [V^{-1} ( y - t ) ]_i ) \cdot | [V^{-1} ( y - t ) ]_i |^p$.
If $y \in \bar{B}_m^{(2)} ( 0 , R ) \subseteq \mathbbm{R}^m$, then 
$\| \alpha_d^p g \|_2 \leq m \cdot \left( \alpha_d n S^2 R \right)^{p+1}$,
where $S$ is an upper bound on the size of $V^{-1}$ and $t$.
\end{lemma}

\begin{proof}
The correctness of the construction follows directly from Lemma \ref{lemma_subgradient_Fp}.
Since
$\| g \|_2 \leq m \cdot \max \{ | g_i | | 1 \leq i \leq m \}$, 
it is enough to compute an upper bound on the coefficient of the vector $g$.\\
If $V^{-1} = ( v_{ij} )_{i,j} \in \mathbbm{Q}^{n \times n}$ and $t = ( t_i )_i \in \mathbbm{Q}^n$, the $k$-th coefficient, $1 \leq k \leq n$, of the vector $V^{-1} ( y - t )$ is given by
$$| [ V^{-1} ( y - t ) ]_k | \leq \sum_{j=1}^n | v_{kj} \cdot ( y_j - t_j ) |.$$
Since the coefficients of $V^{-1}$ and $t$ are at most $S$ and since each coefficient of $y$ is at most $R$ (in absolute values), we obtain
$$| [ V^{-1} ( y - t ) ]_k | \leq n \cdot S ( R + S ) \leq n R S^2.$$
Hence, each coefficient of the vector $\bar{g}$ is at most
$$| g_i | \leq ( n R S^2 )^p.$$
With the same argumentation, we obtain that each coefficient of the vector $g$ is at most
$$| g | \leq n \cdot S ( n R S^2 )^p \leq ( n S^2 R )^{p+1}.$$
\end{proof}


\begin{lemma} \label{lemma_appendix_change_repsize_mem_alg}
Let $D \in \mathbbm{Q}^{n \times n}$ be a symmetric positive definite matrix.
Let $\tilde{d} \in \mathbbm{Z}^n$ be the flatness direction of the ellipsoid defined by the matrix $D$.
Then 
$$\| \tilde{d} \|_2 \leq n^{(n+2)/2} \cdot \size ( D )^{(n+1)/2}.$$
\end{lemma}

\noindent In the proof of this lemma, we use that for every symmetric positive definite matrix $A$, there exists a uniquely determined symmetric positive definite matrix $X$ such that
$A = X^T \cdot X = X \cdot X$.
We call $X$ the square root of $A$, denoted by $A^{1/2}$ (see \cite{bk_matrixanalysis}).

\begin{proof}
To prove an upper bound on the length of the vector $\tilde{d}$,
we observe that $\tilde{d} = Q^{-1} v$, where $v$ is a shortest non-zero lattice vector in $\mathcal{L} ( Q )$ and that the length of $d$ is the same as the length of a shortest vector in the lattice $\mathcal{L} ( D^{1/2} )$,
$$\lambda_1^{(2)} ( \mathcal{L} ( Q ) ) = \lambda_1^{(2)} ( D^{1/2} ),$$
as we have seen in Proposition \ref{prop_flatness_direction}.
Especially, the length of the vector $\tilde{d} \in \mathbbm{Z}^n$ is at most
\begin{equation} \label{eq_change_repsize_1}
\| \tilde{d} \|_2 = \| Q^{-1} d \|_2
\leq \| Q^{-1} \| \cdot \| d \|_2
= \| Q^{-1} \| \cdot \lambda_1^{(2)} ( D^{1/2} ).
\end{equation}
Using Minkowski's Theorem, see for example \cite{bk_cassel}, the minimum distance of the lattice $\mathcal{L} ( D^{1/2} )$ is at most
\begin{equation} \label{eq_change_repsize_2}
\lambda_1^{(2)} ( D^{1/2} ) \leq \sqrt{n} \det ( D^{1/2} )^{1/n}
= \sqrt{n} \det ( D )^{1/2n}.
\end{equation}
Since the decomposition of a symmetric positive definite matrix in $D = Q^T Q$ is unique up to multiplication with an orthogonal matrix, there exists an orthogonal matrix $O \in \mathbbm{R}^{n \times n}$ such that $O \cdot Q = D^{1/2}$.
From this, one can show that the matrices $Q^{-1} = D^{-1/2} \cdot O$ and $D^{-1/2}$ have the same spectral norm:
$$\| Q^{-1} \| = \sqrt{\eta_n ( O^T D^{-1} O)}
= \sqrt{\eta_n ( D^{-1} )}
= \sqrt{\eta_n ( ( D^{-1/2} )^T D^{-1/2} )}
= \| D^{-1/2} \|,$$
where $\eta_n$ denotes the larget eigenvalue of the matrix.
Combining this with (\ref{eq_change_repsize_1}) and (\ref{eq_change_repsize_2}), we obtain the following upper bound for the length of the vector $\tilde{d}$,
$$\| \tilde{d} \|_2 \leq \sqrt{n} \| D^{-1/2} \| \cdot \det ( D )^{1/2n}.$$

\noindent The spectral norm of the matrix $\| D^{-1/2} \|$ is given by square root of the spectral norm of $D^{-1}$,
$$\| D^{-1/2} \|
= \sqrt{\eta_n ( D^{-1} )}
= \| D^{-1} \|^{1/2},$$
where the spectral norm of $D^{-1}$ is the inverse of an eigenvalue of $D$.
It is easy to see, that each eigenvalue of the symmetric positive definite matrix is at least $1 / \size ( D )$, see for example \cite{pp_ye}.
Hence, we obtain that
$$\| D^{-1} \|_2^{1/2} \leq \size ( D )^{1/2}$$
and respectively the following upper bound for the length of $\tilde{d}$,
$$\| \tilde{d} \|_2
\leq \sqrt{n} \det ( D )^{\frac 1 2 ( 1 + \frac 1 n)}.$$

\noindent The determinant $\det ( D )$ is the product of its eigenvalues (see \cite{bk_strang}) and the size of each eigenvalue of the symmetric positive definite matrix is at most $n \cdot \size ( D )$, see again \cite{pp_ye}.
We obtain that
$$\det ( D ) \leq (  n \cdot \size ( D ) )^{n}.$$
Hence, the length of the vector $\tilde{d}$ is at most
$$\| \tilde{d} \|_2 \leq \sqrt{n} ( n \cdot \size ( D ) )^{\frac n 2 ( 1 + \frac 1 n )}
= \sqrt{n} ( n \cdot \size ( D ) )^{(n+1)/2}.$$
\end{proof}

\end{document}